\documentclass[twoside,11pt]{article}

\usepackage{blindtext}
\usepackage{amsmath}
\usepackage{graphicx}
\usepackage{enumerate}
\usepackage{natbib}
\usepackage{url} 
\usepackage{amsmath,amssymb,amsthm,xcolor,mathrsfs}
\usepackage{setspace}
\usepackage{caption}
\usepackage{bbm, bm}
\usepackage{algorithm,algpseudocode}
\usepackage{url}
\usepackage{natbib}
\usepackage{varwidth}
\usepackage{multirow}
\usepackage{cancel}
\usepackage[normalem]{ulem}
\usepackage{subfig, float}
\usepackage{rotating}
\usepackage{array}
\allowdisplaybreaks
\usepackage{fancyvrb}
\usepackage{verbatim}
\usepackage{jmlr2e}

\newtheorem{Theorem}{Theorem}
\newtheorem{condition}{Condition}[section]

\newcommand{\beq}{\begin{equation}}
	\newcommand{\eeq}{\end{equation}}

\newcommand{\Pro}{\mathbb{P}}

\newcommand{\bu}{{\mathbf u}}
\newcommand{\bSigma}{\boldsymbol{\Sigma}}

\def\bx{\bm{x}}

\DeclareMathOperator*{\argmax}{arg\,max}
\DeclareMathOperator*{\argmin}{arg\,min}

\def\mR{{\mathbb R}}

\def\E{\mathbb{E}}

\def\L{{\mathcal L}}

\usepackage{diagbox}

\usepackage[symbol]{footmisc}

\usepackage{lastpage}
\jmlrheading{27}{2026}{1-\pageref{LastPage}}{9/24; Revised
1/26}{4/26}{24-1413}{Zhanrui Cai, Sai Li, Xintao Xia, Linjun Zhang}

\ShortHeadings{Differentially Private Estimation and Inference in High-Dimensional Regression}{Zhanrui Cai, Sai Li, Xintao Xia, Linjun Zhang}
\firstpageno{1}

\begin{document}

\title{Differentially Private Estimation and Inference in High-Dimensional Regression with FDR Control}

\author{\name Zhanrui Cai \email zhanruic@hku.hk \\
       \addr Faculty of Business and Economics\\
      The University of Hong Kong\\
       Pok Fu Lam Road, Hong Kong, China
       \AND
       \name Sai Li \email saili@tsinghua.edu.cn \\
       \addr Department of Statistics and Data Science\\
      Tsinghua University\\
       No. 30 Shuangqing Road, 100084, Beijing, China
              \AND
       \name Xintao Xia \email xintaox@zju.edu.cn \\
       \addr Center for Data Science\\
       Zhejiang University\\
      No. 866 Yuhangtang Road, 310058, Zhejiang, China
              \AND
       \name Linjun Zhang \email lz412@stat.rutgers.edu \\
       \addr Department of Statistics\\
     Rutgers University\\
    110 Frelinghuysen Rd, 08854, New Jersey, USA}

\footnote[0]{Authors are listed alphabetically. Zhanrui Cai is the corresponding author.}

\editor{Raef Bassily}
  
\maketitle

\begin{abstract}
This paper proposes new methodologies for conducting practical differentially private (DP) estimation and inference in high-dimensional linear regression. We first introduce a DP Bayesian Information Criterion (DP-BIC) for selecting the unknown sparsity parameter in differentially private sparse linear regression (DP-SLR), eliminating the need for prior knowledge of model sparsity, which is a requisite in the existing literature. Next, we develop the DP debiased algorithm that enables privacy-preserving inference on a particular subset of regression parameters. Our proposed method enables privacy-preserving inference on the regression parameters by leveraging the inherent sparsity of high-dimensional linear regression models. Additionally, we address private feature selection by considering multiple testing in high-dimensional linear regression by introducing a DP multiple testing procedure that controls the false discovery rate (FDR). This allows for accurate and privacy-preserving identification of significant predictors in the regression model. Through extensive simulations and real data analyses, we demonstrate the effectiveness of our proposed methods in conducting inference for high-dimensional linear models while safeguarding privacy and controlling the FDR.
\end{abstract}


\section{Introduction}\label{sec: intro}

In the era of big data, the significance of data privacy has grown considerably. With the continuous collection, storage, processing, and sharing of vast amounts of personal data, there is a pressing need to protect sensitive information. Unfortunately, traditional data analytics and statistical inference tools may fail to ensure such protection. The concept of differential privacy, initially proposed by theoretical computer scientists \citep{dwork2006calibrating}, has made substantial progress and found widespread use in various large-scale applications. Differentially private algorithms incorporate random noise independent of the original database and produce privatized summary statistics or model parameters. The ultimate goal of differentially private analysis is to safeguard individual data while allowing meaningful statistical analysis of the original database.

This work is motivated by the growing need to conduct statistical inference on confidential data, particularly when variable selection is required. In this paper, we analyze the \href{https://www.nrcs.usda.gov/nri}{National Resources Inventory} (NRI) data, a statistical survey of land use and natural resource conditions on U.S. non-Federal lands. Each observation in the NRI data includes information on soil conditions, water conditions, and other related resources at a specific geographic location on U.S. non-Federal lands. The NRI aims to assess the quantity and quality of natural resources while closely monitoring changes and trends, with a particular focus on soil erosion. Thus, it is crucial to provide accurate estimates and reliable confidence intervals for soil erosion to facilitate regular evaluations of the effectiveness of soil and water conservation practices, irrigation techniques, and farming technologies and practices. In this paper, we build a regression model to predict the long-term average annual soil loss based on available covariates such as climatic factors, erodibility factors, soil loss tolerance, land cover and use, wetland conditions, and other variables in the NRI dataset.

However, the NRI data are highly confidential. Using standard statistical methods may pose significant confidentiality risks. The locations of sampled points, along with other identifying details, are considered confidential information under 7 USC 2276 and the interpretive policy in NRCS General Manual Title 290, Part 400.11, B(4) in Appendix A. Improper release of such information violates federal law and can lead to serious legal consequences. Preventing the disclosure of sample location information in released analysis results is therefore essential. If attackers were able to identify the geographic location of even a single sample point, altering the original land conditions could introduce substantial bias into national resource estimates. Such bias could mislead government policy and ultimately threaten national security. Traditional data analytics and statistical inference tools often fail to protect NRI data, particularly with respect to location confidentiality. We apply our proposed methods to analyze water erosion using the NRI dataset, obtaining accurate estimates and valid confidence intervals while protecting data privacy.

In this paper, we develop a novel framework for conducting differentially private statistical inference in high-dimensional linear regression. Let $Y\in\mR$ denote the response and $\boldsymbol{X}\in\mR^p$ denote the covariates. Assume the random vector $(Y,\boldsymbol{X})$ follows the linear model
\begin{equation*}
	Y=\boldsymbol{\beta}^\top \boldsymbol{X}+e,
\end{equation*}
where $e$ is random noise following a Gaussian distribution. Let $\{(\boldsymbol{x}_i, y_i)\}_{i=1}^{n}$ be independent realizations of $(Y,\boldsymbol{X})$. We focus on the high-dimensional setting where $p$ may grow exponentially with $n$, and only a small subset of the coefficients in $\boldsymbol{\beta}$ are nonzero. Our goal is to develop differentially private estimation and inference methods for $\boldsymbol{\beta}$, along with a differentially private false discovery rate control procedure for selecting the nonzero coefficients.

In the typical non-private setting, numerous approaches have been developed to address the challenge of statistical inference in high-dimensional linear models. The debiased Lasso \citep{zhang2014confidence, javanmard2014confidence, van2014asymptotically} emerged as a technique to mitigate the bias inherent in the Lasso estimator, thereby providing asymptotically optimal confidence intervals for regression coefficients \citep{cai2017confidence}. More recently, \cite{wang2022finite} introduced the repro framework for finite-sample inference with high-dimensional covariates. Beyond inference on individual parameters, another key objective in high-dimensional linear regression is controlling the FDR of the variable selection. This objective has led to the development of FDR control methods in the literature. One influential approach is the knockoff framework introduced by \cite{barber2015controlling}, which exploits the symmetry of the statistics under the null hypothesis. The idea was further developed in numerous other settings \citep{candes2018panning, cai2025knockoffs}. Recently, \cite{dai2022false, dai2023scale} proposed a method that combines symmetric mirror statistics with data splitting to asymptotically control the FDR.

Addressing privacy concerns in high-dimensional statistical inference has received significant attention in recent literature. \cite{avella2021differentially} applied first and second-order optimization algorithms to develop private M-estimators and analyzed their asymptotic normality, along with the associated privacy error rate. \cite{xia2025statistical} proposed the statistical inference method for differentially private stochastic gradient descent. They demonstrated, both theoretically and empirically, that the error induced by the privacy mechanism can be made arbitrarily small. The problem of private multiple testing has also been actively studied \citep{dwork2018differentially, xia2023fdr, cai2025knockoffs}. 

This paper contributes to the differentially private analysis of high-dimensional linear regression in several key aspects. 
\begin{enumerate}
	\item We propose a DP-BIC to accurately select the unknown sparsity parameter in DP-SLR proposed by \cite{cai2019cost}, eliminating the need for prior knowledge of the model sparsity. This advancement enhances the reliability of the DP-SLR framework and can be used in many downstream tasks. 
	\item We develop a differentially private debiased procedure that yields asymptotically normal estimators under privacy guarantees. This procedure enables the construction of differentially private confidence intervals for individual parameters of interest.
	\item We design a differentially private method for controlling the FDR in multiple testing scenarios, which inevitably arise in high-dimensional inference problems under privacy constraints. Our approach achieves FDR control at any user-specified rate $\alpha$ and attains asymptotic power approaching one under mild conditions.
\end{enumerate}

\textit{Notation:} For any $p$-dimensional vector $\boldsymbol{x}=(x_1,\dots,x_p)^{\top}$, we define the $l_q$-norm of $\boldsymbol{x}$ for $1\leq q<\infty$ as $\|\boldsymbol{x}\|_q:=(\sum_{i=1}^{p}|x_i|^q)^{1/q}$ for $1\leq q$, with $|\cdot|$ representing the absolute value. The $l_\infty$-norm of $\boldsymbol{x}$ is defined as $\|\boldsymbol{x}\|_{\infty}:=\max_{i=1,\dots,p}|x_i|$. The $\text{supp}(\boldsymbol{x})=\{i:|x_i|>0\}$ is the index set of nonzero elements in $\boldsymbol{x}$. We define the $l_0$-norm of $\boldsymbol{x}$ by $\|\boldsymbol{x}\|_0=|\text{supp}(\boldsymbol{x})|$, which is the number of nonzero coordinates of $\boldsymbol{x}$. For a positive integer $n$, we use $[n]$ to denote the set $\{1,\dots,n\}$. For a subset $\mathcal{S}\subseteq[p]$ and vector $\boldsymbol{x}\in\mathbb{R}^{p}$, we use $\boldsymbol{x}_{\mathcal{S}}$ to denote the restriction of vector $\boldsymbol{x}$ to the index set $\mathcal{S}$ and $|\mathcal{S}|$ to denote the number of elements in $\mathcal{S}$. For a vector $\boldsymbol{x}\in\mathbb{R}^{p}$, we use $\Pi_R(\boldsymbol{x})$ to denote the projection of $\boldsymbol{x}$ onto the $l_2$-ball $\{\boldsymbol{u}\in\mathbb{R}^p:\|\boldsymbol{u}\|_2\leq R\}$, where $R$ is a positive real number. For a real symmetric matrix $\boldsymbol{A}\in\mathbb{R}^{p\times p}$, we use $\Lambda_{\min}(\boldsymbol{A})$ and $\Lambda_{\max}(\boldsymbol{A})$ to denote the minimum and maximum eigenvalues of $\boldsymbol{A}$. For a set of random variables $\{X_n\}_{n=1}^{\infty}$ and a random variable $X$, the notation $X_n\xrightarrow{D}X$ means $X_n$ converges to $X$ in distribution and the notation $X_n=O_P(a_n)$ means $X_n/a_n$ is stochastically bounded for a sequence of positive real numbers $\{a_n\}_{n=1}^{\infty}$.

\section{Preliminaries}\label{sec: formulation}

Consider the dataset $D:= \{(\boldsymbol{x}_i, y_i)\}_{i=1}^{n}\in\mathcal{D}$, drawn independently and identically from a distribution satisfying
$$y_i=\boldsymbol{\beta}^\top \boldsymbol{x}_i+e_i,$$
where $e_i$ follows a sub-Gaussian distribution and the unknown parameter $\boldsymbol{\beta}\in\mR^p$ satisfies $\|\boldsymbol{\beta}\|_0\le s$. We focus on the high-dimensional setting where the dimension $p$ may grow exponentially with the sample size $n$, while the sparsity $s$ grows slowly with $n$, all under the $(\varepsilon,\delta)$-DP framework. In what follows, we introduce the formal definitions of differential privacy and sensitivity.

\begin{definition}[Differential Privacy \citep{dwork2006calibrating}]
	\label{def:dp}
	A randomized algorithm $M(\cdot):\mathcal{D}\to\mathcal{R}$ is $(\varepsilon,\delta)$-DP for $\varepsilon,\delta>0$ if for every pair of neighboring data sets $D,D'\in\mathcal{D}$ that differ by one individual datum and every measurable set $\mathcal{S}\subset\mathcal{R}$ with respect to $M(\cdot)$,
	\begin{equation}
		\mathbb{P}(M(D)\in\mathcal{S})\leq e^{\varepsilon}\mathbb{P}(M(D')\in\mathcal{S})+\delta,
	\end{equation}
	where the probability measure $\mathbb{P}$ is induced by the randomness of $M(\cdot)$ only.
\end{definition}


\begin{definition}[Sensitivity]
\label{def:sensitivity}
For a vector-valued deterministic algorithm $\mathcal{T}(\cdot):\mathcal{D}\to\mR^{m}$, the $l_q$ sensitivity of $\mathcal{T}(\cdot)$ is defined as
\begin{equation}
\Delta_q(\mathcal{T}):=\sup_{D,D'\in\mathcal{D}}\|\mathcal{T}(D)-\mathcal{T}(D')\|_q,
\end{equation}
where $D$ and $D'$ only differ in one single entry.
\end{definition}

Sensitivity is extremely useful in characterizing the magnitude of change in the algorithm when a single individual in the dataset is replaced. In the appendix, we introduce some useful tools in DP, such as privacy mechanisms and composition theorems. In high-dimensional problems, parameters of interest are often assumed to be sparse. Reporting the entire set of estimation results can introduce substantial additional randomness due to privacy requirements. Fortunately, by exploiting sparsity, one can selectively disclose only the significant nonzero coordinates. The ``peeling" algorithm \citep{dwork2018differentially} is a differentially private algorithm that addresses this problem by identifying and returning the top-$k$ most significant coordinates based on the absolute values. Since its proposal by \cite{dwork2018differentially}, the algorithm has been widely used for protecting privacy in high-dimensional data analysis \citep{cai2019cost, xia2023fdr, xia2025differentially}. We summarize its details in Algorithm \ref{algo: noisy hard thresholding} and present its theoretical properties in Lemma \ref{lem:nht}.

\begin{algorithm}[!t]	
	\caption{Noisy Iterative Hard Thresholding (Peeling) ($NoisyIHT(\mathcal{T}(D), s^{\prime},\varepsilon,\delta,\lambda)$)}
	\label{algo: noisy hard thresholding}
	\begin{algorithmic}[1]
		\Require{Dataset $D$, vector-valued function $\mathcal{T}(D)=(\mathcal{T}(D)_1,\dots,\mathcal{T}(D)_d)^{\top} \in \mathbb{R}^d$, target sparsity $s^{\prime}$, privacy parameters $(\varepsilon, \delta)$, noise scale $\lambda$.}
		\State Initialize $S = \emptyset$.
		\For{$i = 1$ to $s^{\prime}$}
		\State	Generate $\boldsymbol{\eta}_i=(\eta_{i1}, \eta_{i2}, \cdots, \eta_{id} )^{\top} \in \mathbb{R}^d$ with $\eta_{i1}, \eta_{i2}, \cdots, \eta_{id} \stackrel{\text{i.i.d.}}{\sim} \text{Laplace}\{\lambda \cdot 2\sqrt{3s^{\prime}\log(1/\delta)}/\varepsilon\}$.
		\State	Append $j^* = \argmax_{j \in [d] \setminus S} |\mathcal{T}(D)_j| + \eta_{ij}$ to $S$.
		\EndFor
		\State Set $\tilde{P}_s\{\mathcal{T}(D)\} =\mathcal{T}(D)_S$.
		\State Generate $\tilde{\boldsymbol{\eta}}=(\tilde{\eta}_{1}, \tilde{\eta}_{2}, \cdots, \tilde{\eta}_{d} )^{\top} \in \mathbb{R}^d$ with $\tilde \eta_{1}, \cdots, \tilde \eta_{d} \stackrel{\text{i.i.d.}}{\sim} \text{Laplace}\{\lambda \cdot 2\sqrt{3s^{\prime}\log(1/\delta)}/\varepsilon\}$.
		\Ensure{$\tilde{P}_s\{\mathcal{T}(D)\} + \tilde{\boldsymbol{\eta}}_S$.}
	\end{algorithmic}
\end{algorithm}

\begin{lemma}[\cite{dwork2018differentially} and \cite{cai2019cost}]
	\label{lem:nht}
	For a vector-valued function $\mathcal{T}$ with $\|\mathcal{T}(D)-\mathcal{T}(D')\|_{\infty}\leq\lambda$, where $D'$ is a neighboring data set of $D$, Algorithm \ref{algo: noisy hard thresholding} is $(\varepsilon,\delta)$-DP.
\end{lemma}


\section{Differentially Private Estimation}
\label{sec: estimation}

The estimation of regression parameters in the high-dimensional differentially private setting has been studied by \cite{talwar2015nearly,thakurta2013differentially} and more recently by \cite{cai2019cost} with optimality guarantees for both statistical errors and privacy errors. However, existing algorithms \citep{thakurta2013differentially,cai2019cost} for high-dimensional differentially private estimation, when the dimension $p$ grows exponentially with the sample size $n$, require prior knowledge of the sparsity parameter $s$, which is typically unknown in practice. In this section, we propose the DP-BIC in Algorithm \ref{alg:linear} to select the sparsity parameter adaptively, eliminating the need for prior knowledge of the model sparsity. The pipeline of the proposed estimation algorithm is presented in Algorithm \ref{alg:linear}.

\begin{algorithm}[!t]
\caption{Adaptive Differentially Private Sparse Linear Regression}\label{alg:linear}
\begin{algorithmic}[1]
	\Require Dataset $\{(\boldsymbol{x}_i,  y_i)\}_{i=1}^{n}$, candidate set size $K$, step size $\eta^0$, privacy parameters $(\varepsilon, \delta)$, noise scale $B$, number of iterations $T$, truncation level $R$, feasibility parameter $C$, initial value $\boldsymbol{\beta}_{ini}$, constant $c_B$ in BIC criterion.
	\State Data splitting: randomly split the dataset into $T$ subsets of roughly equal size, $[n]= \mathcal{S}_0\cup\dots\cup\mathcal{S}_{T-1}$, where $\mathcal{S}_i\cap\mathcal{S}_j=\emptyset$ for $i\neq j$.
	\For {$k$ in $0$ to $K$}
	\State Initialization: $s^{\prime}=2^k$, $\boldsymbol{\beta}^{(0)}_k=\boldsymbol{\beta}_{ini}$.
	\If{$k>0$}
	\State Warm start: $\boldsymbol{\beta}^{(0)}_k=\hat{\boldsymbol{\beta}}(k-1)$.
	\EndIf 
	\For{$t$ in $0$ to $T-1$}
	\State Gradient descent: compute $\boldsymbol{\beta}^{(t + 0.5)}_k = \boldsymbol{\beta}^{(t)}_k - (\eta^0/|\mathcal{S}_{t}|)\sum_{i\in\mathcal{S}_{t}}  ( \Pi_R(\boldsymbol{x}_i^\top \boldsymbol{\beta}_k^{(t)})-\Pi_{R}(y_i))\boldsymbol{x}_i$, where $|\mathcal{S}_t|$ is the size of set $\mathcal{S}_t$ and $\Pi_R(x)$ denotes the projection of $x$ onto the $l_2$-ball $\{u\in\mathbb{R}:\|u\|_2\leq R\}$.
	\State	Private report: $\boldsymbol{\beta}^{(t+1)}_k = \Pi_C(\text{NoisyIHT}(\boldsymbol{\beta}_k^{(t+0.5)}, s^{\prime}, \varepsilon/\{T(K+2)\}, \delta/\{T(K+1)\}, \eta^0 B/|\mathcal{S}_t|))$, where $\Pi_C(\boldsymbol{x})$ denotes the projection of $\boldsymbol{x}$ onto the $l_2$-ball $\{\boldsymbol{u}\in\mathbb{R}^p:\|\boldsymbol{u}\|_2\leq C\}$.
	\EndFor
    \State Parameter clipping: $\hat{\boldsymbol{\beta}}(k)=\boldsymbol{\beta}_k^{(T)}/\max_i\{|\boldsymbol{x}_i^{\top}\boldsymbol{\beta}_k^{(T)}|/R,1\}$.
	\EndFor
	\State	Model selection:
	\begin{align*}
		\label{hbeta-adapt}
		\hat{\boldsymbol{\beta}}=\argmin_{\hat{\boldsymbol{\beta}}(k):0\leq k\leq K}\bigg[&\sum_{i=1}^{n}\{\Pi_R(y_i)-\Pi_R(\boldsymbol{x}_i^{\top}\hat{\boldsymbol{\beta}}(k))\}^2+z_k\\
        &+c_B\bigg\{\log(p)\log(n)\cdot 2^k+\frac{\log(p)^2\cdot 2^{2k}\log(1/\delta)\log(n)^7}{n\varepsilon^2}\bigg\}\bigg],
	\end{align*}
	where $z_k\stackrel{i.i.d.}{\sim} \text{Laplace}\{2(2R)^2(K+2)/\varepsilon\}$.
	\Ensure $\hat{\boldsymbol{\beta}}$.
\end{algorithmic}
\end{algorithm}

In high-dimensional model selection, information criteria such as the Bayesian Information Criterion (BIC) and the Generalized Information Criterion (GIC) have been widely studied. In general, an information criterion is constructed as
\begin{equation*}
    \text{estimate of risk functions}+a_n\times\text{measure of model complexity},
\end{equation*}
where $a_n$ is a positive sequence depending only on the sample size and the dimensionality of the covariates. Specifically, the ``measure of model complexity" corresponds to the sparsity parameter $s$ of the candidate model \citep{fan2013tuning}. When $p = O(n^{\kappa})$ for some $\kappa > 0$, \cite{wang2009shrinkage} proposed using $a_n = c_B \log(p)/n$ in the non-private setting, which corresponds to the first term in the proposed DP-BIC (Step 13 in Algorithm 2). A similar choice of $a_n$ was considered by \cite{fan2013tuning} for generalized linear models when $\log(p) = O(n^{\kappa})$. Intuitively, the second part of the equation helps avoid overfitted models by adding a penalty related to model sparsity. In this sense, $a_n$ should be larger. On the other hand, for underfitted models, the penalty should not exceed the improvement in the risk function achieved by incorporating important features. Intuitively, by choosing the penalty to approximately match the tight $\ell_2$ estimation error bound of $\boldsymbol{\beta}$ (i.e., $O(s \log(p)/n)$ for linear models under regularity conditions), one can prevent overfitting while mitigating underfitting. In the differentially private setting, greater sparsity implies larger error variance and thus reduces estimation accuracy. Our choice for the DP-BIC is closely related to the $\ell_2$ estimation error bound, which includes an additional term that depends on both model complexity and sparsity parameters \citep{cai2019cost}.

Algorithm \ref{alg:linear} incorporates several innovations. First, our choice to use powers of $2$ as candidate values for the sparsity parameter strikes a delicate balance and achieves two critical goals: (1) it ensures that the candidate set covers the true model by defining an interval in which $s$ falls, i.e., $s^*<s<2s^*$ for some $s^*$ in the candidate set; and (2) it limits the total number of candidate models to $O(\log(n))$, which is $o(n)$. This guarantees that the cost of privacy does not affect estimation accuracy beyond a logarithmic factor in the asymptotic setting. The required candidate set size $K=O\big(\max\{\log_2(\sqrt{n}/\log(p)^2 ),1\}\big)$ in Theorem \ref{thm:main} aligns with the sparsity requirements for statistical inference, as discussed in Section \ref{sec: inference}. The conditions can be relaxed to $K=O\big(\max\{\log_2(n/\log(p) ),1\}\big)$ by employing the cross-fitting technique of \cite{chernozhukov2018double}. Moreover, the choice of powers of $2$ can be replaced with any fixed base, providing greater flexibility in the algorithm. Second, the proposed algorithm employs random sample splitting, which is equivalent to employing the stochastic gradient descent algorithm with one pass of the entire dataset. Because of the splitting, $\boldsymbol{x}_i$ used in $t$-th iteration and $\boldsymbol{\beta}_k^{(t)}$ are independent, allowing us to obtain a high probability bound of $|\boldsymbol{x}_i^{\top}\hat{\boldsymbol{\beta}}_k^{(t)}|$ using the Chernoff bound. Note that sample splitting is not strictly necessary under stronger design assumptions commonly adopted in the differential privacy literature. For example, \cite{talwar2015nearly} considered the optimization over the set $\|\boldsymbol{\beta}\|_1\leq C$ for a given constant $C$, which is stronger than our Condition \ref{cond:2}; \cite{cai2019cost} assumed that for any subset $I\subset\{1,\dots,p\}$, $\|\boldsymbol{x}_I\|_{\infty}\leq c_x/\sqrt{|I|}$ and $1/L\leq |I|\cdot\Lambda_{\min}(\text{Cov}(\boldsymbol{x}_I\boldsymbol{x}_I^{\top}))\leq |I|\cdot\Lambda_{\max}(\text{Cov}(\boldsymbol{x}_I\boldsymbol{x}_I^{\top}))\leq L$, which are less commonly imposed than our Condition \ref{cond:1}. In the finite-sample case, when the sparsity satisfies $\sqrt{s}\leq\log(n)$, sample splitting in Algorithm \ref{alg:linear} can be omitted, and the full sample can be used at each step. Third, the proposed algorithm leverages private estimation outcomes from earlier steps with lower sparsity levels as warm starts, thereby improving the accuracy of subsequent estimation.

\begin{condition}
\label{cond:1}
The covariates $\boldsymbol{x}_i$ are independently sub-Gaussian with mean zero and covariance matrix $\boldsymbol{\Sigma}$, which satisfies $1/L\leq \Lambda_{\min}(\boldsymbol{\Sigma})\leq \Lambda_{\max}(\boldsymbol{\Sigma})\leq L$. Moreover, there exists a positive constant $c_x<\infty$ such that $\|\boldsymbol{x}_i\|_{\infty}\leq c_x$. 
\end{condition}

The design condition $\|\boldsymbol{x}_i\|_{\infty}\leq c_x$ in Condition \ref{cond:1} is widely adopted in the differential privacy literature to ensure bounded sensitivity (e.g., \cite{dwork2014analyze,talwar2015nearly,thakurta2013differentially}). It was also imposed in \cite{cai2019cost} to facilitate the statistical analysis of DP-SLR. This condition can be relaxed by employing a robust loss function \citep{avella2021differentially}. The upper bound on the infinity norm of $\boldsymbol{x}_i$ can also be weakened to hold with high probability, which is easily obtained for sub-Gaussian distributions with $c_x=O(\sqrt{\log(p)})$. With similar technical procedures, the second term of the error bound will be increased by $\log(p)$. The sub-Gaussian and bounded eigenvalue assumptions in Condition \ref{cond:1} are frequently assumed in high-dimensional literature \citep{van2014asymptotically}. Unlike the algorithm of \cite{cai2019cost}, we employ a data-splitting technique to establish independence between $\boldsymbol{\beta}^{(t)}_k$ and the sub-data $\boldsymbol{x}_i$ used in the $t$-th iteration. Combining this independence with Condition \ref{cond:2}, and by properties of sub-Gaussian random variables, we obtain the high-probability bound $|\boldsymbol{x}_i^{\top}\boldsymbol{\beta}^{(t)}|=O_p\{\sqrt{\log(n)}\}$. Lemma \ref{lem:privacy-adapt} provides the privacy guarantee of Algorithm \ref{alg:linear}, where only Condition \ref{cond:1} is required.

\begin{lemma}
\label{lem:privacy-adapt}
Suppose Condition \ref{cond:1} holds and $B\geq 4Rc_x$, Algorithm \ref{alg:linear} is $(\varepsilon,\delta)$-DP.
\end{lemma}

\begin{condition}
\label{cond:2}
The true parameter satisfies $\|\boldsymbol{\beta}\|_2\leq c_0$ for some constant $c_0>0$ and $\|\boldsymbol{\beta}\|_0\leq s$.
\end{condition}

The sparsity assumption in Condition \ref{cond:2} is commonly imposed in the high-dimensional literature and can be relaxed to approximate sparsity \citep{chen2007large,belloni2019valid}. In Condition \ref{cond:2}, the upper bound on the $\ell_2$ norm of $\boldsymbol{\beta}$ is used to control the sensitivity of the gradient, as in \cite{cai2019cost}. Bounding the sensitivity of the gradient function is necessary in differential privacy \citep{avella2021differentially}. Our Conditions \ref{cond:1} and \ref{cond:2} are less restrictive than the design conditions considered in \cite{cai2019cost} and \cite{thakurta2013differentially}. This relaxation comes at the cost of reducing the stochastic batch size by a factor of $1/T$, analogous to the comparison between stochastic gradient descent and traditional gradient descent. As shown in Theorem \ref{thm-adapt}, the number of iterations $T$ satisfies $T = O(\log(n))$, which leads to an increase of $O(\log(n))$ in the error bound. Throughout our analysis, the privacy parameters $(\varepsilon,\delta)$ are allowed to depend on the sample size and are not assumed to be fixed constants. We now establish an error bound for the proposed estimation procedure.

\begin{Theorem}
\label{thm-adapt}
Assume that Conditions \ref{cond:1} and \ref{cond:2} hold. Let $R=c_{\sigma}\sqrt{2\log(n)}$, $B=4Rc_x$ and $C>c_0$, where $c_{\sigma}$ is a positive constant only depends on $L$, $c_0$ and the distribution of the random error $e_i$. Suppose that the parameters satisfy $K=O\big(\max\{\log_2(\sqrt{n}/\log(p)^2 ),1\}\big)$ and $T=\rho L^2\log(8c_0^2Ln)$ for some positive constant $\rho$. Assume further that the following sparsity, dimensionality, and privacy conditions hold: $2^K>\rho L^4s$, $s^2\log (p)\log (n)=o(n)$, $s^{1.5}\log (p)\sqrt{\log(1/\delta)}\log (n)^{3.5}/\varepsilon=o(n)$ and $\log(1/\delta)\log(n)^3/\varepsilon^2=o(n^{1/2})$. Let the constant $c_B$ in the BIC criterion be a sufficiently large constant. Then with probability at least $1-\exp\{-c_1\log (n)\}$, there exist constants $c_2,c_3,c_4$, such that
\begin{equation*}
	\begin{split}
		\|\hat{\boldsymbol{\beta}}-\boldsymbol{\beta}\|_2^2\leq&c_2\frac{s\log(p)\log(n)}{n}+c_3\frac{s^2\log(p)^2\log(1/\delta)\log(n)^7}{n^2\varepsilon^2}+c_4\frac{\log(n)^3}{n\varepsilon}.
	\end{split}
\end{equation*}
\end{Theorem}

The first two terms in the upper bound match the minimax lower bound established in \cite{cai2019cost}, up to a logarithmic factor in $n$. Compared with the algorithm of \cite{cai2019cost}, which assumes a known sparsity level $s$, our proposed algorithm introduces an additional term $\log(n)^3/(n\epsilon)$, arising from the large deviation of the added random variable $z_k$ in the BIC criterion. The extra $\log(n)$ factors in both the statistical error and privacy cost terms result from our use of data splitting in the estimation process, the output of $K$ estimates in total, and the application of BIC for selecting the ``optimal" model. Additional design assumptions can further reduce the privacy error. For example, under assumptions $|y_i|=O_p(1)$ \citep{talwar2015nearly} and $\|\boldsymbol{x}_i\|_2=O_p(1)$ \citep{dwork2014analyze}, we can use the full dataset in each iteration. As a result, the estimation error can be reduced to $ \|\hat{\boldsymbol{\beta}}-\boldsymbol{\beta}\|_2^2=O_p(s\log(p)/n+s^2\log(p)^2\log(1/\delta)\log(n)^2/(n^2\varepsilon^2)+\log(n)/(n\varepsilon))$ when $K=O(1)$. Such design assumptions are often reasonable in practice, as the data are typically normalized before analysis.

The choice of an upper bound on model complexity in the candidate model class ensures that over-parameterized models still converge, albeit potentially at slower rates. Let the upper bound on sparsity be denoted by $s_{\max} = 2^K$. For high-dimensional model selection using BIC, Theorem 1 of \citet{fan2013tuning} required that $s_{\max}=o\big(\sqrt{\frac{n}{\log(p)\log(n)}}\big)$ across all candidate models. For high-dimensional model selection using cross-validation, \citep{chetverikov2021cross} imposed a lower bound on the $\ell_1$ penalty, which serves a role similar to that of an upper bound on model complexity and ensures the convergence of candidate models. In our setting, we require $s_{\max} = O\big(\frac{\sqrt{n}}{\log^2(p)}\big)$. This condition is motivated by the sparsity requirement for the standard debiased Lasso. According to \citep{van2014asymptotically}, the model sparsity $s$ should satisfy $s=o(\frac{\sqrt{n}}{\log(p)})$, which is consistent with our condition.

The results in Theorem \ref{thm-adapt} do not rely on a minimum signal strength condition, which is commonly assumed in the high-dimensional tuning-parameter selection literature, see \cite{fan2013tuning}. A key advantage of the debiased estimator—introduced in \eqref{eq:db}—is that valid inference requires only a specific convergence rate of the estimators. Consequently, our BIC procedure needs only to ensure a reasonable convergence rate. Our results further show that the $\ell_2$ difference between the estimates and the true coefficients can be bounded by the minimax rates, with an additional term $\log(n)^3/(n\varepsilon)$ arising from privacy constraints.

\section{Differentially Private Confidence Interval}\label{sec: inference}

In this section, we construct a confidence interval for a particular regression coefficient $\beta_j$, for $j\in[p]$ under $(\varepsilon,\delta)$-DP. Following the debiased Lasso framework, we first estimate the precision matrix $\boldsymbol{\Omega}:=\boldsymbol{\Sigma}^{-1}$ in a privacy-preserving manner. The $j$th column of $\boldsymbol{\Omega}$, denoted by $\boldsymbol{w}_j$, satisfies the linear equation $\boldsymbol{e}_j=\boldsymbol{\Sigma}\boldsymbol{w}_j$, where $\boldsymbol{e}_j$ is the unit vector with its $j$th component equal to $1$ and all other components equal to $0$. Thus, $\boldsymbol{w}_j$ is the unique minimizer of the convex quadratic function
\begin{equation}
\label{eq:w}
\frac{1}{2}\boldsymbol{w}_j^{\top}\boldsymbol{\Sigma}\boldsymbol{w}_j-\boldsymbol{w}_j^{\top}\boldsymbol{e}_j.
\end{equation}

We propose to estimate $\boldsymbol{w}_j$ by solving the empirical version of \eqref{eq:w} with an $\ell_0$ constraint. Our method differs slightly from node-wise regression \citep{van2014asymptotically}, which first performs a regression of $x_j$ on $\boldsymbol{x}_{-j}$, where $\boldsymbol{x}_{-j}$ contains all columns of $\boldsymbol{x}$ except $x_j$, and then estimates the residual variance. A key advantage of our approach is that it directly estimates $\boldsymbol{w}_j$, thereby eliminating the need for an additional composition theorem to combine the private estimation of node-wise regression coefficients with the private estimation of residual variance.

\begin{condition}
\label{cond:3}
The $j$-th column of $\boldsymbol{\Omega}$ is sparse and satisfies $\|\boldsymbol{w}_j\|_0\leq s_j$ for $j\in[p]$.
\end{condition}

The sparsity assumption of the precision matrix is frequently adopted in the high-dimensional statistical inference literature \citep{zhang2014confidence, javanmard2014confidence, van2014asymptotically}. This condition is also essential for enabling differentially private precision matrix estimation. Note that the $\ell_2$ norm of $\boldsymbol{w}_j$ is bounded by $L$ under Condition \ref{cond:1}. Following the idea of Algorithm \ref{alg:linear}, we propose using a BIC criterion to select the optimal model when the sparsity level $s_j$ is unknown. The algorithm for estimating $\boldsymbol{w}_j$ is summarized in Algorithm \ref{alg3}. Lemma \ref{lem:privacy_adap_w} shows that, under certain regularity conditions, Algorithm \ref{alg3} is $(\varepsilon,\delta)$-DP and the theoretical properties of $\hat{\boldsymbol{w}}_j$ are established in Lemma \ref{lem:est_w_bic}.

\begin{algorithm}[!t]\caption{Adaptive Differentially Private Estimation of $\boldsymbol{w}_j$}\label{alg3}
\begin{algorithmic}[1]
	\Require Dataset $ \{\boldsymbol{x}_i\}_{i =1}^{n}$, candidate set size $K$, step size $\eta^0$, privacy parameters $(\varepsilon, \delta)$, noise scale $B$, number of iterations $T$, truncation level $R$, feasibility parameter $C$, initial value $\boldsymbol{w}_{ini}$, constant $c_B$ in BIC criterion.
	\State Data splitting: randomly split data into $T$ parts of roughly equal size, $[n] = \mathcal{S}_0\cup\dots\cup\mathcal{S}_{T-1}$, where $\mathcal{S}_i\cap\mathcal{S}_j=\emptyset$ for $i\neq j$.
	\For{$k$ in $0$ to $K$}
	\State Initialization: $s_j=2^k$, $\boldsymbol{w}_{j,k}^{(0)}=\boldsymbol{w}_{ini}$.
	\If{$k>0$}
	\State Warm start: $\boldsymbol{w}_{j,k}^{(0)}=\hat{\boldsymbol{w}}_j(k-1)$.
	\EndIf
	\For{$t$ in $0$ to $T-1$}
	\State Gradient descent: $\boldsymbol{w}_{j,k}^{(t + 0.5)} = \boldsymbol{w}_{j,k}^{(t)} -\eta^0\{\boldsymbol{e}_j-\sum_{i\in\mathcal{S}_t}\boldsymbol{x}_i\Pi_R(\boldsymbol{x}_i^{\top}\boldsymbol{w}_{j,k}^{(t)})/|\mathcal{S}_t|\}$.
	\State Private report: $\boldsymbol{w}_{j,k}^{(t+1)} = \Pi_C(\text{NoisyIHT}(\boldsymbol{w}_{j,k}^{(t+0.5)}, s_j, \varepsilon/\{T(K+2)\}, \delta/\{T(K+1)\}, \eta^0 B/|\mathcal{S}_t|))$.
	\EndFor
     \State Parameter clipping: $\hat{\boldsymbol{w}}_j(k)=\boldsymbol{w}_{j,k}^{(T)}/\max_i\{|\boldsymbol{x}_i^{\top}\boldsymbol{w}_{j,k}^{(T)}|/R,1\}$.
	\EndFor
	\State  Model selection:
	\begin{align*}
		\label{hbeta-adapt} \hat{\boldsymbol{w}}_j=\argmin_{\hat{\boldsymbol{w}}_j(k):0\leq k\leq K}\bigg[&\sum_{i=1}^{n}\{\Pi_R(\hat{\boldsymbol{w}}_j(k)^{\top}\boldsymbol{x}_i)\Pi_R(\boldsymbol{x}_i^{\top}\hat{\boldsymbol{w}}_j(k))/2-\hat{\boldsymbol{w}}_j(k)^{\top}\boldsymbol{e}_j\}+z_k\\
		&+c_B\bigg\{\log(p)\log(n)\cdot 2^k+\frac{\log(p)^2\cdot 2^{2k}\log(1/\delta)\log(n)^7}{n\varepsilon^2}\bigg\}\bigg]\nonumber,
	\end{align*}
	where $z_k\stackrel{i.i.d.}{\sim} \text{Laplace}\{(K+2)R^2/\varepsilon\}$.
	\Ensure{$\hat{\boldsymbol{w}}_j$.}
\end{algorithmic}
\end{algorithm}

\begin{lemma}
\label{lem:privacy_adap_w}
Under Conditions \ref{cond:1}, \ref{cond:2}, and $B\geq 2Rc_x$, Algorithm \ref{alg3} is $(\varepsilon,\delta)$-DP.
\end{lemma}

\begin{lemma}
\label{lem:est_w_bic} 
Assume that Conditions \ref{cond:1}, \ref{cond:2} and \ref{cond:3} hold. Let $B=2Rc_x$, $C>L$ and $R=C_1\sqrt{\log(n)}$ for a constant $C_1$. Suppose that the tuning parameters satisfy $K=O\big(\max\{\log_2(\sqrt{n}/\log(p)^2 ),1\}\big)$ and $T=\rho L^2\log(8L^3n)$ for some positive constant $\rho$. Assume further that the following sparsity, dimensionality, and privacy conditions hold: $2^K>\rho L^4s_j$, $T=\rho L^2\log(8L^3n)$, $s_j^2\log (p)\log (n)=o(n)$ and $s_j^{1.5}\log (p)\sqrt{\log(1/\delta)}\log (n)^{3.5}/\varepsilon=o(n)$ and $\log(1/\delta)\log(n)^3/\varepsilon^2=o(n^{1/2})$. Let the constant $c_B$ in the BIC criterion be a sufficiently large constant. Then, with probability at least $1-\exp\{-c_1\log (n)\}$, there exist constants $c_2,c_3,c_4$, such that
\begin{equation*}
	\begin{split}
		\|\hat{\boldsymbol{w}}_j-\boldsymbol{w}_j\|_2^2\leq&c_2\frac{s_j\log (p)\log (n)}{n}+c_3\frac{s_j^2\log (p)^2\log(1/\delta)\log(n)^7}{n^2\varepsilon^2}+c_4\frac{\log(n)^3}{n\varepsilon}.
	\end{split}
\end{equation*}
\end{lemma}

One can show that the first two terms in the error bound of Lemma \ref{lem:est_w_bic} match the minimax lower bound, up to a logarithmic factor of $n$, using the ``tracing attack" technique developed in \cite{cai2019cost}. Compared with the case where the sparsity level $s_j$ is known, the proposed algorithm introduces an additional factor of $\log(n)$ in the privacy cost component of the error bound due to the DP-BIC selection step. The privacy error can be further reduced under additional assumptions, as discussed earlier.

After obtaining the private estimator $\hat{\boldsymbol{w}}_j$, we propose the following differentially private debiased estimator to facilitate private inference:
\begin{equation}
\label{eq:db}
\hat{\beta}_j^{(db)}=\hat{\beta}_j+\frac{1}{n}\sum_{i=1}^{n} \Pi_R(\boldsymbol{x}_i^{\top}\hat{\boldsymbol{w}}_j)(\Pi_R(y_i)-\Pi_R(\boldsymbol{x}_i^{\top}\hat{\boldsymbol{\beta}}))+z_j^{(db)},
\end{equation} 
where $\hat{\beta}_j^{(db)}$ denotes the debiased estimator of the $j$th component, $\hat{\beta}_j$ is the $j$th component of $\hat{\boldsymbol{\beta}}$, and $z_j^{(db)}\sim N(0,(4R^2/n)^2\cdot 2\log(1.25/\delta)/\varepsilon^2)$. Unlike the non-private debiased estimator in \cite{van2014asymptotically}, the proposed estimator \eqref{eq:db} incorporates additional random noise $z_j^{(db)}$ to guarantee $(\varepsilon,\delta)$-DP, since the debiasing step involves the dataset. Given $\hat{\boldsymbol{w}}_j$ and $\hat{\boldsymbol{\beta}}$, the debiased estimator $\hat{\beta}_j^{(db)}$ is $(\varepsilon,\delta)$-DP by the Gaussian mechanism. Owing to privacy constraints, the variance analysis of $\hat{\beta}_j^{(db)}$ differs from that in \cite{van2014asymptotically}. The following lemma provides theoretical insights into the decomposition of the private debiased estimator $\hat{\beta}^{(db)}_j$.

\begin{lemma}[Limiting distribution of the private debiased estimator]
\label{lem:three}
Assume the same conditions as in Theorem \ref{thm-adapt} and Lemma \ref{lem:est_w_bic}. Let $s_0=\max\{s,s_j\}$, $R=\max\{c_{\sigma},L\}\sqrt{2\log(n)}$. Then
\begin{align*}
\sqrt{n}(\hat{\beta}^{(db)}_j-\beta_j)=u_j+v_j+\sqrt{n}z_j^{(db)}, 
\end{align*}
where $u_j\stackrel{D}{\to}N(0,\Omega_{j,j}\sigma^2)$ and is independent of $z_j^{(db)}$, $\Omega_{j,j}$ is the $(j,j)$th entry of the precision matrix $\boldsymbol{\Omega}$, and $\sigma^2$ is the variance of $e_i$ in the linear model. Let
$$r_n=s_0\log(p)\log(n)/n^{1/2}+s_0^{2}\log(p)^2\log(1/\delta)\log(n)^7/(n^{1.5}\varepsilon^2)+\log(n)^3/(n^{1/2}\varepsilon).$$
The remainder term satisfies $v_j=O_p(\max(r_n^{1/2},r_n))$.
\end{lemma}

Therefore, we need to estimate the variance $\Omega_{j,j}\sigma^2$ in a differentially private manner in order to construct a differentially private confidence interval. Note that an estimate of $\Omega_{j,j}$ can be directly obtained from $\hat w_{j,j}$, the $j$th component of $\hat{\boldsymbol{w}}_j$. Thus, we only need to estimate $\sigma^2$. We propose the following differentially private estimator of $\sigma^2$:
\begin{equation*}
\hat\sigma^2=\frac{1}{n}\sum_{i=1}^n\{\Pi_R(y_i)-\Pi_R(\boldsymbol{x}_i^\top\hat{\boldsymbol{\beta}})\}^2+z,
\end{equation*}
where $z\sim N(0,\{\frac{2(2R)^2}{n}\}^2\cdot\frac{2\log(1.25/\delta)}{\varepsilon^2})$. The added noise term $z$ ensures that the estimate $\hat\sigma^2$ satisfies $(\varepsilon,\delta)$-DP.

For the reader's convenience, we summarize the complete algorithm for constructing a differentially private confidence interval for $\beta_j$ in Algorithm \ref{algo:inference}. The algorithm consists of four steps, with an allocated privacy budget of $(\varepsilon/4,\delta/4)$ for each step: (1) estimating the regression parameter $\hat{\boldsymbol{\beta}}$; (2) estimating the corresponding column of the precision matrix, $\hat{\boldsymbol{w}}_j$; (3) computing the debiased estimator $\hat{\beta}_j^{(db)}$; and (4) estimating the standard error of the debiased estimator. Since each of these four steps is $(\varepsilon/4,\delta/4)$-DP, Algorithm \ref{algo:inference} as a whole satisfies $(\varepsilon,\delta)$-DP. The privacy budget allocation is flexible and can be adjusted in practice depending on specific requirements. The overall privacy guarantee, along with the nominal coverage of the proposed confidence interval, is given in Theorem \ref{thm:main}.

\begin{algorithm}[ht]	\caption{$(1-\alpha)\times 100\%$ differentially private confidence interval for $\beta_j$}\label{algo:inference}

\begin{algorithmic}[1]
	\Require{Dataset $ \{(\boldsymbol{x}_i,  y_i)\}_{i=1}^{n}$, privacy parameters $(\varepsilon, \delta)$, confidence level $\alpha$, truncation level $R$.}	
	\State Compute $\hat{\boldsymbol{\beta}}$ using Algorithm \ref{alg:linear} with privacy parameters $(\varepsilon/4,\delta/4)$ and tuning parameters defined in Theorem \ref{thm-adapt}.
	\State Compute $\hat{\boldsymbol{w}}_j$ using Algorithm \ref{alg3} with privacy parameters $(\varepsilon/4,\delta/4)$ and tuning parameters defined in Lemma \ref{lem:est_w_bic}.
	\State Debiased estimator:  
	\begin{align*}
		\hat{\beta}_j^{(db)}=\hat{\beta}_j+\frac{\sum_{i=1}^{n} \Pi_R(\boldsymbol{x}_i^{\top}\hat{\boldsymbol{w}}_j)(\Pi_R(y_i)-\Pi_R(\boldsymbol{x}_i^{\top}\hat{\boldsymbol{\beta}}))}{n}+z_j^{(db)},
	\end{align*}
	where $z_j^{(db)}\sim N(0,16(4R^2/n)^2\cdot 2\log(4\times 1.25/\delta)/\varepsilon^2)$.
	\State Compute confidence interval:
	\begin{equation}
		\label{eq-ci}
		I_j = [\hat{\beta}^{(db)}_j- z_{1-\alpha/2}\sqrt{\widehat{V}_j},     \;\;   \hat{\beta}^{(db)}_j+ z_{1-\alpha/2}\sqrt{\widehat{V}_j}], 
	\end{equation}
	where $\hat V_j$ is defined as $$
	\hat V_j^2=\frac{\hat{w}_{j,j}\hat\sigma^2}{n},$$ 
	and $z_{1-\alpha/2}$ is the $(1-\alpha/2)$th quantile of the standard normal distribution, and $\hat{\sigma}^2=\frac{1}{n}\sum_{i=1}^n(\Pi_R(y_i)-\Pi_R(\boldsymbol{x}_i^\top\hat{\boldsymbol{\beta}}))^2+z$, where $z\sim N(0,16\{2(2R)^2/n\}^2\cdot 2\log(4\times 1.25/\delta)/\varepsilon^2)$.
	\Ensure{$I_j$.}
\end{algorithmic}
\end{algorithm}

\begin{Theorem}[Validity of the proposed CI]
\label{thm:main}
Under Conditions \ref{cond:1} and \ref{cond:2}, Algorithm~\ref{algo:inference} is $(\varepsilon,\delta)$-DP. Under the assumptions of Lemma \ref{lem:three}, and we assume $s_0\log (p)\log (n)/\sqrt n=o(1)$, $s_0^2\log^2 (p)\log(1/\delta)\log(n)^{7}/(n^2\varepsilon^2)=o(n^{-1/2})$, $\log(n)^3/\varepsilon=o(n^{1/2})$, $\log(n)\log(1/\delta)^{1/2}/\varepsilon=o(n^{1/2})$. We have 
$$
\lim_{n\to\infty}\Pro(\beta_j\in I_j)=1-\alpha.
$$
\end{Theorem}

Theorem \ref{thm:main} shows that the proposed algorithm achieves asymptotic nominal coverage while ensuring privacy. The condition $s_0\log (p)\log (n)/\sqrt n=o(1)$ matches that assumed in the non-private debiased Lasso of \cite{cai2017confidence}, while the additional rate conditions arise from privacy constraints. The choice of the upper bound $K$ in Algorithms \ref{alg:linear} and \ref{alg3} is crucial for obtaining the $\ell_1$ bound of the estimation error $\hat{\boldsymbol{\beta}}-\boldsymbol{\beta}$ and $\hat{\boldsymbol{w}}_j-\boldsymbol{w}_j$, which are key to deriving the debiased estimator. These conditions can be relaxed to $K = \log_2(n/\log(p))$ by applying the data-splitting technique of \cite{chernozhukov2018double}.

The condition $\log(n)\log(1/\delta)^{1/2}/\varepsilon=o(n^{1/2})$ in Theorem \ref{thm:main} ensures that the variance of $\sqrt{n}z_j^{(db)}$ is $o(1)$, and further implies that the asymptotic variance of $\sqrt{n}(\hat{\beta}^{(db)}_j-\beta_j)$ equals that of the non-private debiased estimator. Nevertheless, in finite samples, we recommend incorporating a minor correction by including the variance of $\sqrt{n}z_j^{(db)}$ in the confidence interval to improve finite-sample performance. A similar approach was previously discussed by \cite{avella2021differentially} in the context of low-dimensional noisy gradient descent and noisy Newton's method algorithms. 

Note that the proposed debiased estimator \eqref{eq:db} incorporates an additional noise $z_j^{(db)}$, generated from a Gaussian distribution with known variance. The confidence interval with finite-sample correction is defined by accounting for the variance of $z_j^{(db)}$ as follows:
\begin{equation}\label{eq: correction}
I_j = [\hat{\beta}^{(db)}_j- z_{1-\alpha/2}\sqrt{\widehat{V}_j/n+V_c},     \;\;   \hat{\beta}^{(db)}_j+ z_{1-\alpha/2}\sqrt{\widehat{V}_j/n+V_c}], 
\end{equation}
where $V_c=16(\frac{4R^2}{n})^2\cdot\frac{2\log(4\times 1.25/\delta)}{\varepsilon^2}$ represents the variance of $z_j^{(db)}$. Since $V_c$ is small by assumption, it is dominated by $\widehat{V}_j/n$ as $n \to \infty$. Consequently, the corrected confidence interval remains asymptotically efficient relative to the debiased Lasso. However, in small samples, the effect of the additional noise should be taken into account, as demonstrated in the simulation study.

\section{Differentially Private FDR Control}\label{sec: fdr}

Under differential privacy constraints, it is crucial to perform parameter selection with FDR control and to release debiased estimators only for the selected subset of parameters. In Section \ref{sec: inference}, we consider inference for a particular $\beta_j$ by constructing a debiased estimator $\hat{\beta}_j^{(db)}$ under $(\varepsilon, \delta)$-DP. For commonly used privacy parameters—such as $\varepsilon = 1$ and $\delta = n^{-1-\kappa}$ for some $\kappa > 0$—the composition theorem implies that releasing the full set of debiased estimators $\{\hat{\beta}_j^{(db)}\}_{j=1}^{p}$ would require allocating a privacy budget of $(\varepsilon/p, \delta/p)$ to each individual estimator. Such an allocation induces a large privacy error in estimating $\hat{\boldsymbol{w}}_j$ (defined in Lemma \ref{lem:est_w_bic}), making it impossible to obtain a consistent estimator of $\boldsymbol{w}_j$ and breaking the validity of the inference procedure. These observations underscore the necessity of variable selection with FDR control and the release of debiased estimators only for selected parameters.

False Discovery Rate (FDR) control with privacy guarantees in high-dimensional linear models is a challenging problem. Existing approaches to differentially private FDR control \citep{dwork2018differentially, xia2023fdr} require mutual independence of $p$-values under the null hypotheses, an assumption that does not necessarily hold in linear regression settings. Our approach draws inspiration from the recent advancements in mirror statistics \citep{dai2022false, dai2023scale}. In particular, the use of sample splitting and post-selection techniques enables effective dimensionality reduction, transforming a high-dimensional problem into one of substantially lower dimension. This reduction, in turn, allows us to more efficiently manage the scale of noise required for privacy preservation.

Specifically, we divide the data into two parts, denoted by $\mathcal{D}_1$ and $\mathcal{D}_2$. We first apply the high-dimensional DP-SLR algorithm to $\mathcal{D}_1$. The resulting estimator is denoted by $\tilde{\boldsymbol{\beta}}_{(1)}$, with its support defined as $\mathcal{A}:=\{j\in[p]:\tilde{\beta}_{(1)j}\neq 0\}$, where $\tilde{\beta}_{(1)j}$ is the $j$th component of $\tilde{\boldsymbol{\beta}}_{(1)}$. We then use $\mathcal{D}_2$ to fit a differentially private ordinary least squares (DP-OLS) model based on the estimated active set $\mathcal{A}$, and denote the resulting estimator by $\tilde{\boldsymbol{\beta}}_{(2)}$. For each $j\in\mathcal{A}$, we define the mirror statistic $M_j$ as $M_j:=\text{sign}(\tilde{\beta}_{(1)j}\tilde{\beta}_{(2)j})f(|\tilde{\beta}_{(1)j}|,|\tilde{\beta}_{(2)j}|)$. Following \cite{dai2022false}, the function $f$ can be chosen as $f(u,v)=2\min(u,v)$, $f(u,v)=uv,$ or $f(u,v)=u+v$. The data-driven cutoff $\tau_q$ is defined as 
\begin{equation*}
	\tau_q:=\min\bigg\{t>0:\frac{\#\{j:M_j<-t,j\in\mathcal{A}\}}{\#\{j:M_j>t,j\in\mathcal{A}\}\lor 1}\leq q\bigg\},
\end{equation*}
where $q$ is the target FDR level and $\#$ denotes the cardinality of a set. We select the subset of variables $\mathcal{A}_{\tau_q}=\{j\in\mathcal{A}:M_j>\tau_q\}$ as the important variables. Let $\mathcal{S}:=\{j\in[p]:\beta_j\neq 0\}$ denote the true support set, and let $\bar{\mathcal{S}}=[p]-\mathcal{S}$ denote its complement. The \textit{false discovery proportion} (FDP), FDR and \textit{power} of the proposed selection procedure are defined as
\begin{equation*}
\text{FDP}(\mathcal{A}_{\tau_q}):=\frac{|\mathcal{A}_{\tau_q}\cap\bar{\mathcal{S}}|}{|\mathcal{A}_{\tau_q}|\lor 1},\quad\text{FDR}(\mathcal{A}_{\tau_q})=\mathbb{E}\{\text{FDP}(\mathcal{A}_{\tau_q})\}, \quad \text{Power}(\mathcal{A}_{\tau_q}):=\frac{|\mathcal{A}_{\tau_q}\cap\mathcal{S}|}{|\mathcal{S}|}.
\end{equation*}
We summarize the details of the algorithm in Algorithm \ref{alg:fdr}. The privacy guarantee of the proposed procedure is established in Lemma \ref{lem:privacy_fdr}.

\begin{lemma}
\label{lem:privacy_fdr}
Assume Conditions \ref{cond:1} and \ref{cond:2} hold. Then Algorithm \ref{alg:fdr} is $(2\varepsilon,2\delta)$-DP provided that $B_1\geq 4|\mathcal{A}|c_x^2/n$ and $B_2\geq 4R\sqrt{|\mathcal{A}|}c_x /n$.
\end{lemma}

Under mild conditions, the proposed method asymptotically controls the FDR at a user-specified level $q$, while the power approaches $1$. We summarize these results in Theorem \ref{thm:fdr}.

\begin{Theorem}
\label{thm:fdr}
Suppose the conditions in Theorem \ref{thm-adapt} hold and assume that $\hat{s}^3\sqrt{\log(1/\delta)}/\varepsilon=o(n^{1/2})$, where $\hat{s}$ denotes the size of the selected support set $\mathcal{A}$. If the signal strength satisfies: \begin{equation*}
	\begin{split}
		\min_{j\in\mathcal{S}}|\beta_j|\gg&\max\{\sqrt{s\log (p)\log(n)/n},s\log(p)\log(n)^{3.5}\log(1/\delta)^{0.5}/(n\varepsilon),\log(n)^{1.5}/\sqrt{n\varepsilon}\},
	\end{split}
\end{equation*}
where the true support set is defined as $\mathcal{S}:=\{j\in[p]:\beta_j\neq 0\}$, then the output of Algorithm \ref{alg:fdr} satisfies $\limsup_{n,p\to\infty}{\rm FDR}(\mathcal A_{\tau_q})\leq q$, for any nominal FDR level $q\in(0,1)$.

Moreover, if the signal strength further satisfies
$$\min_{j\in\mathcal{S}}|\beta_j|\gg\max\{\hat{s}^{1/2}\log(n)^{1/2},\hat{s}^{3/2}\}\sqrt{\log(1/\delta)}/(n\varepsilon),$$
then $\liminf_{n,p\to\infty}{\rm Power}(\mathcal A_{\tau_q})=1.$
\end{Theorem}

The first minimal signal strength condition guarantees the SURE screening property \citep{fan2008sure}, i.e., the set $\mathcal{A}$ contains all active coefficients. This property is essential for controlling the FDR in high-dimensional linear models; see \cite{barber2019knockoff} and \cite{dai2022false}. A critical requirement for valid FDR control is that the linear model continues to hold conditional on the selected set $\mathcal{A}$. By employing a data-splitting strategy, this condition can be relaxed to require only that the selected set $\mathcal{A}$ contains all active coefficients with high probability. The sparsity assumption ensures the consistency of the DP-OLS estimator. Similar conditions were imposed by \cite{dwork2014analyze} for the consistent estimation of covariance matrices. This requirement can be satisfied by choosing an appropriate upper bound for the sparsity level $2^K$ in Algorithm \ref{alg:fdr}. Since the target of the first-stage estimation is to prescreen the data, we can instead apply the algorithm of \cite{cai2019cost} with a conservative choice of sparsity level. For the power analysis, a minimal signal strength condition is also necessary to account for the estimation error inherent in the DP-OLS procedure.
	
Regarding DP-FDR control, we acknowledge the latest mirror statistics developed in \cite{dai2023scale}. However, directly implementing the algorithm of \cite{dai2023scale} would result in the noise required for privacy overwhelming the signals. This is because DP-FDR control requires the screening step to reduce the number of tests to a moderate level, ensuring that the amount of noise needed remains manageable. See, for example, the peeling algorithm in \cite{dwork2018differentially} and the mirror-peeling algorithm in \cite{xia2023fdr}.
	
\begin{algorithm}[ht]\caption{Differentially Private False Discovery Rate Control}\label{alg:fdr}
\begin{algorithmic}[1]
	\Require{Dataset $ \{(\boldsymbol{x}_i,  y_i)\}_{i}^{n}$, privacy parameters $(\varepsilon, \delta)$, noise scale $B_1$ and $B_2$, target FDR $q$.}
	\State  Data splitting: randomly split data into $\mathcal{D}_1$ and $\mathcal{D}_2$, each of roughly equal size.
	\State Compute $\tilde{\boldsymbol{\beta}}_{(1)}$ using DP-SLR with data $\mathcal{D}_1$ with privacy parameters $(\varepsilon,\delta)$. Denote the support set of $\tilde{\boldsymbol{\beta}}_{(1)}$ by $\mathcal{A}$.
	\State Estimate
	\begin{equation*}
		\begin{split}
			\tilde{\boldsymbol{\beta}}_{(2)\mathcal{A}}:=&\big(\sum_{i\in\mathcal{D}_2}\boldsymbol{x}_{i,\mathcal{A}}\boldsymbol{x}_{i,\mathcal{A}}^{\top}/|\mathcal{D}_2|+\boldsymbol{N}_{XX}\big)^{-1}\times\big(\sum_{i\in\mathcal{D}_2}\boldsymbol{x}_{i,\mathcal{A}}^{\top}\Pi_R(y_i)/|\mathcal{D}_2|+\boldsymbol{N}_{XY} \big),
		\end{split}
	\end{equation*}
	where $\boldsymbol{x}_{i,\mathcal{A}}$ is the subvector of $\boldsymbol{x}_i$ corresponding to the index set $\mathcal{A}$. The matrix $\boldsymbol{N}_{XX}$ is a $|\mathcal{A}|\times|\mathcal{A}|$ symmetric matrix with i.i.d. entries drawn from $N(0,B_1^2\cdot8\log(2.5/\delta)/\varepsilon^2)$, and $\boldsymbol{N}_{XY}$ is a $|\mathcal{A}|\times 1$ vector with i.i.d. entries drawn from $N(0,B_2^2\cdot8\log(2.5/\delta)/\varepsilon^2)$.
	\State For each $j\in\mathcal{A}$, compute the mirror statistic $M_j$ by $$M_j=\text{sign}(\tilde{\beta}_{(1)j}\tilde{\beta}_{(2)j})f(|\tilde{\beta}_{(1)j}|,|\tilde{\beta}_{(2)j}|);$$
	\State   Let the data-driven cutoff $\tau_q$ be defined as $$\tau_q:=\min\bigg\{t>0:\frac{\#\{j:M_j<-t,j\in\mathcal{A}\}}{\#\{j:M_j>t,j\in\mathcal{A}\}\lor 1}\leq q\bigg\};$$
	
	\Ensure{subset $\mathcal{A}_{\tau_q}=\{j\in\mathcal{A}:M_j>\tau_q\}$}.
\end{algorithmic}
\end{algorithm}

\section{Numeric Study}
\label{sec:numeric}
\subsection{Simulation}

We evaluate the finite-sample behavior of the private debiased procedure for inference on individual regression coefficients, as well as the false discovery rate of the selection procedure. In the simulation study, we consider linear models where the rows of the covariate matrix  $\boldsymbol{X}$ are i.i.d. drawn from $N(\boldsymbol{0},\boldsymbol{\Sigma})$. The response variable $\boldsymbol{y}$ is generated according to the linear model $\boldsymbol{y}=\boldsymbol{X}\boldsymbol{\beta}+\boldsymbol{e}$, where $\boldsymbol{y}\in\mathbb{R}^n$, $\boldsymbol{X}\in\mathbb{R}^{n\times p}$, and $\boldsymbol{e}\in\mathbb{R}^n$.

\subsubsection{Debiased Inference}

We first evaluate the performance of the proposed debiased procedure under two designs. Consider the Toeplitz covariance matrices (AR) and the block equicorrelated covariance matrices for the design matrix:
\begin{equation*}
\begin{split}
	\text{Toeplitz: }&\Sigma_{j,k}=\rho^{|j-k|}\text{ for }j,k\in\{1,\dots,p\}\\
	\text{Block equicorrelated: }&\Sigma = I_{p/4}\otimes\bigl((1-\rho)I_{4}+\rho J_{4}\bigr),
\end{split}
\end{equation*}
where $\otimes$ denotes the Kronecker product, so that $\Sigma$ is block diagonal with $p/4$ identical $4\times 4$ equi-correlated blocks. The active set has cardinality $s_0=|S_0|=3$ and is given by $S_0=\{1,2,3\}$. The nonzero regression coefficients are fixed at $1$. The errors are independently drawn from $N(0,1)$. The sample size is set to $n=2000$, and the number of covariates is $p=2000$. The privacy parameters are $\varepsilon=4$ and $\delta=1/n^{1.1}$ for each coordinate. The number of candidate models is $K=2$ for the debiased inference. The number of iterations is $T=2$, and the step size is $\eta^0=4$.
For comparison, we report coverages and interval lengths for three methods: DB-Lasso, which is the debiased Lasso method in \cite{van2014asymptotically}; DP naive, which is the proposed DP debiased procedure (Algorithm \ref{algo:inference}) without finite-sample correction; and DP correction, the proposed debiased procedure with finite-sample correction in formula (\ref{eq: correction}). All the results are based on $100$ independent repetitions of the model with random design and fixed regression coefficients.

\begin{figure}[ht]
\centering
\includegraphics[width=0.5\linewidth]{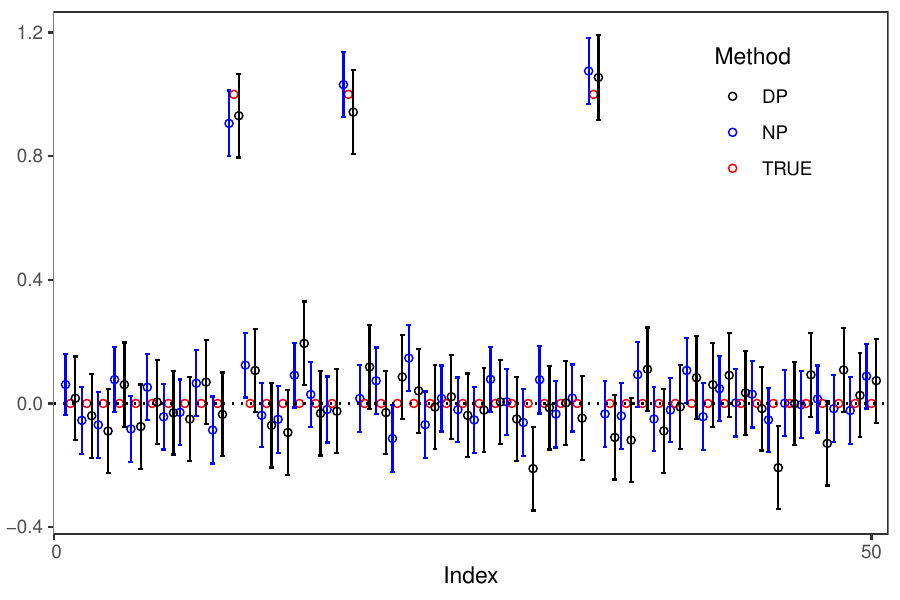}
\caption{ 95\% confidence intervals for one realization of DP correction under the AR covariance structure with $\rho=0.2$ for the first $50$ regression parameters. The true parameters are denoted in red, and the debiased estimators are denoted in black.}
\label{fig:1}
\end{figure}

To demonstrate the effectiveness of the proposed debiased procedure, we randomly select the active set from the first $50$ coordinates and present the estimated confidence intervals for these coordinates in one particular realization in Figure \ref{fig:1}. Notably, the estimated confidence intervals cover all signals, which correspond to the first three coordinates. The overall coverage for the coordinates shown in the figure is approximately $95\%$. Additional results are reported in Table \ref{tab:1} for the Toeplitz covariance design and in Table \ref{tab:2} for the equicorrelated design.

\begin{table}[ht]
\centering
\begin{tabular}{llcccc}
	\hline
	Measure&Method & $\rho=0.0$ & $\rho=0.2$ & $\rho=0.4$ & $\rho=0.6$  \\ 
	\hline
	\multirow{3}{*}{Avgcov}&DB-Lasso &0.950 & 0.954 & 0.959 & 0.966 \\
	&DP naive&0.823 & 0.829 & 0.845 & 0.896 \\ 
	&DP correction &0.951 & 0.951 & 0.949 & 0.964 \\ \hline
	\multirow{3}{*}{Avglength}&DB-Lasso &0.087 & 0.089 & 0.098 & 0.117\\ 
	&DP naive&0.087 & 0.089 & 0.097 & 0.112\\ 
	&DP correction & 0.126 & 0.127 & 0.133 & 0.145\\ \hline
\end{tabular}
\caption{Average coverage and length of the 95\% confidence interval under the Toeplitz covariance matrix.}
\label{tab:1}
\end{table}

\begin{table}[ht]
\centering
\begin{tabular}{llcccc}
	\hline
	Measure&Method & $\rho=0.05$ & $\rho=0.10$ & $\rho=0.15$ & $\rho=0.20$  \\ 
	\hline
	\multirow{3}{*}{Avgcov}&DB-Lasso &0.958 & 0.960 & 0.960 & 0.961\\
	&DP naive & 0.843 & 0.826 & 0.824 & 0.820 \\
	&DP correction &0.949 & 0.950 & 0.939 & 0.945 \\
	\hline
	\multirow{3}{*}{Avglength}&DB-Lasso &0.089 & 0.091 & 0.093 & 0.096 \\ 
	&DP naive &0.096 & 0.088 & 0.096 & 0.089 \\ 
	&DP correction &0.132 & 0.127 & 0.132 & 0.127 \\   \hline
\end{tabular}
\caption{Average coverage and length of the 95\% confidence interval under the blocked equal covariance matrix. }
\label{tab:2}
\end{table}

The numeric performance of our proposed debiased procedure exhibits remarkable similarity between the Toeplitz covariance design (Table \ref{tab:1}) and the equal correlation design (Table \ref{tab:2}). Notably, the coverage rates for DP naive fall significantly below the $95\%$ benchmark, empirically confirming our intuition that additional correction is necessary for finite samples, as discussed in Section \ref{sec: inference}. In contrast, the DP correction method achieves substantially improved coverage compared to DP naive, albeit with wider confidence intervals. The corrected confidence intervals are approximately $30\%$ wider than those of DP naive. The interval length for DP correction is approximately $30\%$ greater than that of DB-Lasso, reflecting the efficiency loss introduced by privacy constraints. Overall, the proposed method exhibits coverage rates of roughly $95\%$ with only a marginal reduction in efficiency.

\subsubsection{FDR Control}

Next, we evaluate the algorithm's performance in controlling the FDR. To assess its effectiveness, we consider the Toeplitz covariance matrices. The active set, denoted as $S_0$, consists of $|S_0|=30$ covariates randomly chosen from the full set of covariates. The nonzero regression coefficients $\beta_j$ for $j\in S_0$ are independently sampled from a normal distribution with mean zero and standard deviation $\xi$, where $\xi$ represents the signal strength. The errors in the linear model are assumed to follow $N(0, 1)$. The sample size is set to $n=10,000$, and the number of covariates is $p=10,000$. The privacy parameters are set to $\varepsilon=4$ and $\delta=1/n^{1.1}$, and the target FDR control level is $q=0.1$. Equal-sized data splitting is used.

We compare our method with the non-private FDR control algorithm presented in \cite{dai2022false}. The empirical FDR and power are reported, and all results are based on 100 independent simulations of the model with a fixed design and random regression coefficients. Figure \ref{fig:2} presents the empirical FDR and power across various signal levels. Both the proposed DP-FDR control procedure and the non-private procedure effectively control the empirical FDR at the predetermined level of $q=0.1$. The power of the proposed method exhibits a minor reduction compared to the non-private procedure due to privacy constraints. For reasonably large sample sizes, the proposed algorithm can maintain FDR control with a slight sacrifice in power compared to the non-private approach.

\begin{figure}[ht]
\centering
\includegraphics[width=0.7\linewidth]{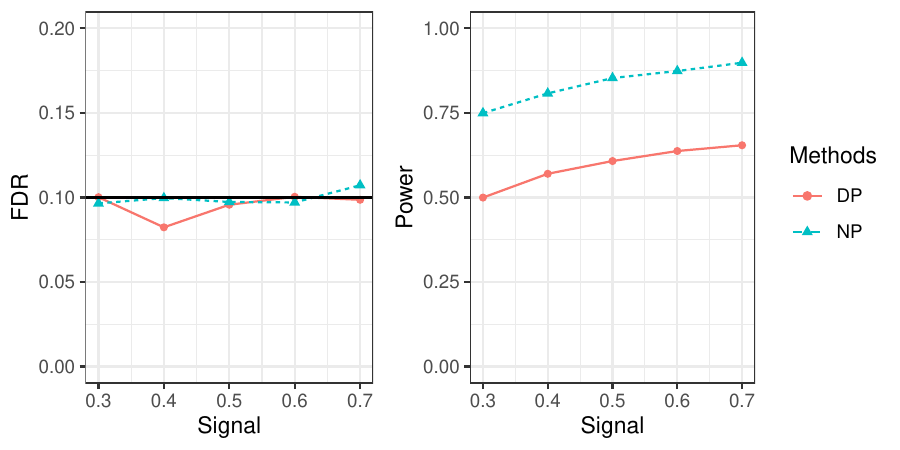}
\caption{Empirical FDRs and powers of Algorithm \ref{alg:fdr} (DP) and the non-private algorithm (NP) with increasing signals $\xi$ for $\rho=0.2$ and $n=10,000$.}
\label{fig:2}
\end{figure}

\begin{figure}[ht]
\centering
\includegraphics[width=0.7\linewidth]{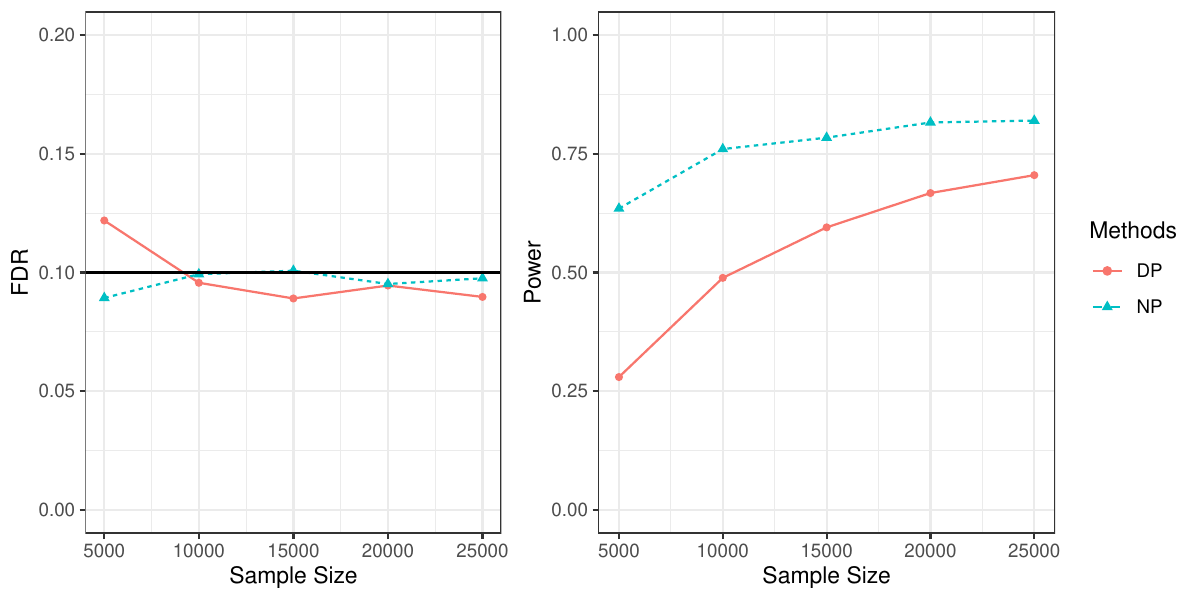}
\caption{Empirical FDRs and powers of Algorithm \ref{alg:fdr} (DP) and the non-private algorithm (NP) with increasing sample sizes $n$ for $\rho=0.2$ and $\xi = 0.3$.}
\label{fig:3}
\end{figure}

Figure \ref{fig:3} presents the empirical FDR and power across increasing sample sizes. It is important to note that the proposed procedure may fail to control the empirical FDR when the sample size is very small. This is primarily because, in cases of small sample sizes, the initial step involving DP-SLR may not accurately identify all active features. Additionally, there may be a nontrivial bias in the second step, which is the DP-OLS estimation. However, as the sample size increases, the proposed method successfully controls the empirical FDR at the predetermined level of $q=0.1$. Similarly, for small sample sizes, the power of the proposed procedure is notably lower than that of the non-private procedure due to reduced estimation accuracy caused by privacy constraints. Nevertheless, as the sample size grows, both the differentially private algorithm and the non-private algorithm exhibit increased power, and the difference between them diminishes. This improvement is due to increased estimation accuracy in both the DP and non-private algorithms. Overall, our numerical study demonstrates that for reasonably large sample sizes, the proposed algorithm can effectively maintain FDR control with a slight reduction in power compared to the non-private algorithm.

\subsection{Real Example: Soil Erosion in National Resources Inventory}

In this section, we demonstrate the performance of the proposed differentially private algorithms in analyzing soil erosion using the National Resources Inventory (NRI) dataset. Due to legal requirements, we cannot report non-private estimators. An additional real data analysis is provided in the appendix, where we compare the proposed methods with non-private methods on a public dataset.

Protecting data privacy is crucial when analyzing the NRI dataset. The NRI collects longitudinal data on land use, land cover, and natural resource conditions on non-Federal lands in the United States \citep{nusser1997national}. The integrity and confidentiality of the data collection sites, which are selected using rigorous scientific sample survey methods, are paramount. According to policies set by the USDA (United States Department of Agriculture) and the NRCS (Natural Resources Conservation Service), the NRI program is conducted in a manner that ensures the confidentiality of information and restricts access to the locations of data collection sites. This includes keeping confidential the location coordinates, maps, photographs, observations of local conditions, and other materials collected for inventories, as they do not constitute public information and are intended solely for use in official inventory activities or as authorized by the Secretary of Agriculture. Furthermore, any NRI data that could reveal the identity of owners, operators, or the locations of data collection sites is strictly protected and not disclosed outside the USDA.

Soil erosion, a natural process influenced by both environmental factors and human activities, leads to runoff over less permeable sub-layers and causes indirect environmental harm. Estimating soil erosion by water is crucial because of its impact on agriculture, infrastructure, ecological sustainability, and water quality \citep{kim2005rapid}. A primary goal of the NRI is to estimate erosion reductions that may result from the implementation of conservation plans. An accurate soil erosion model plays an increasingly important role in the design and implementation of soil management and conservation strategies \citep{panagos2015tackling}. Our focus is on developing a soil erosion model using the NRI dataset by identifying key features and providing valid confidence intervals within the framework of differential privacy.

The dataset comprises sampled locations from the state of Kansas, collected in 2017. The original dataset contains $40,475$ observations, including both real observations and imputed points, with $636$ covariates describing various land features and sample indicators. The dataset is pre-processed by focusing on core points consistently observed in every survey year and by removing sample indicators from the covariates. This processing method is widely adopted within the NRI to ensure sample reliability. After processing, the dataset contains $n=2100$ observations with $p=474$ covariates. The response variable $Y$ is the long-term average annual soil loss.

We first evaluate the performance of the DP-FDR control algorithm on the NRI data. The privacy parameters are set to $\epsilon=4$ and $\delta=1/n^{1.1}$. Equal-sized data splitting is used. When the FDR is controlled at $q=0.1$, the selected variables are presented in Table \ref{tab:3}.

\begin{table}[ht]
\centering
\small
\begin{tabular}{rrrrr}
	\hline
	WCFact & KWFact & TFact & IFact  \\ 
	\hline
	USLE1 & USLE2&USLE3&USLE4  \\ 
	\hline
\end{tabular}
\caption{Selected Real Feature by Algorithm \ref{alg:fdr} with FDR control at $q=0.1$}
\label{tab:3}
\end{table}

The proposed method selected a reasonable subset of important features. For instance, KWFact denotes the soil erodibility factor in the Universal Soil Loss Equation (USLE), and TFact signifies soil loss tolerance, indicating the acceptable level of annual soil loss in tons per acre. These two covariates are known to be highly correlated with soil erosion \citep{alewell2019using} and are selected in both steps of the procedure. WCFact represents the climatic factor in the Wind Erosion Equation (WEQ), which is directly associated with the wind erosion model and is typically not incorporated into water erosion models. However, as highlighted by \cite{nearing2004expected}, the dynamics of how climate change influences soil erosion by water are multifaceted. For example, rainfall patterns may vary in volume and intensity, frequency of precipitation days, and proportion of rain to snow. These variations affect plant biomass production, the rate of plant residue decomposition, soil microbial activity, and evapotranspiration. Thus, it is reasonable to incorporate climatic factors into the model. Our procedure also selected IFact, the soil erodibility index, which appears in the WEQ model. WEQ is an empirical modeling procedure used to estimate soil loss caused by wind erosion from agricultural fields and has become the most comprehensive and widely used model for this purpose. Since water erosion and wind erosion compete with one another, it is reasonable to expect that increasing wind erosion reduces water erosion. Additionally, USLE1, USLE2, USLE3, and USLE4 function as polynomial expressions of various USLE factors, including rainfall, soil erodibility, cover and management, support practices, slope length, and slope percentage. These factors are integral to the USLE model and were used to predict soil erosion in previous NRI studies. It is therefore consistent that our methods selected these four USLE variables. Overall, the proposed procedure successfully identified critical features from the prior NRI soil loss model while also incorporating additional variables that significantly affect water erosion but have not yet been considered in the current NRI project.

To evaluate the performance of Algorithm \ref{algo:inference}, we report the $95\%$ confidence intervals of the variables selected in the DP-FDR control step in Table \ref{tab:4}. All covariates, except IFact, are significant at the $95\%$ confidence level. The parameter associated with WCFact has a negative sign, reflecting the positive correlation between climatic factors and vegetation cover, which in turn leads to a negative correlation with soil erosion. The parameters associated with KWFact and TFact are positive: KWFact reflects soil erodibility, while TFact represents soil loss tolerance. Both are determined using expert knowledge and historical information, and are directly linked to soil erosion. The parameter associated with IFact also has a negative sign, consistent with the competitive relationship between water and wind erosion. However, its coefficient is not significant at the $95\%$ level, suggesting that the competition effect is weak. Table \ref{tab:4} can be used to forecast soil loss under specific soil conditions, making it a valuable tool for developing conservation strategies and crop management plans.

\begin{table}[ht]
\centering
\begin{tabular}{r|rrr}
	\hline
	Feature & Parameter &Lower bound& Upper bound \\ 
	\hline
	WCFact&-0.177 & -0.246 & -0.109 \\ 
	KWFact & 0.028 & 0.006 & 0.051 \\ 
	TFact & 0.058 & 0.036 & 0.081 \\
	IFact & -0.022 & -0.057 & 0.012 \\ 
	USLE1 & -0.134 & -0.157 & -0.112 \\
	USLE2 & 0.029 & 0.006 & 0.051 \\ 
	USLE3 & -0.043 & -0.066 & -0.021 \\ 
	USLE4 &  -0.030 & -0.052 & -0.007 \\ 
	\hline
\end{tabular}
\caption{The $95\%$ confidence intervals for selected features in Table \ref{tab:3} by proposed Algorithm \ref{algo:inference} with finite sample correction.}
\label{tab:4}
\end{table}

\section{Discussion}

This paper presents a comprehensive framework for conducting differentially private analysis in high-dimensional linear models, encompassing estimation, inference, and false discovery rate (FDR) control. The framework is particularly valuable in scenarios where individual privacy in the dataset must be protected and can be readily applied across various disciplines. The numerical studies conducted in this work demonstrate that privacy protection can be achieved with only a minor loss in the accuracy of confidence intervals and multiple testing.

We briefly discuss several possible extensions. For example, the tools developed for DP estimation, the debiased Lasso, and FDR control in this paper can be extended to generalized linear models. It would also be interesting to explore scenarios where part of a dataset—potentially following a different distribution—is publicly available and not subject to privacy constraints. In addition, the newly developed DP-BIC could be adapted for other tasks involving the selection of tuning parameters with privacy guarantees. These directions are left for future research.

\acks{The authors thank the anonymous reviewers for their valuable suggestions. Zhanrui Cai was supported in part by the Hong Kong Research Grants Council (Grant No.27301925) and the National Natural Science Foundation of China (Grant No.12501386). Sai Li was supported by the National Natural Science Foundation of China (No. 12571314). Linjun Zhang was supported in part
by NSF DMS-2015378 and NSF CAREER DMS-2340241. 
}



\appendix

\section{Proofs}

\subsection{The convergence rate in Algorithm \ref{alg:linear}}

We first establish the privacy guarantee and derive the convergence rate of $\hat{\boldsymbol{\beta}}$ in Algorithm \ref{alg:linear}.

\begin{proof}[Proof of Lemma \ref{lem:privacy-adapt}]
	
For $0 \leq k \leq K$, the $\ell_{\infty}$ sensitivity of the gradient at the $t$-th iteration, given by $-\eta_0/|S_t|\sum_{i \in S_t}(\Pi_R(\boldsymbol{x}_i^{\top}\boldsymbol{\beta}_k^{(t)}) - \Pi_R(y_i))\boldsymbol{x}_i$ as defined in line 9 of Algorithm \ref{alg:linear}, satisfies:
\begin{equation*}
\begin{split}
&\sup_{(\boldsymbol{x}_i,y_i),(\boldsymbol{x}_i^{\prime},y_i^{\prime})}\eta_0/|S_t|\cdot\|(\Pi_R(\boldsymbol{x}_i^{\top}\boldsymbol{\beta}^{(t)}_k)-\Pi_R(y_i))\boldsymbol{x}_i-(\Pi_R(\boldsymbol{x}_i^{\prime\top}\boldsymbol{\beta}^{(t)}_k)-\Pi_R(y_i^\prime))\boldsymbol{x}_i^\prime\|_{\infty}\\
&\leq \eta_0T/n\cdot 2(R+R)c_x,
\end{split}
\end{equation*}
where we use the fact that the sample size is $|S_t|=n/T$ and Condition \ref{cond:1}, which assumes that $\|\boldsymbol{x}_i\|_{\infty}$ is bounded by $c_x$. By the Gaussian mechanism (Lemma \ref{lemma:laplace_gaussian_mechanism}) and the advanced composition theorem (Lemma \ref{lem:post_combin_dp}), reporting the gradient in line 9 of Algorithm \ref{alg:linear} is $(\varepsilon/\{T(K+2)\},\delta/\{T(K+1)\})$-DP. Thus, by the composition theorem (Lemma \ref{lem:post_combin_dp}), for $0\leq k\leq K$, the output $\hat{\boldsymbol{\beta}}(k)$ is $(\varepsilon/(K+2),\delta/(K+1))$-DP. Finally, by applying the composition theorem, releasing all $\{\hat{\boldsymbol{\beta}}(k)\}_{k=0}^{K}$ is $(\varepsilon(K+1)/(K+2),\delta)$-DP.

Next, we consider the sensitivity of the BIC loss. Note that
\begin{equation*}
\begin{split}			&\sup_{(\boldsymbol{x}_i,y_i),(\boldsymbol{x}_i^{\prime},y_i^{\prime})}|(\Pi_R(\boldsymbol{x}_i^{\top}\boldsymbol{\beta}(k))-\Pi_R(y_i))^2-(\Pi_R(\boldsymbol{x}_i^{\prime\top}\boldsymbol{\beta}(k))-\Pi_R(y_i^\prime))^2|\leq 2(2R)^2,
\end{split}
\end{equation*}
for every $0\leq k\leq K$. The BIC selection procedure returns the noisy minimizer. By Claim 3.9 in \cite{dwork2014algorithmic}, the BIC selection procedure is $(\varepsilon/(K+2),0)$-DP. Finally, by the composition theorem, the output of Algorithm \ref{alg:linear} is $(\varepsilon,\delta)$-DP.
\end{proof}

\begin{proof}[Proof of Theorem \ref{thm-adapt}]
Let $\hat{k}$ be the selected number corresponding to the selected model $\hat{\boldsymbol{\beta}}$ in Algorithm \ref{alg:linear}, and let $k^*$ denote the true index such that $2^{k^*-1}\leq\rho L^4 s \leq 2^{k^*}$. By the condition $2^K>\rho L^4s$ stated in Theorem \ref{thm-adapt}, the true parameter $k^*$ satisfies $k^*<K$. Note that $k^*$ is uniquely determined by $s$ and $(\rho,L)$. For two sequences of positive integers, $\{a_n\}_{n=1}^{\infty}, \{b_n\}_{n=1}^{\infty}$, the notation $a_n=o(b_n)$ means that $\lim_{n\to\infty}a_n/b_n=0$.

Define the event
\begin{equation*}
\begin{split}
E_0:=\big\{&\inf_{\|\bu\|_0=o(n),\|\bu\|_2=1} {\bu}^{\top}\widehat{\bSigma}{\bu}\geq c_{\gamma_l}\|\bu\|_2^2,\sup_{\|\bu\|_0=o(n),\|\bu\|_2=1} {\bu}^{\top}\widehat{\bSigma}{\bu}\leq c_{\gamma_u}\|\bu\|_2^2\big\},
\end{split}
\end{equation*}
where $\widehat{\bSigma}=\sum_{i=1}^{n}\boldsymbol{x}_i^{\top}\boldsymbol{x}_i/n$, and $c_{\gamma_l},c_{\gamma_u}$ are positive constants depending only on the eigenvalues of the population covariance matrix $\boldsymbol{\Sigma}$. The event $E_0$ provides lower and upper bounds on the sparse eigenvalues. It is closely related to the well-known restricted eigenvalue conditions essential in high-dimensional linear regression. By Theorem 16 in \cite{rudelson2012reconstruction}, the event $E_0$ holds with probability at least $1-\exp(-C_0 n)$ for a positive constant $C_0$. Define the event under which the truncation operators do not take effect to be
\begin{equation*}
\begin{split}
E_1:=&\big\{\max_{i=1,\dots,n}|y_i|\leq R,\max_{t=0,\dots,T-1;k=0,\dots,K}|\boldsymbol{x}_i^{\top}\boldsymbol{\beta}_k^{(t)}|\leq R\text{ for all }i\in \mathcal{S}_t\big\}.
\end{split}
\end{equation*}
We use $\|\cdot\|_{\psi_2}$ to denote the sub-Gaussian norm and $\|\cdot\|_{\psi_1}$ to denote the sub-exponential norm, respectively. By Condition \ref{cond:1} and the independence between $\boldsymbol{x}_i$ and $\boldsymbol{\beta}_k^{(t)}$ due to data splitting, we apply the Chernoff bound to obtain the following large deviation result:
$$\mathbb{P}(|\boldsymbol{x}_i^{\top}\boldsymbol{\beta}_k^{(t)}|\geq R)\leq 2\exp\{-cR^2/(C^2\|\boldsymbol{x}_i\|_{\psi_2}^2)\},$$
where we use the fact that $\boldsymbol{x}_i^{\top}\boldsymbol{\beta}_k^{(t)}$ is sub-Gaussian with its sub-Gaussian norm bounded by $C\|\boldsymbol{x}_i\|_{\psi_2}$, and $c$ is an absolute constant. The use of $c$ is standard in the high-dimensional statistics literature; see, for example, Theorem 2.6.2 in \cite{vershynin2010introduction}. By definition, the sub-Gaussian norm of $\boldsymbol{x}_i^{\top}\boldsymbol{\beta}_k^{(t)}$ is bounded above by $\|\boldsymbol{\beta}_k^{(t)}\|_2\|\boldsymbol{x}_i\|_{\psi_2}\leq C\|\boldsymbol{x}_i\|_{\psi_2}$, where we use the assumptions that $\boldsymbol{x}_i$ is a sub-Gaussian random vector and that $\|\boldsymbol{\beta}_k^{(t)}\|_2\leq C$ due to truncation. Furthermore, by Condition \ref{cond:2} and the linear model assumption, the response variable $y_i$ is also sub-Gaussian, with its sub-Gaussian norm bounded by $c\sqrt{c_0\|\boldsymbol{x}_i\|_{\psi_2}^2+\|e_i\|^2_{\psi_2}}$, where $c$ is an absolute constant. Note that the event $E_1$ is the intersection of $n+n(K+1)$ sub-events. A union bound for the probability of $E_1$ can be obtained using the inequality $(1-p_1)\times(1-p_2)\times\dots\times(1-p_m)\geq 1-p_1-\dots-p_m$. Thus, we have
\begin{equation*}
\begin{split}
\mathbb{P}(E_1)&\geq 1-\sum_{i=1}^{n}\mathbb{P}(|y_i|\geq R)-\sum_{k=0}^{K}\sum_{t=0}^{T-1}\sum_{i\in \mathcal{S}_t}\mathbb{P}(|\boldsymbol{x}_i^{\top}\boldsymbol{\beta}_k^{(t)}|\geq R)\\
&\geq 1-2n(K+2)\exp(-cR^2\max\{C^2\|\boldsymbol{x}_i\|_{\psi_2}^2,c_0\|\boldsymbol{x}_i\|_{\psi_2}^2+\|e_i\|^2_{\psi_2}\}),
\end{split}
\end{equation*}
where the first inequality follows from applying the Chernoff bound $n(K+2)$ times. By choosing $R\geq\sqrt{2\max\{C^2\|\boldsymbol{x}_i\|_{\psi_2}^2,c_0\|\boldsymbol{x}_i\|_{\psi_2}^2+\|e_i\|^2_{\psi_2}\}\log(n)/c}$, we have $\mathbb{P}(E_1)\geq 1-2n(K+2)\exp\{-2\log(n)\}=1-2(K+2)\exp\{-\log(n)\}\stackrel{n\to\infty}{\to}1$, where we use the assumption that $K=O\big(\log(n)\big)$. It remains to analyze the convergence of the differentially private sparse linear regression. Define the event
\begin{equation*}
\begin{split}
E_2&=\bigg\{\|\boldsymbol{\beta}_{k}^{(T)}-\boldsymbol{\beta}\|^2_2\leq c_2^{\prime}\frac{2^{k}\log (p)\log (n)}{n}+c_3^{\prime}\frac{2^{2k}\log (p)^2\log(1/\delta)\log(n)^7}{n^2\varepsilon^2}\text{ for }2^k\geq\rho L^4 s\bigg\},
\end{split}
\end{equation*}
where $E_2$ captures the event that, for a given $k$, the corresponding estimator $\boldsymbol{\beta}_{k}^{(T)}$ achieves the convergence rate stated in Theorem 4.4 of \cite{cai2019cost}. Under the event $E_1$, and by Theorem 4.4 in \cite{cai2019cost}—with $T$ therein replaced by $KT$—the event $E_2$ holds for a given $k$ with probability at least $1-\exp\{-c_1^{\prime}\log (n)\}$, for some positive constants $c_1^{\prime},c_2^{\prime},c_3^{\prime}$. By applying the union bound, we conclude that the event $E_2$ holds with probability at least $1-(K+1)\exp\{-c_1^{\prime}\log (n)\}$. We now analyze the theoretical performance of the proposed BIC method under the event $E_0\cap E_1\cap E_2$.
	
Note that under the event $E_1\cap E_2$, for $2^k\geq\rho L^4 s$, we have
\begin{equation*}
\begin{split}
|\bx_i^{\top}\boldsymbol{\beta}_{k}^{(T)}|&\leq|\bx_i^{\top}\boldsymbol{\beta}|+|\bx_i^{\top}(\boldsymbol{\beta}-\boldsymbol{\beta}_{k}^{(T)})|\leq|\bx_i^{\top}\boldsymbol{\beta}|+\|\bx_i\|_{\infty}\|\boldsymbol{\beta}-\boldsymbol{\beta}_{k}^{(T)}\|_{1}\\
&\leq |\bx_i^{\top}\boldsymbol{\beta}| + c_x\sqrt{2^{k}}\|\boldsymbol{\beta}-\boldsymbol{\beta}_{k}^{(T)}\|_{2},
\end{split}
\end{equation*}
where the first inequality follows from the triangle inequality, the second inequality follows from Hölder's inequality, and the last inequality uses the bound $\|\cdot\|_1\leq\sqrt{\|\cdot\|_0}\times\|\cdot\|_2$. By the assumptions in Theorem \ref{thm-adapt}, we have 
\begin{equation*}
    \begin{split}
        c_x\|\boldsymbol{\beta}-\boldsymbol{\beta}_{k}^{(T)}\|_{1}&\leq c_x\sqrt{c_2^{\prime}\frac{2^{2k}\log (p)\log (n)}{n}+c_3^{\prime}\frac{2^{3k}\log (p)^2\log(1/\delta)\log(n)^7}{n^2\varepsilon^2}}\\
        &\leq c_x\sqrt{ c_2^{\prime}\frac{2^{2K}\log (p)\log (n)}{n}+c_3^{\prime}\frac{2^{3K}\log (p)^2\log(1/\delta)\log(n)^7}{n^2\varepsilon^2}}\\
        &=O\bigg(\sqrt{\frac{n\log(p)\log(n)}{n\log(p)^4}+\frac{n^{3/2}\log (p)^2\log(1/\delta)\log(n)^7}{n^2\varepsilon^2\log (p)^6}}\bigg)\\
        &=O\bigg(\sqrt{\frac{1}{\log(p)^2}+\frac{\log(1/\delta)\log(n)^3}{n^{1/2}\varepsilon^2}}\bigg)=o(1)
    \end{split}
\end{equation*}
and thus, for a proper choice of $R$, the parameter clipping does not occur for $2^k\geq\rho L^4s$. In the remainder of the proof for the BIC criterion, we use $\boldsymbol{y}$ and $\boldsymbol{X}$ to denote the vector $(y_1,\dots,y_n)^{\top}$ and the matrix $(\bx_1,\dots,\bx_n)^{\top}$, respectively. By the oracle inequality for the BIC criterion, we obtain the following expression, which is a direct consequence of the selection procedure:
\[
\|\boldsymbol{y}-\boldsymbol{X}\hat{\boldsymbol{\beta}}\|_2^2+c_Bf(n,\hat{k})+z_{\hat{k}}\leq \|\boldsymbol{y}-\boldsymbol{X}\hat{\boldsymbol{\beta}}(k^*)\|_2^2+c_Bf(n,k^*)+z_{k^*},
\]
where the function $f(n,k)=2^k\log (p)\log(n)+\{2^{2k}\log (p)^2\log(1/\delta)\log(n)^7\}/(n\varepsilon^2)$, and $z_{\hat{k}},z_{k^*}$ are the added noise terms due to privacy. Furthermore, by taking the maximum of the additional noise terms, we have:
\[
\|\boldsymbol{y}-\boldsymbol{X}\hat{\boldsymbol{\beta}}\|_2^2+c_Bf(n,\hat{k})\leq \|\boldsymbol{y}-\boldsymbol{X}\hat{\boldsymbol{\beta}}(k^*)\|_2^2+c_Bf(n,k^*)+\epsilon_{privacy},
\]
where $\epsilon_{privacy}$ is defined as $2\sup_{k=0,\dots,K}|z_k|$. The above inequality implies that
\begin{equation}
\begin{split}
\label{eq:thm1_1}			\|\boldsymbol{X}(\hat{\boldsymbol{\beta}}-\hat{\boldsymbol{\beta}}(k^*))\|_2^2&\leq 2|\langle \boldsymbol{X}(\hat{\boldsymbol{\beta}}-\hat{\boldsymbol{\beta}}(k^*)),\boldsymbol{y}-\boldsymbol{X} \hat{\boldsymbol{\beta}}(k^*)\rangle|+c_B\{f(n,k^*)-f(n,\hat{k})\}+\epsilon_{privacy}.
\end{split}
\end{equation}
Let the support set be $\widehat{U}=\text{supp}(\hat{\boldsymbol{\beta}}-\hat{\boldsymbol{\beta}}(k^*))$. Note that
\[
|\widehat{U}|=O(\sqrt{n}/\log(p)^2+s)=o(n).
\]
Hence, under the event $E_0$, the inequality,
\begin{align}
& c_{\gamma_l}\|\hat{\boldsymbol{\beta}}-\hat{\boldsymbol{\beta}}(k^*)\|_2^2\leq\frac{1}{n}\|\boldsymbol{X}(\hat{\boldsymbol{\beta}}-\hat{\boldsymbol{\beta}}(k^*))\|_2^2\nonumber\\
\leq& \frac{2}{n}|\langle \boldsymbol{X}(\hat{\boldsymbol{\beta}}-\hat{\boldsymbol{\beta}}(k^*)),\boldsymbol{X} (\hat{\boldsymbol{\beta}}(k^*)-\boldsymbol{\beta})\rangle|+\frac{2}{n}|\langle \boldsymbol{X}(\hat{\boldsymbol{\beta}}-\hat{\boldsymbol{\beta}}(k^*)),\boldsymbol{y}-\boldsymbol{X}\boldsymbol{\beta}\rangle|\nonumber\\
&+c_B\{f(n,k^*)-f(n,\hat{k})\}/n+\epsilon_{privacy}/n \nonumber\\
\leq& \frac{2}{n}\|\boldsymbol{X}(\hat{\boldsymbol{\beta}}-\hat{\boldsymbol{\beta}}(k^*))\|_2\|\boldsymbol{X} (\hat{\boldsymbol{\beta}}(k^*)-\boldsymbol{\beta})\|_2+2\|\hat{\boldsymbol{\beta}}-\hat{\boldsymbol{\beta}}(k^*)\|_1\|\frac{1}{n}\boldsymbol{X}^{\top}\boldsymbol{e}\|_{\infty}\nonumber\\
&+c_B\{f(n,k^*)-f(n,\hat{k})\}/n+\epsilon_{privacy}/n \nonumber\\
\leq& 2c_{\gamma_u}\|\hat{\boldsymbol{\beta}}-\hat{\boldsymbol{\beta}}(k^*)\|_2\|\hat{\boldsymbol{\beta}}(k^*)-\boldsymbol{\beta}\|_2+\|\hat{\boldsymbol{\beta}}-\hat{\boldsymbol{\beta}}(k^*)\|_1\sqrt{2\|\boldsymbol{x}_ie_i\|^2_{\psi_1}/c\frac{\log (p)}{n}}\nonumber\\
&+c_B\{f(n,k^*)-f(n,\hat{k})\}/n+\epsilon_{privacy}/n\label{eq:thm1_bic}
\end{align}
holds with probability at least $1-2\exp(-\log(p))$, where the second inequality follows from the relationship in \eqref{eq:thm1_1}, the third inequality follows from Hölder's inequality, and the last inequality follows from the event $E_0$ and a concentration inequality. In this expression, we use $\boldsymbol{e}=(e_1,\dots,e_n)^{\top}$ to denote the vector of random errors in the linear model. Note that each component of $\boldsymbol{x}_ie_i$ is a product of two sub-Gaussian random variables, and is therefore sub-exponential. Thus, we have
\begin{equation*}
\begin{split}
&\mathbb{P}\bigg(\|\frac{1}{n}\boldsymbol{X}^{\top}\boldsymbol{e}\|_{\infty}\geq\sqrt{2\|\boldsymbol{x}_ie_i\|^2_{\psi_1}/c\frac{\log(p)}{n}}\bigg)\\
&\leq\sum_{j=1}^{p}\mathbb{P}\bigg(|\frac{1}{n}\sum_{i=1}^{n}x_{i,j}e_i|\geq\sqrt{2\|\boldsymbol{x}_ie_i\|^2_{\psi_1}/c\frac{\log (p)}{n}}\bigg)\\
&\leq 2p\exp\bigg(-c\frac{2\|\boldsymbol{x}_ie_i\|^2_{\psi_1}\log(p)}{c\|\boldsymbol{x}_ie_i\|^2_{\psi_1}}\bigg)=2p\exp(-2\log(p))=2\exp(-\log(p)),
\end{split}
\end{equation*}
where we use the union bound in the first inequality and Bernstein’s inequality in the second inequality.
	
We first consider the case where $\hat{k}< k^*$. We obtain the following inequality:
\begin{equation*}
\begin{split}
&c_{\gamma_l}\|\hat{\boldsymbol{\beta}}-\hat{\boldsymbol{\beta}}(k^*)\|_2^2\leq 2c_{\gamma_u}\|\hat{\boldsymbol{\beta}}-\hat{\boldsymbol{\beta}}(k^*)\|_2\|\hat{\boldsymbol{\beta}}(k^*)-\boldsymbol{\beta}\|_2\\
&+\|\hat{\boldsymbol{\beta}}-\hat{\boldsymbol{\beta}}(k^*)\|_2\sqrt{\frac{(2^{\hat{k}}+2^{k^*})\log(p)}{n}}\sqrt{2\|\boldsymbol{x}_ie_i\|^2_{\psi_1}/c}\\
&+c_B\{f(n,k^*)-f(n,\hat{k})\}/n+\epsilon_{privacy}/n,
\end{split}
\end{equation*}
by applying Hölder's inequality to \eqref{eq:thm1_bic},
$$\|\hat{\boldsymbol{\beta}}-\hat{\boldsymbol{\beta}}(k^*)\|_1\leq\|\hat{\boldsymbol{\beta}}-\hat{\boldsymbol{\beta}}(k^*)\|_2\sqrt{\|\hat{\boldsymbol{\beta}}-\hat{\boldsymbol{\beta}}(k^*)\|_0},$$
and we use the fact that $\|\hat{\boldsymbol{\beta}}-\hat{\boldsymbol{\beta}}(k^*)\|_0\leq 2^{\hat{k}}+2^{k^*}$. By treating $\|\hat{\boldsymbol{\beta}}-\hat{\boldsymbol{\beta}}(k^*)\|_2:=t$ as an unknown variable, the preceding expression becomes a quadratic function in $t$. To simplify the notation, we define $$a_1=2c_{\gamma_u}/c_{\gamma_l}\|\hat{\boldsymbol{\beta}}(k^*)-\boldsymbol{\beta}\|_2+1/c_{\gamma_l}\sqrt{\frac{(2^{\hat{k}}+2^{k^*})\log(p)}{n}}\sqrt{2\|\boldsymbol{x}_ie_i\|^2_{\psi_1}/c}$$ 
and
$$a_2=c_B/c_{\gamma_l}\{f(n,k^*)-f(n,\hat{k})\}/n+1/c_{\gamma_l}\times\epsilon_{privacy}/n.$$
Then we have the inequality $t^2-a_1t-a_2\leq 0$. By the assumption that $\hat{k}< k^*$ and $c_B>0$, it follows that $c_B\{f(n,k^*)-f(n,\hat{k})\}+\epsilon_{privacy}>0$. Therefore, the solution to the quadratic inequality exists and satisfies $t\leq a_1/2+\sqrt{a_1^2/4+a_2}$. Furthermore, by the inequality $\sqrt{a+b}\leq\sqrt{a}+\sqrt{b}$ for $a,b\geq 0$, we obtain $t\leq a_1+\sqrt{a_2}$. Finally, by the event $E_2$ and the fact that $2^{\hat{k}}+2^{k^*}\leq 2^{k^*+1}\leq 2\rho L^4s$, we have
\begin{equation*}
\begin{split}
a_1\leq& \frac{2c_{\gamma_u}}{c_{\gamma_l}}\sqrt{c_2^{\prime}\frac{2^{k^*}\log (p)\log (n)}{n}+c_3^{\prime}\frac{2^{2k^*}\log (p)^2\log(1/\delta)\log(n)^7}{n^2\epsilon^2}}\\
&+1/c_{\gamma_l}\sqrt{4\|\boldsymbol{x}_ie_i\|^2_{\psi_1}/c}\sqrt{\frac{ 2^{k^*}\log(p)}{n}}.
\end{split}
\end{equation*}
It remains to consider the term $\sqrt{a_2}$. Since the distribution of $z_i$ is Laplace, it is sub-exponential. We have
\begin{equation*}
\begin{split}
\mathbb{P}\bigg\{\epsilon_{privacy}\geq 4c\log(n)\frac{2(2R)^2(K+2)}{\varepsilon}\bigg\}&\leq\sum_{i=0}^{K}\mathbb{P}\bigg\{|z_i|\geq 4c\log(n)\frac{2(2R)^2(K+2)}{\varepsilon}\bigg\}\\
&\leq(K+1)\exp\{-2\log(n)\}\leq\exp\{-\log(n)\}.
\end{split}
\end{equation*}
By the definition of $f(n,k)$, we have
\begin{equation*}
\begin{split}
a_2&\leq c_B/c_{\gamma_l}\{f(n,k^*)-f(n,\hat{k})\}/n+1/c_{\gamma_l}\epsilon_{privacy}/n\\
&\leq c_B/c_{\gamma_l}f(n,k^*)/n+2c\log(n)\frac{2(4R)^2(K+2)}{\varepsilon}/(c_{\gamma_l}n)\\
&\leq c_B/c_{\gamma_l}\bigg[2^{k^*}\log (p)+\frac{2^{2k^*}\log (p)^2\log(1/\delta)\log(n)^6}{n\varepsilon^2}\bigg]\frac{\log(n)}{n}\\
&+2c\log(n)\frac{2(4R)^2(K+2)}{\varepsilon}\frac{1}{c_{\gamma_l}n}.
\end{split}
\end{equation*}
By combining the upper bounds of $a_1^2$ and $a_2$, we have:
\begin{equation*}
\begin{split}
\|\hat{\boldsymbol{\beta}}-\hat{\boldsymbol{\beta}}(k^*)\|_2^2&\leq (a_1+\sqrt{a_2})^2\leq 2a_1^2+2a_2\\
&\leq c_2\frac{s\log(p)\log(n)}{n}+c_3\frac{s^2\log(p)^2\log(1/\delta)\log(n)^7}{n^2\varepsilon^2}+c_4\frac{\log(n)^3}{n\varepsilon},
\end{split}
\end{equation*}
for some constant $c_2,c_3,c_4$.
	
Next, we consider the case where $\hat{k}\geq k^*$. By applying the triangle inequality to \eqref{eq:thm1_bic}, we obtain:
\begin{align*}
c_{\gamma_l}\|\hat{\boldsymbol{\beta}}-\hat{\boldsymbol{\beta}}(k^*)\|_2^2
\leq& 2c_{\gamma_u}\|\hat{\boldsymbol{\beta}}-\hat{\boldsymbol{\beta}}(k^*)\|_2\|\hat{\boldsymbol{\beta}}(k^*)-\boldsymbol{\beta}\|_2+\|\hat{\boldsymbol{\beta}}-\hat{\boldsymbol{\beta}}(k^*)\|_1\sqrt{2\|\boldsymbol{x}_ie_i\|^2_{\psi_1}/c\frac{\log (p)}{n}}\\
&+c_4\{f(n,k^*)-f(n,\hat{k})\}/n+\epsilon_{privacy}/n\\
\leq& 2c_{\gamma_u}\|\hat{\boldsymbol{\beta}}-\hat{\boldsymbol{\beta}}(k^*)\|_2\|\hat{\boldsymbol{\beta}}(k^*)-\boldsymbol{\beta}\|_2+\|\hat{\boldsymbol{\beta}}-\boldsymbol{\beta}\|_1\sqrt{2\|\boldsymbol{x}_ie_i\|^2_{\psi_1}/c\frac{\log (p)}{n}}\\
&+\|\boldsymbol{\beta}-\hat{\boldsymbol{\beta}}(k^*)\|_1\sqrt{2\|\boldsymbol{x}_ie_i\|^2_{\psi_1}/c\frac{\log (p)}{n}}\\
&+c_B\{f(n,k^*)-f(n,\hat{k})\}/n+\epsilon_{privacy}/n.
\end{align*}
By treating $\|\hat{\boldsymbol{\beta}}-\hat{\boldsymbol{\beta}}(k^*)\|_2:=t$ as an unknown variable, the previous expression becomes a quadratic function in $t$. To simplify the notation, we define $a_1^{\prime}=2c_{\gamma_u}/c_{\gamma_l}\|\hat{\boldsymbol{\beta}}(k^*)-\boldsymbol{\beta}\|_2$ and $a_2^{\prime}=c_B/c_{\gamma_l}\{f(n,k^*)-f(n,\hat{k})\}/n+1/c_{\gamma_l}\epsilon_{privacy}/n+1/c_{\gamma_l}\|\hat{\boldsymbol{\beta}}-\boldsymbol{\beta}\|_1\sqrt{2\|\boldsymbol{x}_ie_i\|^2_{\psi_1}/c\frac{\log (p)}{n}}+1/c_{\gamma_l}\|\boldsymbol{\beta}-\hat{\boldsymbol{\beta}}(k^*)\|_1\sqrt{2\|\boldsymbol{x}_ie_i\|^2_{\psi_1}/c\frac{\log (p)}{n}}$. Under the event $E_2$, we have
\begin{equation*}
a_1^{\prime}\leq\frac{2c_{\gamma_u}}{c_{\gamma_l}}\sqrt{ c_2^{\prime}\frac{2^k\log (p)\log (n)}{n}+c_3^{\prime}\frac{2^{2k}\log (p)^2\log(1/\delta)\log(n)^7}{n^2\varepsilon^2}}.
\end{equation*}
By the inequality $\|\cdot\|_1\leq\|\cdot\|_2\times\sqrt{\|\cdot\|_0}$, we have 
\begin{equation*}
\begin{split}
a_2^{\prime}\leq& 2/c_{\gamma_l}\sqrt{2\frac{\|\boldsymbol{x}_ie_i\|^2_{\psi_1}}{c}}\sqrt{\frac{2^{\hat{k}}\log(p)}{n}}\\
&\times\sqrt{c_2^{\prime}\frac{2^{\hat{k}}\log (p)\log (n)}{n}+c_3^{\prime}\frac{2^{2\hat{k}}\log (p)^2\log(1/\delta)\log(n)^7}{n^2\epsilon^2}}\\
+&c_B/c_{\gamma_l}\{f(n,k^*)-f(n,\hat{k})\}/n+1/c_{\gamma_l}\epsilon_{privacy}/n.
\end{split}
\end{equation*}
For $c_B>2\sqrt{\max\{c_2^{\prime},c_3^{\prime}\}}\sqrt{2\|\boldsymbol{x}_ie_i\|^2_{\psi_1}/c}$, we have
\begin{equation*}
a_2^{\prime}\leq c_B/c_{\gamma_l}f(n,k^*)/n+1/c_{\gamma_l}\epsilon_{privacy}/n.
\end{equation*}
By properties of solutions to quadratic inequalities and the bound $t\leq a_1+\sqrt{a_2}$, we have
\[
\|\hat{\boldsymbol{\beta}}-\hat{\boldsymbol{\beta}}(k^*)\|_2^2\leq 2(\frac{2c_{\gamma_u}}{c_{\gamma_l}})^2 \|\hat{\boldsymbol{\beta}}(k^*)-\boldsymbol{\beta}\|_2^2+2a_2^{\prime}.
\]
Then, using the fact that $f(n,k^*)/n\leq\frac{1}{\max\{c_2^{\prime},c_3^{\prime}\}}\|\hat{\boldsymbol{\beta}}(k^*)-\boldsymbol{\beta}\|_2^2$ and applying the large deviation bound for $\epsilon_{privacy}$ as used in the bound for $a_2$, we have
\begin{equation*}
\begin{split}
\|\hat{\boldsymbol{\beta}}-\hat{\boldsymbol{\beta}}(k^*)\|_2^2\leq c_2\frac{s\log(p)\log(n)}{n}+c_3\frac{s^2\log(p)^2\log(1/\delta)\log(n)^7}{n^2\varepsilon^2}+c_4\frac{\log(n)^3}{n\varepsilon},
\end{split}
\end{equation*}
for some constant $c_2,c_3,c_4$.
\end{proof}

\subsection{Proofs of Statistical Inference}

Given a pre-specified sparsity level $s_j$, the differentially private estimation algorithm for $\boldsymbol{w}_j$ is presented in Algorithm \ref{algo:wj}.

\begin{algorithm}[ht]\caption{Differentially Private Estimation of  $\boldsymbol{w}_j$ given sparsity}\label{algo:wj}
	\begin{algorithmic}[1]
		\Require{Dataset $\{\boldsymbol{x}_i\}_{i}^{n}$, step size $\eta^0$, privacy parameters $(\varepsilon, \delta)$, noise scale $B$, number of iterations $T$, truncation level $R$, feasibility parameter $C$, sparsity $s^*$, initial value $\boldsymbol{w}^{(0)}_j$.}
		\State Random split data into $T$ parts of roughly equal size: $\{1,\dots,n\}= \mathcal{S}_0\cup\dots\cup\mathcal{S}_{T-1}$ and $\mathcal{S}_i\cap\mathcal{S}_j=\emptyset$ for $i\neq j$.
		\For{$t$ in $0$ to $T-1$}
		\State	Gradient descent: $\boldsymbol{w}_j^{(t + 0.5)} = \boldsymbol{w}_j^{(t)} -\eta^0(\boldsymbol{e}_j-\sum_{i\in\mathcal{S}_t}\boldsymbol{x}_i\Pi_R(\boldsymbol{x}_i^{\top}\boldsymbol{w}_j^{(t)})/|\mathcal{S}_t|)$.
		\State Private report: $\boldsymbol{w}_j^{(t+1)} = \Pi_C(\text{NoisyIHT}(\boldsymbol{w}^{(t+0.5)}_j, s^*, \varepsilon/T, \delta/T, \eta^0 B/|\mathcal{S}_t|))$.
		\EndFor
		\Ensure{$\boldsymbol{w}_j^{(T)}$.}
	\end{algorithmic}
\end{algorithm}

The $\ell_2$ error bound for the output of Algorithm \ref{algo:wj} is outlined in Lemma \ref{lem:est_w}. The proof follows arguments similar to those in Theorem 4.4 of \cite{cai2019cost}.

\begin{lemma}
\label{lem:est_w} 
Suppose conditions \ref{cond:1}, \ref{cond:2} and \ref{cond:3} hold, and let $B=2Rc_x$, $C>L$ and $R=C_1\sqrt{\log(n)}$ for a constant $C_1$. There exists a constant $\rho$ such that, if $s^*=\rho L^4s_j$, $T=\rho L^2\log(8L^3n)$, $s_j\log(p)=o(n)$ and $s_j\log (p)\log(1/\delta)\log(n)^{2.5}/\varepsilon=o(n)$. Then with probability at least $1-\exp(-c_1^{\prime}\log n)$, there exist constants $c_2^{\prime}$ and $c_3^{\prime}$, such that
\begin{equation*}
\|\boldsymbol{w}_j^{(T)}-\boldsymbol{w}_j\|_2^2\leq c_2^{\prime}\frac{s_j\log (p)\log(n)}{n}+c_3^{\prime}\frac{s_j^{2}\log(p)^2 \log(1/\delta)\log(n)^5}{n^2\varepsilon^2}.
\end{equation*}
\end{lemma}

\subsubsection{Proof of Lemma \ref{lem:est_w}}
\begin{proof}[Proof of Lemma \ref{lem:est_w}]

We begin the proof by first presenting the statistical error without differential privacy constraints. Let $S_{oracle}=\text{supp}(\boldsymbol{w}_j)$ denote the support of the true parameter $\boldsymbol{w}_j$. For any subset $S$ satisfying $S_{oracle}\subseteq S$, $|S|\leq c_js_j$, and $s^*\leq|S|$, the oracle estimator $\hat{\boldsymbol{w}}_j^o$ is defined as follows:
\begin{equation*}
\hat{\boldsymbol{w}}_j^o=\argmin_{\boldsymbol{w}\in\mathbb{R}^p,\text{supp}(\boldsymbol{w})\subseteq S}\L_n(\boldsymbol{w}):=\frac{1}{2}\boldsymbol{w}^{\top}\widehat{\bSigma}\boldsymbol{w}-\boldsymbol{w}^{\top}\boldsymbol{e}_j,
\end{equation*}
where $c_j$ is a positive constant and $\widehat{\bSigma}=\sum_{i=1}^{n}\boldsymbol{x}_i^{\top}\boldsymbol{x}_i/n$. The condition $|S|\leq c_js_j$ implies that the sparsity requirement for $s^*$ is satisfied. 

The name ``oracle" refers to the fact that $\hat{\boldsymbol{w}}_j^o$ is an estimator that uses the true support set. We first study the statistical properties of $\hat{\boldsymbol{w}}_j^o$. The nonzero components of $\hat{\boldsymbol{w}}_j^o$ are given by
$$\hat{\boldsymbol{w}}_{j,S}^o=\argmin_{\boldsymbol{w}\in\mathbb{R}^{|S|}}\frac{1}{2}\boldsymbol{w}^{\top}\widehat{\bSigma}_{SS}\boldsymbol{w}-\boldsymbol{w}^{\top}\boldsymbol{e}_{j,S},$$ where $\hat{\boldsymbol{w}}_{j,S}^o$ is the sub-vector of $\hat{\boldsymbol{w}}_j^o$, and $\widehat{\bSigma}_{SS}$ is the sub-matrix of $\widehat{\bSigma}$, with both indexed by the set $S$. Since $j\in S$, the sub-vector $\boldsymbol{e}_{j,S}$ remains a unit vector. The analytic solution is given by $\hat{\boldsymbol{w}}_{j,S}^o=\widehat{\bSigma}_{SS}^{-1}\boldsymbol{e}_{j,S}$. Then,
\begin{equation*}
\begin{split}
&\|\hat{\boldsymbol{w}}_j^o-\boldsymbol{w}_j\|_2=\|\hat{\boldsymbol{w}}_{j,S}^o-\boldsymbol{w}_{j,S}\|_2=\|(\widehat{\bSigma}_{SS}^{-1}-\boldsymbol{\Sigma}_{SS}^{-1})\boldsymbol{e}_{j,S}\|_2\\
&\leq\|\widehat{\bSigma}_{SS}^{-1}-\boldsymbol{\Sigma}_{SS}^{-1}\|_2,
\end{split}
\end{equation*}
where the first equality uses the fact that the support of both $\hat{\boldsymbol{w}}_j^o$ and $\boldsymbol{w}_j$ lies in $S$, the second equality follows from the analytic solution form, and the last inequality uses the definition of the matrix $\ell_2$ norm.
	
By Corollary 10.1 in \cite{tan2020sparse}, for any constant $c^{\prime}_w$ and any set $S$ satisfying $|S|\leq c_js_j$, there exists a constant $c_w>0$, such that
\begin{equation*}
\|\widehat{\bSigma}_{SS}-\boldsymbol{\Sigma}_{SS}\|_2^2\leq\frac{c_w}{n}s_j\log(ep/s_j),
\end{equation*}
with probability at least $1-\exp\{-c_w^{\prime}s_j\log(ep/s_j)\}$. Then we have the following relation:
\begin{equation*}
\begin{split}
\|\widehat{\bSigma}_{SS}^{-1}-\boldsymbol{\Sigma}^{-1}_{SS}\|_2&=\|\widehat{\bSigma}_{SS}^{-1}(\widehat{\bSigma}_{SS}-\boldsymbol{\Sigma}_{SS})\boldsymbol{\Sigma}^{-1}_{SS}\|_2\\
&\leq\|\widehat{\bSigma}_{SS}^{-1}\|_2\|\widehat{\bSigma}_{SS}-\boldsymbol{\Sigma}_{SS}\|_2\|\boldsymbol{\Sigma}^{-1}_{SS}\|_2\\
&\leq 2L^2\sqrt{\frac{c_w}{n}s_j\log(ep/s_j)},
\end{split}
\end{equation*}
where the first equality holds because $\boldsymbol{\Sigma}_{SS}$ is invertible by Condition \ref{cond:1}, and $\widehat{\bSigma}_{SS}$ converges to $\boldsymbol{\Sigma}_{SS}$, implying that $\widehat{\bSigma}_{SS}$ is also invertible for sufficiently large $n$. The second inequality uses the bound on $\|\widehat{\bSigma}_{SS}-\boldsymbol{\Sigma}_{SS}\|_2$ from the previous result, and Condition \ref{cond:1}, which implies $\|\boldsymbol{\Sigma}^{-1}_{SS}\|_2\leq\|\boldsymbol{\Sigma}^{-1}\|_2\leq L$. Since $\widehat{\bSigma}_{SS}^{-1}$ converges to $\boldsymbol{\Sigma}_{SS}^{-1}$, we also have $\|\widehat{\bSigma}^{-1}_{SS}\|_2\leq 2\|\boldsymbol{\Sigma}^{-1}\|_2\leq 2L$ for sufficiently large $n$. The constant $2$ is not tight, but keeps the correct order. A similar technique will be used later in the proof. Then we have the following bound:
\begin{equation*}
\|\hat{\boldsymbol{w}}_j^o-\boldsymbol{w}_j\|_2^2\leq 4L^4\frac{c_w}{n}s_j\log(ep/s_j),
\end{equation*}
with probability at least $1-\exp\{-c_w^{\prime}s_j\log(ep/s_j)\}$.
	
Next, we consider the properties of the gradient descent algorithm. Before discussing the algorithm, we define an event under which the truncation operators do not take effect:
\begin{equation*}
\begin{split}
E_3^{\prime}:=\{\max_{t=0,\dots,T-1}|\boldsymbol{x}_i^{\top}\boldsymbol{w}_j^{(t)}|\leq R\text{ for all }i\in \mathcal{S}_t\big\}.
\end{split}
\end{equation*}
By Condition \ref{cond:1} and the independence between $\boldsymbol{x}_i$ and $\boldsymbol{w}^{(t)}$ induced by data splitting, we apply the Chernoff bound to obtain the following large deviation result:
$$\mathbb{P}(|\boldsymbol{x}_i^{\top}\boldsymbol{w}_j^{(t)}|\geq R)\leq 2\exp\{-cR^2/(C^2\|\boldsymbol{x}_i\|_{\psi_2}^2)\},$$
where $c$ is an absolute constant and we use the fact that $\boldsymbol{x}_i^{\top}\boldsymbol{w}_j^{(t)}$ is sub-Gaussian with sub-Gaussian norm bounded by $C\|\boldsymbol{x}_i\|_{\psi_2}$. By applying the union bound, we have
\begin{equation*}
\begin{split}
\mathbb{P}(E_3^{\prime})&\geq 1-\sum_{t=0}^{T-1}\sum_{i\in S_t}\mathbb{P}(|\boldsymbol{x}_i^{\top}\boldsymbol{w}_j^{(t)}|\geq R)\geq 1-2n\exp\{-cR^2C^2\|\boldsymbol{x}_i\|_{\psi_2}^2\}\},
\end{split}
\end{equation*}
where the first inequality follows from applying the Chernoff bound $n$ times. By choosing $R=\sqrt{2C^2\|\boldsymbol{x}_i\|_{\psi_2}^2\log(n)/c}$, we obtain $\mathbb{P}(E_3^{\prime})\geq 1-2n\exp(-2\log(n))=1-2\exp(-\log(n))$. Thus, truncation operators do not occur with high probability, and we omit them in the remainder of the proof. To simplify notation, we define the empirical loss function as 
$$\L_n(\boldsymbol{w})=\frac{1}{2}\boldsymbol{w}^{\top}\widehat{\bSigma}\boldsymbol{w}-\boldsymbol{w}^{\top}\boldsymbol{e}_j.$$ 
Since data splitting is used in the algorithm, the sample size in each iteration is $n/T$. For clarity of presentation, we omit the subsample notation. Note that $\L_n(\boldsymbol{w})$ satisfies the following property:
$$\langle\nabla\L_n(\boldsymbol{w}_1)-\nabla\L_n(\boldsymbol{w}_2),\boldsymbol{w}_1-\boldsymbol{w}_2\rangle=(\boldsymbol{w}_1-\boldsymbol{w}_2)^{\top}\widehat{\bSigma}(\boldsymbol{w}_1-\boldsymbol{w}_2).$$
Thus, we have
\begin{equation}
\label{eq:gamma}
\alpha\|\boldsymbol{w}_1-\boldsymbol{w}_2\|_2^2\leq \langle\nabla\L_n(\boldsymbol{w}_1)-\nabla\L_n(\boldsymbol{w}_2),\boldsymbol{w}_1-\boldsymbol{w}_2\rangle\leq \gamma\|\boldsymbol{w}_1-\boldsymbol{w}_2\|_2^2,
\end{equation}
for all $\boldsymbol{w}_1,\boldsymbol{w}_2\in\mathbb{R}^p$ such that $\max\{|\text{supp}(\boldsymbol{w}_1)|,|\text{supp}(\boldsymbol{w}_2)|\}\leq c_js_j/2$. Since $|\text{supp}(\boldsymbol{w}_1)\cup\text{supp}(\boldsymbol{w}_2)|\leq c_js_j$, and by the uniform convergence of submatrices, the above inequality holds with $\alpha=1/(2L)$ and $\gamma=2L$ with high probability, where we use Condition \ref{cond:1} for the population matrix $\boldsymbol{\Sigma}$. Then we have
\begin{equation*}
\begin{split}
\L_n(\boldsymbol{w}_j^{(t+1)})-\L_n(\boldsymbol{w}_j^{(t)})
=&\frac{1}{2}\boldsymbol{w}_j^{(t+1)\top}\widehat{\bSigma}\boldsymbol{w}_j^{(t+1)}-\frac{1}{2}\boldsymbol{w}_j^{(t)\top}\widehat{\bSigma}\boldsymbol{w}_j^{(t)}-(\boldsymbol{w}_j^{(t+1)}-\boldsymbol{w}_j^{(t)})^{\top}\boldsymbol{e}_j\\
=&\langle\boldsymbol{w}_j^{(t+1)}-\boldsymbol{w}_j^{(t)},\boldsymbol{w}_j^{(t)\top}\widehat{\bSigma}-\boldsymbol{e}_j\rangle\\
&+\frac{1}{2}(\boldsymbol{w}_j^{(t+1)}-\boldsymbol{w}_j^{(t)})^{\top}\widehat{\bSigma}(\boldsymbol{w}_j^{(t+1)}-\boldsymbol{w}_j^{(t)})\\
&\leq \langle\boldsymbol{w}_j^{(t+1)}-\boldsymbol{w}_j^{(t)},\boldsymbol{g}^t\rangle+\frac{\gamma}{2}\|\boldsymbol{w}_j^{(t+1)}-\boldsymbol{w}_j^{(t)}\|_2^2,
\end{split}
\end{equation*}
where $\boldsymbol{g}^t=\boldsymbol{w}_j^{(t)\top}\widehat{\bSigma}-\boldsymbol{e}_j$ is the gradient of $\L_n(\boldsymbol{w})$ evaluated at $\boldsymbol{w}_j^{(t)}$. Let $S^t=\text{supp}(\boldsymbol{w}_j^{(t)})$, $S^{t+1}=\text{supp}(\boldsymbol{w}_j^{(t+1)})$, and define $I^t=S^{t+1}\cup S^{t}\cup S$. Let $\boldsymbol{n}_1^t,\boldsymbol{n}_2^t,\dots,\boldsymbol{n}_{s^*}^t$ be the noise vectors added to $\boldsymbol{w}_j^{(t)}-\eta^0\nabla\L_n(\boldsymbol{w}_j^{(t)})$ during the peeling mechanism over a total of $s^*$ iterations in the $t$th step, and define $\boldsymbol{N}^t=4\sum_{i\in[s^*]}\|\boldsymbol{n}_i^t\|^2_{\infty}$. Then we have the following decomposition:
\begin{equation*}
\begin{split}
\langle\boldsymbol{w}_j^{(t+1)}-\boldsymbol{w}_j^{(t)},\boldsymbol{g}^t\rangle+\frac{\gamma}{2}\|\boldsymbol{w}_j^{(t+1)}-\boldsymbol{w}_j^{(t)}\|_2^2=&\frac{\gamma}{2}\|\boldsymbol{w}_{j,I^t}^{(t+1)}-\boldsymbol{w}_{j,I^t}^{(t)}+\frac{\eta}{\gamma}\boldsymbol{g}^t_{I^t}\|_2^2-\frac{\eta^2}{2\gamma}\|\boldsymbol{g}^t_{I^t}\|_2^2\\
&+(1-\eta)\langle\boldsymbol{w}_j^{(t+1)}-\boldsymbol{w}_j^{(t)},\boldsymbol{g}^t\rangle,
\end{split}
\end{equation*}
where $\gamma$ is defined in \eqref{eq:gamma}, and we introduce the notation $\eta:=\gamma\cdot\eta^0$.

We first consider the first two terms. Let $R$ be a subset of $S^t\backslash S^{t+1}$ such that $|R|=|I^t\backslash (S^t\cup S)|=|S^{t+1}\backslash (S^{t}\cup S)|$. Then, using the fact that $\boldsymbol{w}^{(t)}_{j,I^t/(S^t\cup S)}=\boldsymbol{0}$, and by Lemma 3.4 in \cite{cai2019cost}, we have, for every $c>1$,
\begin{equation*}
\begin{split}
&\frac{\eta^2}{\gamma^2}\|\boldsymbol{g}^t_{I^t\backslash(S^t\cup S)}\|_2^2=\|\boldsymbol{w}_{j,I^t\backslash(S^t\cup S)}^{(t)}-\frac{\eta}{\gamma}\boldsymbol{g}^t_{I^t\backslash(S^t\cup S)}\|_2^2\geq (1-1/c)\|\boldsymbol{w}_{j,R}^{(t)}-\frac{\eta}{\gamma}\boldsymbol{g}^t_{R}\|_2^2-c\boldsymbol{N}^t.
\end{split}
\end{equation*}
Since $\boldsymbol{w}_{j}^{(t+1)}$ is obtained by selecting the noisy maximum of $\boldsymbol{w}_j^{(t+0.5)}$ and then adding noise, we can write $\boldsymbol{w}_{j}^{(t+1)}=\tilde{\boldsymbol{w}}_{j}^{(t+1)}+\tilde{\boldsymbol{n}}_{S^{t+1}}$, where $\tilde{\boldsymbol{w}}_{j}^{(t+1)}$is the vector corresponding to the noisy maximum index of $\boldsymbol{w}^{(t+0.5)}_j$ and $\tilde{\boldsymbol{n}}_{S^{t+1}}$ represents the additional noise introduced by the peeling mechanism. Then we have
\begin{equation*}
\begin{split}
&\frac{\gamma}{2}\|\boldsymbol{w}_{j,I^t}^{(t+1)}-\boldsymbol{w}_{j,I^t}^{(t)}+\frac{\eta}{\gamma}\boldsymbol{g}^t_{I^t}\|_2^2-\frac{\eta^2}{2\gamma}\|\boldsymbol{g}^t_{I^t/(S^t\cup S)}\|_2^2\\
\leq&\frac{\gamma}{2}\|\tilde{\boldsymbol{n}}_{S^{t+1}}\|_2^2+\frac{\gamma}{2}\|\tilde{\boldsymbol{w}}_{j,I^t}^{(t+1)}-\boldsymbol{w}_{j,I^t}^{(t)}+\frac{\eta}{\gamma}\boldsymbol{g}^t_{I^t}\|_2^2-\frac{\gamma}{2}(1-1/c)\|\boldsymbol{w}_{j,R}^{(t)}-\frac{\eta}{\gamma}\boldsymbol{g}^t_{R}\|_2^2+\frac{c\gamma}{2}\boldsymbol{N}^t\\
=&\frac{\gamma}{2}\|\tilde{\boldsymbol{w}}_{j,I^t}^{(t+1)}-\boldsymbol{w}_{j,I^t}^{(t)}+\frac{\eta}{\gamma}\boldsymbol{g}^t_{I^t}\|_2^2-\frac{\gamma}{2}\|\tilde{\boldsymbol{w}}_{j,R}^{(t+1)}-\boldsymbol{w}_{j,R}^{(t)}+\frac{\eta}{\gamma}\boldsymbol{g}^t_{R}\|_2^2+\frac{\gamma}{2}(1/c)\|\boldsymbol{w}_{j,R}^{(t)}-\frac{\eta}{\gamma}\boldsymbol{g}^t_{R}\|_2^2\\
&+\frac{\gamma}{2}\|\tilde{\boldsymbol{n}}_{S^{t+1}}\|_2^2+\frac{c\gamma}{2}\boldsymbol{N}^t\\
\leq&\frac{\gamma}{2}\|\tilde{\boldsymbol{w}}_{j,I^t\backslash R}^{(t+1)}-\boldsymbol{w}_{j,I^t\backslash R}^{(t)}+\frac{\eta}{\gamma}\boldsymbol{g}^t_{I^t\backslash R}\|_2^2+\frac{\eta^2}{2c\gamma}(1+1/c)\|\boldsymbol{g}^t_{I^t\backslash(S^t\cup S)}\|_2^2+\frac{\gamma}{2}\|\tilde{\boldsymbol{n}}_{S^{t+1}}\|_2^2+c\gamma\boldsymbol{N}^t,
\end{split}
\end{equation*}
where we apply the selection criterion in the first inequality, use the fact that $\tilde{\boldsymbol{w}}_{j,R}^{(t+1)}=\boldsymbol{0}$ in the second equality, and apply Lemma 3.4 in \cite{cai2019cost} to $\|\boldsymbol{w}_{j,R}^{(t)}-\frac{\eta}{\gamma}\boldsymbol{g}^t_{R}\|_2^2$ in the last inequality. By Lemma A.3. in \cite{cai2019cost}, we have
\begin{equation*}
\begin{split}
&\|\tilde{\boldsymbol{w}}_{j,I^t\backslash R}^{(t+1)}-\boldsymbol{w}_{j,I^t\backslash R}^{(t)}+\frac{\eta}{\gamma}\boldsymbol{g}^t_{I^t\backslash R}\|_2^2\leq\frac{3}{2}\frac{|I^t/R|-s^*}{|I^t/R|-s_j}\|\hat{\boldsymbol{w}}^o_{j,I^t\backslash R}-\boldsymbol{w}_{j,I^t\backslash R}^{(t)}+\frac{\eta}{\gamma}\boldsymbol{g}^t_{I^t\backslash R}\|_2^2+3\boldsymbol{N}^t,
\end{split}
\end{equation*}
where $\hat{\boldsymbol{w}}_{j}^o$ is the oracle estimator. Plugging it into the previous inequality, we have
\begin{equation}
\begin{split}
\label{eq:second1}
&\frac{\gamma}{2}\|\boldsymbol{w}_{j,I^t}^{(t+1)}-\boldsymbol{w}_{j,I^t}^{(t)}+\frac{\eta}{\gamma}\boldsymbol{g}^t_{I^t}\|_2^2-\frac{\eta^2}{2\gamma}\|\boldsymbol{g}^t_{I^t\backslash(S^t\cup S)}\|_2^2\\
\leq&\frac{3\gamma}{4}\frac{|I^t/R|-s^*}{|I^t/R|-s_j}\|\hat{\boldsymbol{w}}^o_{j,I^t}-\boldsymbol{w}_{j,I^t}^{(t)}+\frac{\eta}{\gamma}\boldsymbol{g}^t_{I^t}\|_2^2+3\gamma/2\boldsymbol{N}^t+\frac{\eta^2(1+1/c)}{2c\gamma}\|\boldsymbol{g}^t_{I^t/(S^t\cup S)}\|_2^2\\
&+\frac{\gamma}{2}\|\tilde{\boldsymbol{n}}_{S^{t+1}}\|_2^2+c\gamma\boldsymbol{N}^t\\
\leq&\frac{3\gamma}{4}\frac{2s_j}{s^*+s_j}\|\hat{\boldsymbol{w}}^o_{j,I^t}-\boldsymbol{w}_{j,I^t}^{(t)}+\frac{\eta}{\gamma}\boldsymbol{g}^t_{I^t}\|_2^2+3\gamma/2\boldsymbol{N}^t+\frac{\eta^2(1+1/c)}{2c\gamma}\|\boldsymbol{g}^t_{I^t/(S^t\cup S)}\|_2^2\\
&+\frac{\gamma}{2}\|\tilde{\boldsymbol{n}}_{S^{t+1}}\|_2^2+c\gamma\boldsymbol{N}^t,
\end{split}
\end{equation}
where in the second inequality, we use the fact $|I^t\backslash R|\leq2s_j+s^*$ and $I^t\backslash (S^t\cup S)\subset S^{t+1}$. Furthermore,
\begin{equation}
\label{eq:second2}
    \begin{split}
        &\frac{3\gamma}{4}\frac{2s_j}{s^*+s_j}\|\hat{\boldsymbol{w}}_{j,I^t}^o-\boldsymbol{w}_{j,I^t}^{(t)}+\frac{\eta}{\gamma}\boldsymbol{g}^t_{I^t}\|_2^2\\
        \leq&\frac{3s_j}{s^*+s_j}(\eta\langle\hat{\boldsymbol{w}}_{j}^o-\boldsymbol{w}_{j}^{(t)},\boldsymbol{g}^t\rangle+\frac{\gamma}{2}\|\hat{\boldsymbol{w}}_{j}^o-\boldsymbol{w}_{j}^{(t)}\|_2+\frac{\eta^2}{2\gamma}\|\boldsymbol{g}^t_{I^t}\|_2^2)\\
        \leq&\frac{3s_j}{s^*+s_j}(\eta\mathcal{L}_n(\hat{\boldsymbol{w}}_{j}^o)-\eta\mathcal{L}_n(\boldsymbol{w}_{j}^{(t)})+\frac{\gamma-\eta\alpha}{2}\|\hat{\boldsymbol{w}}_{j}^o-\boldsymbol{w}_{j}^{(t)}\|_2+\frac{\eta^2}{2\gamma}\|\boldsymbol{g}^t_{I^t}\|_2^2),
    \end{split}
\end{equation}
where the constant $\alpha$ is defined in \eqref{eq:gamma}.

Next, we consider the second term, which can be decomposed as follows:
\begin{equation*}
\begin{split}
\langle\boldsymbol{w}_j^{(t+1)}-\boldsymbol{w}_j^{(t)},\boldsymbol{g}^t\rangle&=\langle\tilde{\boldsymbol{w}}_{j,S^{t+1}}^{(t+1)}-\boldsymbol{w}_{j,S^{t+1}}^{(t)},\boldsymbol{g}^t_{S^{t+1}}\rangle+\langle\tilde{\boldsymbol{n}}_{S^{t+1}},\boldsymbol{g}^t_{S^{t+1}}\rangle\\
&-\langle\boldsymbol{w}_{j,S^t\backslash S^{t+1}}^{(t)},\boldsymbol{g}_{S^t\backslash S^{t+1}}^t\rangle\\
&\leq-\frac{\eta}{\gamma}\|\boldsymbol{g}^t_{S^{t+1}}\|_2^2+c\|\boldsymbol{n}_{S^{t+1}}\|_2^2+(1/4c)\|\boldsymbol{g}^t_{S^{t+1}}\|^2_2\\
&-\langle\boldsymbol{w}_{j,S^t\backslash S^{t+1}}^{(t)},\boldsymbol{g}_{S^t\backslash S^{t+1}}^t\rangle,
\end{split}
\end{equation*}
where we use the inequality $ab\leq a^2/2+b^2/2$. The last term satisfies
\begin{equation*}
\begin{split}
&-\langle\boldsymbol{w}_{j,S^t\backslash S^{t+1}}^{(t)},\boldsymbol{g}_{S^t\backslash S^{t+1}}^t\rangle\leq\frac{\gamma}{2\eta}\bigg\{\|\boldsymbol{w}_{j,S^t\backslash S^{t+1}}^{(t)}-\frac{\eta}{\gamma}\boldsymbol{g}_{S^t\backslash S^{t+1}}^t\|_2^2-(\frac{\eta}{\gamma})^2\|\boldsymbol{g}_{S^t\backslash S^{t+1}}^t\|_2^2\bigg\},
\end{split}
\end{equation*}
by simple algebra. By applying Lemma 3.4 in \cite{cai2019cost} to $\|\boldsymbol{w}_{j,S^t\backslash S^{t+1}}^{(t)}-\frac{\eta}{\gamma}\boldsymbol{g}_{S^t\backslash S^{t+1}}^t\|_2^2$, we have
\begin{equation*}
\begin{split}
-\langle\boldsymbol{w}_{j,S^t\backslash S^{t+1}}^{(t)},\boldsymbol{g}_{S^t\backslash S^{t+1}}^t\rangle
&\leq\frac{\gamma}{2\eta}\{(1+1/c)\|\tilde{\boldsymbol{w}}_{j,S^{t+1}\backslash S^{t}}^{(t+1)}\|_2^2+(1+c)\boldsymbol{N}^t\}-\frac{\eta}{2\gamma}\|\boldsymbol{g}_{S^t\backslash S^{t+1}}^t\|_2^2\\
&=\frac{\eta}{2\gamma}\{(1+1/c)\|\boldsymbol{g}_{S^{t+1}\backslash S^{t}}^{t}\|_2^2+(1+c)\frac{\gamma}{\eta}\boldsymbol{N}^t\}-\frac{\eta}{2\gamma}\|\boldsymbol{g}_{S^t\backslash S^{t+1}}^t\|_2^2,
\end{split}
\end{equation*}
where we use the fact that $\tilde{\boldsymbol{w}}_{j,S^{t+1}\backslash S^{t}}^{(t+1)}=\frac{\eta}{\gamma}\boldsymbol{g}_{S^{t+1}\backslash S^{t}}^{t}$ in the second equality. Combining the results above, we have:
\begin{equation*}
\begin{split}
\langle\boldsymbol{w}_j^{(t+1)}-\boldsymbol{w}_j^{(t)},\boldsymbol{g}^t\rangle
\leq&-\frac{\eta}{\gamma}\|\boldsymbol{g}^t_{S^{t+1}}\|_2^2+c\|\tilde{\boldsymbol{n}}_{S^{t+1}}\|_2^2+(1/4c)\|\boldsymbol{g}^t_{S^{t+1}}\|^2_2\\
&+\frac{\eta}{2\gamma}\big\{(1+1/c)\|\boldsymbol{g}_{S^{t+1}\backslash S^{t}}^{t}\|_2^2+(1+c)\frac{\gamma}{\eta}\boldsymbol{N}^t\big\}-\frac{\eta}{2\gamma}\|\boldsymbol{g}_{S^t\backslash S^{t+1}}^t\|_2^2\\
\leq&\frac{\eta}{2\gamma}\|\boldsymbol{g}_{S^{t+1}\backslash S^{t}}^{t}\|_2^2-\frac{\eta}{2\gamma}\|\boldsymbol{g}_{S^t\backslash S^{t+1}}^t\|_2^2-\frac{\eta}{\gamma}\|\boldsymbol{g}^t_{S^{t+1}}\|_2^2\\
&+(1/c)(4+\frac{\eta}{2\gamma})\|\boldsymbol{g}^t_{S^{t+1}}\|^2_2+c\|\tilde{\boldsymbol{n}}_{S^{t+1}}\|_2^2+(1+c)\frac{\gamma}{\eta}\boldsymbol{N}^t\\
\leq&-\frac{\eta}{2\gamma}\|\boldsymbol{g}_{S^t\cup S^{t+1}}^t\|_2^2+(1/c)(4+\frac{\eta}{2\gamma})\|\boldsymbol{g}^t_{S^{t+1}}\|^2_2+c\|\tilde{\boldsymbol{n}}_{S^{t+1}}\|_2^2+(1+c)\frac{\gamma}{\eta}\boldsymbol{N}^t,
\end{split}
\end{equation*}
where we use simple algebra. Then, by plugging in the previous results into $\L_n(\boldsymbol{w}_j^{(t+1)})-\L_n(\boldsymbol{w}_j^{(t)})$, we have:
\begin{equation*}
\begin{split}
\L_n(\boldsymbol{w}_j^{(t+1)})-\L_n(\boldsymbol{w}_j^{(t)})\leq&\frac{\gamma}{2}\|\boldsymbol{w}_{j,I^t}^{(t+1)}-\boldsymbol{w}_{j,I^t}^{(t)}+\frac{\eta}{\gamma}\boldsymbol{g}^t_{I^t}\|_2^2-\frac{\eta^2}{2\gamma}\|\boldsymbol{g}^t_{I^t}\|_2^2-\frac{\eta(1-\eta)}{2\gamma}\|\boldsymbol{g}_{S^t\cup S^{t+1}}^t\|_2^2\\
&+(1-\eta)(1/c)(4+\frac{\eta}{2\gamma})\|\boldsymbol{g}^t_{S^{t+1}}\|^2_2+c(1-\eta)\|\tilde{\boldsymbol{n}}_{S^{t+1}}\|_2^2\\
&+(1-\eta)(1+c)\frac{\gamma}{\eta}\boldsymbol{N}^t\\
\leq&\frac{\gamma}{2}\|\boldsymbol{w}_{j,I^t}^{(t+1)}-\boldsymbol{w}_{j,I^t}^{(t)}+\frac{\eta}{\gamma}\boldsymbol{g}^t_{I^t}\|_2^2-\frac{\eta^2}{2\gamma}\|\boldsymbol{g}^t_{I^t\backslash(S^t\cup S)}\|_2^2-\frac{\eta^2}{2\gamma}\|\boldsymbol{g}^t_{S^t\cup S}\|_2^2\\
&-\frac{\eta(1-\eta)}{2\gamma}\|\boldsymbol{g}_{S^{t+1}\backslash(S^t\cup S)}^t\|_2^2+(1-\eta)(1/c)(4+\frac{\eta}{2\gamma})\|\boldsymbol{g}^t_{S^{t+1}}\|^2_2\\
&+c(1-\eta)\|\tilde{\boldsymbol{n}}_{S^{t+1}}\|_2^2+(1-\eta)(1+c)\frac{\gamma}{\eta}\boldsymbol{N}^t,
\end{split}
\end{equation*}
where we use the fact that $S^{t+1}\backslash(S^t\cup S)$ is a subset of $S^t\cup S^{t+1}$. Note that the first two terms
\begin{equation*}  
\frac{\gamma}{2}\|\boldsymbol{w}_{j,I^t}^{(t+1)}-\boldsymbol{w}_{j,I^t}^{(t)}+\frac{\eta}{\gamma}\boldsymbol{g}^t_{I^t}\|_2^2-\frac{\eta^2}{2\gamma}\|\boldsymbol{g}^t_{I^t\backslash(S^t\cup S)}\|_2^2
\end{equation*}
are analyzed in \eqref{eq:second1} and \eqref{eq:second2}. Combining all the results, we obtain:
\begin{equation*}
\begin{split}  
\L_n(\boldsymbol{w}_j^{(t+1)})-\L_n(\boldsymbol{w}_j^{(t)})\leq&\frac{3s_j}{s^*+s_j}(\eta\L_n(\hat{\boldsymbol{w}}_j^o)-\eta\L_n(\boldsymbol{w}_j^{(t)})+\frac{\gamma-\eta\alpha}{2}\|\hat{\boldsymbol{w}}_{j}-\boldsymbol{w}_j^{(t)}\|_2^2+\frac{\eta^2}{2\gamma}\|\boldsymbol{g}^t_{I^t}\|_2^2)\\
&-\frac{\eta^2}{4\gamma}\|\boldsymbol{g}^t_{S^t\cup S}\|_2^2-\frac{\eta(1-\eta)}{4\gamma}\|\boldsymbol{g}_{S^{t+1}\backslash(S^t\cup S)}^t\|_2^2\\
&+\frac{\gamma}{2}(4+3c)\frac{\gamma}{2\eta}\boldsymbol{N}^t+(\frac{\gamma}{2}+\frac{c}{3})\|\tilde{\boldsymbol{n}}_{S^{t+1}}\|_2^2,
\end{split}
\end{equation*}
where we let $\eta=2/3$	and choose the constant $c$ to be sufficiently large.
    
By choosing $s^*=72(\gamma/\alpha)^2s_j=\rho L^4s_j$ where $\rho$ is the absolute constant, we ensure that $3s_j/(s^*+s_j)\leq\alpha^2/\{24\gamma(\gamma-\eta\alpha)\}\leq 1/8$. Then,
\begin{equation*}
\begin{split}
&\L_n(\boldsymbol{w}_j^{(t+1)})-\L_n(\boldsymbol{w}_j^{(t)})\\
\leq&\frac{3s_j}{s_j+s^*}\eta(\L_n(\hat{\boldsymbol{w}}_j^o)-\L_n(\boldsymbol{w}_j^{(t)})\}+\frac{\alpha^2}{48\gamma}\|\hat{\boldsymbol{w}}_j^o-\boldsymbol{w}_j^{(t)}\|_2^2+\frac{1}{36\gamma}\|\boldsymbol{g}^t_{I^t}\|_2^2\\
&-\frac{1}{9\gamma}\|\boldsymbol{g}^t_{S^t\cup S}\|_2^2-\frac{1}{18\gamma}\|\boldsymbol{g}_{S^{t+1}/(S^t\cup S)}^t\|_2^2++\frac{\gamma}{2}(4+3c)\frac{\gamma}{2\eta}\boldsymbol{N}^t+(\frac{\gamma}{2}+\frac{c}{3})\|\tilde{\boldsymbol{n}}_{S^{t+1}}\|_2^2\\
\leq&\frac{3s_j}{s_j+s^*}\eta(\L_n(\hat{\boldsymbol{w}}_j^o)-\L_n(\boldsymbol{w}_j^{(t)})\}+\frac{\alpha^2}{48\gamma}\|\hat{\boldsymbol{w}}_j^o-\boldsymbol{w}_j^{(t)}\|_2^2-\frac{3}{36\gamma}\|\boldsymbol{g}^t_{S^t\cup S}\|_2^2\\
&+\frac{\gamma}{2}(4+3c)\frac{\gamma}{2\eta}\boldsymbol{N}^t+(\frac{\gamma}{2}+\frac{c}{3})\|\tilde{\boldsymbol{n}}_{S^{t+1}}\|_2^2\\
\leq &-(\frac{3\alpha}{72\gamma}+\frac{2s^*}{s_j+s^*})(\L_n(\boldsymbol{w}_j^{(t)})-\L_n(\hat{\boldsymbol{w}}_j^o))+c_n(\|\tilde{\boldsymbol{n}}_{S^{t+1}}\|_2^2+\boldsymbol{N}^t)\\
\leq &-\frac{1}{\rho L^2}(\L_n(\boldsymbol{w}_j^{(t)})-\L_n(\hat{\boldsymbol{w}}_j^o))+c_n(\|\tilde{\boldsymbol{n}}_{S^{t+1}}\|_2^2+\boldsymbol{N}^t),
\end{split}
\end{equation*}
where we use the fact that $\|\boldsymbol{g}^t_{I^t}\|_2^2=\|\boldsymbol{g}^t_{S^t\cup S}\|_2^2+\|\boldsymbol{g}^t_{S^{t+1}\backslash(S^t\cup S)}\|_2^2$ in the second inequality, and apply Lemma A.4 from \cite{cai2019cost} in the third inequality for an appropriate constant $c_n$. Thus, we have
\begin{equation*}
\L_n(\boldsymbol{w}_j^{(t+1)})-\L_n(\hat{\boldsymbol{w}}_j^o)\leq(1-\frac{1}{\rho L^2})(\L_n(\boldsymbol{w}_j^{(t)})-\L_n(\hat{\boldsymbol{w}}_j^o))+c_n(\|\tilde{\boldsymbol{n}}_{S^{t+1}}\|_2^2+\boldsymbol{N}^t).
\end{equation*}
Let $\tilde{\boldsymbol{N}}_t=c_n(\|\boldsymbol{n}_{S^{t+1}}\|_2^2+\boldsymbol{N}^t)$ and iterate above equation,
\begin{equation*}
\begin{split}
&\L_n(\boldsymbol{w}_j^{(T)})-\L_n(\hat{\boldsymbol{w}}_j^o)\leq(1-\frac{1}{\rho L^2})^T\{\L_n(\boldsymbol{w}_j^{(0)})-\L_n(\hat{\boldsymbol{w}}_j^o)\}+\sum_{k=0}^{T-1}(1-\frac{1}{\rho L^2})^{T-k-1}\tilde{\boldsymbol{N}}_k.
\end{split}
\end{equation*}
By choosing $T=\Omega\big(\log(n)\big)$, the first term is of order $1/n$, due to the boundedness of $\L_n(\boldsymbol{w}_j^{(0)})-\L_n(\hat{\boldsymbol{w}}_j^o)$. Furthermore, we have:
\begin{equation*}
\begin{split}
&\L_n(\boldsymbol{w}_j^{(T)})-\L_n(\hat{\boldsymbol{w}}_j^o)\geq\L_n(\boldsymbol{w}_j^{(T)})-\L_n(\boldsymbol{w}_j)\geq \alpha/2\|\boldsymbol{w}_j^{(T)}-\boldsymbol{w}_j\|_2^2-\langle\nabla\L_n(\boldsymbol{w}_j),\boldsymbol{w}_j-\boldsymbol{w}_j^{(T)}\rangle,
\end{split}
\end{equation*}
where the first inequality follows from the oracle property of the finite-sample loss function, and the second from algebra and the properties of submatrices. Combining all the results, we have:
\begin{equation*}
\begin{split}
&\alpha/2\|\boldsymbol{w}_j^{(T)}-\boldsymbol{w}_j\|_2^2\leq\|\nabla\L_n(\boldsymbol{w}_j)\|_{\infty}\sqrt{s_j+s^*}\|\boldsymbol{w}_j-\boldsymbol{w}_j^{(T)}\|_2\\
&+1/n+\sum_{k=0}^{T-1}(1-\frac{1}{\rho L^2})^{T-k-1}\tilde{\boldsymbol{N}}_k,
\end{split}
\end{equation*}
where we use Hölder's inequality and the norm inequality $\|\cdot\|_2\leq\sqrt{\|\cdot\|_0}\times\|\cdot\|_{\infty}$. Thus, by treating $t=\|\boldsymbol{w}_j^{(T)}-\boldsymbol{w}_j\|_2$ as the unknown variable, the inequality above becomes a quadratic inequality. Then, by the argument used in the proof of Theorem \ref{thm-adapt}, we have:
\begin{equation*}
\begin{split}
&(2\alpha)^{-2}\|\boldsymbol{w}_j^{(T)}-\boldsymbol{w}_j\|^2_2\leq (\|\nabla\L_n(\boldsymbol{w}_j)\|_{\infty}\sqrt{s_j+s^*})^2\\
&+(2\alpha)\{1/n+\sum_{k=0}^{T-1}(1-\frac{1}{\rho L^2})^{T-k-1}\tilde{\boldsymbol{N}}_k\},
\end{split}
\end{equation*}
It remains to analyze the two terms separately. Note that $\nabla\L_n(\boldsymbol{w}_j)$ is the gradient evaluated at the true parameter. Under the model, each coordinate of $\nabla\L_n(\boldsymbol{w}_j)=\boldsymbol{w}_j^{\top}\widehat{\bSigma}-\boldsymbol{e}_j$ is an average of $n/T$ i.i.d. sub-exponential random variables, where we use the fact that the product of two sub-Gaussian random variables is sub-exponential. Then, by Bernstein’s inequality and the union bound, we have:
\begin{equation*}
\begin{split}
\mathbb{P}(\|\nabla\L_n(\boldsymbol{w}_j)\|_{\infty}&\leq \sqrt{2\log(p)\|\boldsymbol{w}_j^{\top}\boldsymbol{x}_i\boldsymbol{x}_i^{\top}\|_{\phi_1}/(cn/T)})\\
&\geq 1-p\exp\{-2\log(p)\}=1-\exp\{-\log(p)\}.
\end{split}
\end{equation*}
Thus, with probability at least $1-\exp\{-\log(p)\}$, by the choice $s^*=\Omega(s_j)$, the first term is bounded by $c_2s_j\log(p)\log(n)/n$ for a constant $c_2$. The second term satisfies
\begin{equation*}
\begin{split}
\sum_{k=0}^{T-1}(1-\frac{1}{\rho L^2})^{T-k-1}\tilde{\boldsymbol{N}}_k&\leq \sum_{k=0}^{\infty}(1-\frac{1}{\rho L^2})^{k}\max_{0\leq t\leq T}\tilde{\boldsymbol{N}}_t\leq \rho L^2\max_{0\leq t\leq T-1}\tilde{\boldsymbol{N}}_t.
\end{split}    
\end{equation*}
Note that $\tilde{\boldsymbol{N}}_t=c_n(\|\tilde{\boldsymbol{n}}_{S^{t+1}}\|_2^2+\boldsymbol{N}^t)$. The term $\max_t\{\|\tilde{\boldsymbol{n}}_{S^{t+1}}\|_{\infty},\|\boldsymbol{N}^t\|_{\infty}\}$ is the maximum over $T(ps^*+s^*)$ independent sub-exponential random variables.
Thus, by a large deviation bound (Chernoff’s inequality) and the union bound, we have:
\begin{equation*}
\begin{split}
&\mathbb{P}(\max_t\{\|\boldsymbol{n}_{S^{t+1}}\|_{\infty},\|\boldsymbol{n}_t\|_{\infty}\}\leq 4\eta^0\frac{BT\sqrt{s^*\log(T/\delta)}}{|S_t|\epsilon}\log(p)/c)\\
\geq& 1-(ps^*+s^*)T\exp(-4\log(p))\geq 1-\exp(-\log(p)).
\end{split}
\end{equation*}
Therefore, with probability at least $1-\exp(-\log(p))$, the second term is bounded by $$c_3s_j^2\log (p)^2\log(1/\delta)\log(n)^5/(n^2\varepsilon^2)$$
for a constant $c_3$, where we use the condition $s^*=\Omega(s_j)$. Combining both terms, we obtain:
\begin{equation*}
\|\boldsymbol{w}_j^{(T)}-\boldsymbol{w}_j\|_2^2\leq c_2^{\prime}\frac{s_j\log(p)\log(n)}{n}+c_3^{\prime}\frac{s_j^2\log (p)^2\log(1/\delta)\log(n)^5}{n^2\varepsilon^2}.
\end{equation*}
\end{proof}

\subsubsection{Proof of Lemma \ref{lem:privacy_adap_w}}
\begin{proof}[Proof of Lemma \ref{lem:privacy_adap_w}]
For $0 \leq k \leq K$, the $\ell_{\infty}$ sensitivity of the gradient at the $t$-th iteration, given by 
$$-\eta^0\{\boldsymbol{e}_j-\sum_{i\in S_t}\boldsymbol{x}_i\Pi_R(\boldsymbol{x}_i^{\top}\boldsymbol{w}^{(t)})/|S_t|\},$$
as defined in line 8 of Algorithm \ref{alg3}, is
\begin{equation*}
\begin{split}
&\sup_{\boldsymbol{x}_i,\boldsymbol{x}_i^{\prime}}\eta^0/|S_t|\cdot\|\boldsymbol{x}_i\Pi_R(\boldsymbol{x}_i^{\top}\boldsymbol{w}^{(t)})-\boldsymbol{x}_i^{\prime}\Pi_R(\boldsymbol{x}_i^{\prime\top}\boldsymbol{w}^{(t)})\|_{\infty}
\leq \eta^0 T/n\cdot 2Rc_x,
\end{split}
\end{equation*}
where we use the fact that $\|\boldsymbol{x}\|_{\infty}\leq c_x$, by Condition \ref{cond:1}. By the advanced composition theorem, reporting the gradient is $(\varepsilon/\{T(K+2)\},\delta/\{T(K+1)\})$-DP. For $0\leq k\leq K$, outputting $\hat{\boldsymbol{w}}(k)$ is $(\varepsilon/(K+2),\delta/(K+1))$-DP by the standard composition theorem. Finally, by the composition theorem, returning all $\{\hat{\boldsymbol{w}}(k)\}_{k=0}^{K}$ is $(\varepsilon(K+1)/(K+2),\delta)$-DP.
	
Next, we consider the sensitivity of the BIC loss. Note that
\begin{equation*}
\begin{split}
&\sup_{\boldsymbol{x}_i,\boldsymbol{x}_i^{\prime}}|\Pi_R(\hat{\boldsymbol{w}}(k)^{\top}\boldsymbol{x}_i)\Pi_R(\boldsymbol{x}_i^{\top}\hat{\boldsymbol{w}}(k))/2-\Pi_R(\hat{\boldsymbol{w}}(k)^{\top}\boldsymbol{x}_i^{\prime})\Pi_R(\boldsymbol{x}_i^{\prime\top}\hat{\boldsymbol{w}}(k))/2|\leq R^2.
\end{split}
\end{equation*}
The BIC selection procedure returns the noisy minimizer. By Claim 3.9 in \cite{dwork2014algorithmic}, the output is $(\varepsilon/(K+2),0)$-DP. Finally, by applying the composition theorem, Algorithm \ref{alg3} is $(\varepsilon,\delta)$-DP.
\end{proof}

\subsubsection{Proofs of Lemma \ref{lem:est_w_bic}}

\begin{proof}[Proof of Lemma \ref{lem:est_w_bic}]
	
The proof follows similarly to that of Theorem \ref{thm-adapt}. Let $\hat{k}$ be the selected index corresponding to $\hat{\boldsymbol{w}}_j$ in Algorithm \ref{alg3}, and let $k^*$ denote the true parameter such that $2^{k^*-1}<\rho L^4 s_j\leq 2^{k^*}$. Note that $k^*$ is uniquely determined by $s_j$ and the constants $(\rho,L)$. By the condition $2^K>\rho L^4s_j$ in Lemma \ref{lem:est_w_bic}, the true parameter $k^*<K$ is feasible. Recall that the event
\begin{equation*}
\begin{split}
E_0:=\big\{&\inf_{\|\bu\|_0=o(n),\|\bu\|_2=1} {\bu}^{\top}\widehat{\bSigma}{\bu}\geq c_{\gamma_l}\|\bu\|_2^2,\sup_{\|\bu\|_0=o(n),\|\bu\|_2=1} {\bu}^{\top}\widehat{\bSigma}{\bu}\leq c_{\gamma_u}\|\bu\|_2^2\big\},
\end{split}
\end{equation*}
holds with high probability, as shown in Theorem \ref{thm-adapt}. We now define the event under which the truncation operator does not take effect in the estimation procedure:
\begin{equation*}
\begin{split}
E_3:=\{\max_{t=0,\dots,T-1}\max_{k=0,\dots,K}|\boldsymbol{x}_i^{\top}\boldsymbol{w}_{j,k}^{(t)}|\leq R\text{ for all }i\in S_t\big\}.
\end{split}
\end{equation*}
By the proof of Lemma \ref{lem:est_w}, the event $E_3$ occurs with high probability. Define the event
\begin{equation*}
\begin{split}
E_4&=\big\{\|\boldsymbol{w}_{j,k}^{(T)}-\boldsymbol{w}_j\|_2^2\leq c_2^{\prime}\frac{2^k\log (p)\log (n)}{n}\\
&+c_3^{\prime}\frac{2^{2k}\log (p)^2\log(1/\delta)\log(n)^7}{n^2\varepsilon^2}\text{ for all $k$ such that }2^k\geq \rho L^4s_j\big\}.
\end{split}
\end{equation*}
By Lemma \ref{lem:est_w}, and replacing $T$ with $KT$ in the term $\max_{0\leq t\leq T-1}\tilde{\boldsymbol{N}}_t$ in the proof of Lemma \ref{lem:est_w}, the event $E_4$ occurs with probability at least $1-\exp(-c_1^{\prime}\log n)$. Note that under the event $E_3\cap E_4$, for $2^k\geq\rho L^4 s_j$, we have
\begin{equation*}
\begin{split}
|\bx_i^{\top}\boldsymbol{w}_{j,k}^{(T)}|&\leq|\bx_i^{\top}\boldsymbol{w}_j|+|\bx_i^{\top}(\boldsymbol{w}_j-\boldsymbol{w}_{j,k}^{(T)})|\leq|\bx_i^{\top}\boldsymbol{w}_j|+\|\bx_i\|_{\infty}\|\boldsymbol{w}_j-\boldsymbol{w}_{j,k}^{(T)}\|_{1}\\
&\leq |\bx_i^{\top}\boldsymbol{w}_j| + c_x\sqrt{\|\boldsymbol{w}_j-\boldsymbol{w}_{j,k}^{(T)}\|_{0}}\cdot\|\boldsymbol{w}_j-\boldsymbol{w}_{j,k}^{(T)}\|_{2},
\end{split}
\end{equation*}
where the first inequality follows from the triangle inequality, the second inequality follows from Hölder's inequality, and the last inequality uses the bound $\|\cdot\|_1\leq\sqrt{\|\cdot\|_0}\times\|\cdot\|_2$. By the assumptions in Lemma \ref{lem:est_w_bic}, we have 
\begin{equation*}
    \begin{split}
        c_x\|\boldsymbol{w}_j-\boldsymbol{w}_{j,k}^{(T)}\|_{1}&\leq c_x\sqrt{c_2^{\prime}\frac{2^{2k}\log (p)\log (n)}{n}+c_3^{\prime}\frac{2^{3k}\log (p)^2\log(1/\delta)\log(n)^7}{n^2\varepsilon^2}}\\
        &\leq c_x\sqrt{c_2^{\prime}\frac{2^{2K}\log (p)\log (n)}{n}+c_3^{\prime}\frac{2^{3K}\log (p)^2\log(1/\delta)\log(n)^7}{n^2\varepsilon^2}}\\
        &=O\bigg(\sqrt{\frac{n\log(p)\log(n)}{n\log(p)^4}+\frac{n^{3/2}\log (p)^2\log(1/\delta)\log(n)^7}{n^2\varepsilon^2\log (p)^6}}\bigg)\\
        &=O\bigg(\sqrt{\frac{1}{\log(p)^2}+\frac{\log(1/\delta)\log(n)^3}{n^{1/2}\varepsilon^2}}\bigg)=o(1)
    \end{split}
\end{equation*}
and thus, for a proper choice of $R$, the parameter clipping does not occur for $2^k\geq\rho L^4s_j$.

By the oracle inequality for the BIC criterion, we obtain the following expression as a direct consequence of the selection procedure:
\begin{equation*}
\begin{split}
&\hat{\boldsymbol{w}}_j^{\top}\widehat{\bSigma}\hat{\boldsymbol{w}}_j/2-\hat{\boldsymbol{w}}_j^{\top}\boldsymbol{e}_j+c_Bf(p,\hat{k})/n+z_{\hat{k}}/n\\
&\leq \hat{\boldsymbol{w}}_j(k^*)^{\top}\widehat{\bSigma}\hat{\boldsymbol{w}}_j(k^*)/2-\hat{\boldsymbol{w}}_j(k^*)^{\top}\boldsymbol{e}_j+c_Bf(p,k^*)/n+z_{k^*}/n,
\end{split}
\end{equation*}
where we define the function $f(n,k)=2^k\log(p)\log(n)+\{2^{2k}\log (p)^2\log(1/\delta)\log(n)^7\}/(n\varepsilon^2)$, and $z_{\hat{k}}$ and $z_{k^*}$ are the noise terms added for privacy. Furthermore, by taking the maximum of the additional noise terms, we have:
\begin{equation*}
\begin{split}			&\hat{\boldsymbol{w}}_j^{\top}\widehat{\bSigma}\hat{\boldsymbol{w}}_j/2-\hat{\boldsymbol{w}}_j^{\top}\boldsymbol{e}_j+c_Bf(p,\hat{k})/n\\
&\leq \hat{\boldsymbol{w}}_j(k^*)^{\top}\widehat{\bSigma}_T\hat{\boldsymbol{w}}_j(k^*)/2-\hat{\boldsymbol{w}}_j(k^*)^{\top}\boldsymbol{e}_j+c_Bf(p,k^*)/n+\epsilon_{privacy}/n,
\end{split}
\end{equation*}
where $\epsilon_{privacy}$ is defined as $2\sup_{k=0,\dots,K}|z_k|$. By simple algebra, the above inequality implies that
\begin{equation*}
\begin{split}
(\hat{\boldsymbol{w}}_j-\hat{\boldsymbol{w}}_j(k^*))^{\top}\widehat{\bSigma}(\hat{\boldsymbol{w}}_j-\hat{\boldsymbol{w}}_j(k^*))/2
&\leq |\langle\hat{\boldsymbol{w}}_j-\hat{\boldsymbol{w}}_j(k^*),\hat{\boldsymbol{w}}_j(k^*)^{\top}\widehat{\bSigma}-\boldsymbol{e}_j\rangle|\\
&+c_B/n\{f(p,k^*)-f(p,\hat{k})\}+\epsilon_{privacy}/n.
\end{split}
\end{equation*}
Let $\widehat{U}=\text{supp}(\hat{\boldsymbol{w}}_j-\hat{\boldsymbol{w}}_j(k^*))$.
Note that
\[
|\widehat{U}|=\sqrt{n}/\log(p)^2+s_j=o(n).
\]
Hence, under the event $E_0$, we have:
\begin{align}
&  c_{\gamma_l}/2\|\hat{\boldsymbol{w}}_j-\hat{\boldsymbol{w}}_j(k^*)\|_2^2\leq(\hat{\boldsymbol{w}}_j-\hat{\boldsymbol{w}}_j(k^*))^{\top}\widehat{\bSigma}(\hat{\boldsymbol{w}}_j-\hat{\boldsymbol{w}}_j(k^*))/2\nonumber\\
\leq& |\langle\hat{\boldsymbol{w}}_j-\hat{\boldsymbol{w}}_j(k^*),\boldsymbol{w}_j^{\top}\widehat{\bSigma}-\boldsymbol{e}_j\rangle|+|\langle\hat{\boldsymbol{w}}_j-\hat{\boldsymbol{w}}_j(k^*),(\hat{\boldsymbol{w}}_j(k^*)-\boldsymbol{w}_j)^{\top}\widehat{\bSigma}\rangle|\nonumber\\
&+c_B\{f(p,k^*)-f(p,\hat{k})\}/n+\epsilon_{privacy}/n\nonumber\\
\leq& \|\hat{\boldsymbol{w}}_j-\hat{\boldsymbol{w}}_j(k^*)\|_1\|\boldsymbol{w}_j^{\top}\widehat{\bSigma}-\boldsymbol{e}_j\|_{\infty}+c_{\gamma_u}\|\hat{\boldsymbol{w}}_j-\hat{\boldsymbol{w}}_j(k^*)\|_2\|\boldsymbol{w}_j-\hat{\boldsymbol{w}}_j(k^*)\|_2\nonumber\\
&+c_B\{f(p,k^*)-f(p,\hat{k})\}/n
+\epsilon_{privacy}/n\label{eq:bic_w},
\end{align}
where the second inequality follows from the oracle inequality, and the third inequality follows from Hölder's inequality. By the proof of Lemma \ref{lem:est_w}, we have
\begin{equation*}
\|\boldsymbol{w}_j^{\top}\widehat{\bSigma}-\boldsymbol{e}_j\|_{\infty}\leq\sqrt{2\log(p)\|\boldsymbol{w}_j^{\top}\boldsymbol{x}_i\boldsymbol{x}_i^{\top}\|_{\phi_1}/(cn)}
\end{equation*}
with probability at least $1-\exp(-\log(p))$.
	
We first consider the case where $\hat{k}< k^*$. By applying Hölder's inequality to \eqref{eq:bic_w}, we obtain:
\begin{equation*}
\begin{split}
&c_{\gamma_l}/2\|\hat{\boldsymbol{w}}_j-\hat{\boldsymbol{w}}_j(k^*)\|_2^2\leq c_{\gamma_u}\|\hat{\boldsymbol{w}}_j-\hat{\boldsymbol{w}}_j(k^*)\|_2\|\hat{\boldsymbol{w}}_j(k^*)-\boldsymbol{w}_j\|_2\\
&+\|\hat{\boldsymbol{w}}_j-\hat{\boldsymbol{w}}_j(k^*)\|_2\sqrt{\frac{(2^{\hat{k}}+2^{k^*})\log(p)}{n}}\sqrt{2\|\boldsymbol{w}_j^{\top}\boldsymbol{x}_i\boldsymbol{x}_i^{\top}\|_{\phi_1}/c}\\
&+c_B\{f(p,k^*)-f(p,\hat{k})\}/n+\epsilon_{privacy}/n.
\end{split}
\end{equation*}
Under the assumption that $\hat{k}< k^*$, we have $c_B\{f(p,k^*)-f(p,\hat{k})\}+\epsilon_{privacy}>0$ when $c_B>0$. Let $\|\hat{\boldsymbol{w}}_j-\hat{\boldsymbol{w}}_j(k^*)\|_2:=t$ be treated as an unknown variable. Then the previous inequality becomes a quadratic inequality of the form $t^2\leq a_1t + a_2$, where we define:
$$a_1=2c_{\gamma_u}/c_{\gamma_l}\|\hat{\boldsymbol{w}}_j(k^*)-\boldsymbol{w}_j\|_2+1/c_{\gamma_l}\sqrt{\frac{2^{k^*}\log(p)}{n}}\sqrt{\|\boldsymbol{w}_j^{\top}\boldsymbol{x}_i\boldsymbol{x}_i^{\top}\|^2_{\psi_1}/c},$$ 
and 
$$a_2=2c_B/c_{\gamma_l}\{f(n,k^*)-f(n,\hat{k})\}/n+2/c_{\gamma_l}\epsilon_{privacy}/n.$$
By the solution to a quadratic inequality, it follows that $t\leq a_1+\sqrt{a_2}$. Furthermore, under event $E_4$ and using the fact that $ 2^{k^*}\leq \rho L^4s_j$ by the definition of $k^*$, we obtain the following upper bound for $a_1$:
\begin{equation*}
\begin{split}
a_1\leq& \frac{2c_{\gamma_u}}{c_{\gamma_l}}\sqrt{c_2^{\prime}\frac{2^{k^*}\log (p)\log (n)}{n}+c_3^{\prime}\frac{2^{2k^*}\log (p)^2\log(1/\delta)\log(n)^7}{n^2\varepsilon^2}}\\
&+1/c_{\gamma_l}\sqrt{\|\boldsymbol{w}_j^{\top}\boldsymbol{x}_i\boldsymbol{x}_i^{\top}\|^2_{\psi_1}/c}\sqrt{\frac{ 2^{k^*}\log(p)}{n}}.
\end{split}
\end{equation*}
It remains to consider the term $\sqrt{a_2}$. Since the distribution of each $z_i$ is Laplace, it is sub-exponential. Therefore, we have:
\begin{equation*}
\begin{split}
\mathbb{P}\bigg\{\epsilon_{privacy}\geq 4c\log(n)\frac{R^2(K+2)}{\varepsilon}\bigg\}
&\leq\sum_{i=0}^{K}\mathbb{P}\bigg\{|z_i|\geq 4c\log(n)\frac{R^2(K+2)}{\varepsilon}\bigg\}\\
&\leq(K+1)\exp\{-2\log(n)\}\leq\exp\{-\log(n)\}
\end{split}
\end{equation*}
By the definition of $f(n,k)$, we have
\begin{equation*}
\begin{split}
a_2/2&\leq c_B/c_{\gamma_l}\{f(n,k^*)-f(n,\hat{k})\}/n+1/c_{\gamma_l}\epsilon_{privacy}/n\\
&\leq c_B/c_{\gamma_l}f(n,k^*)/n+2c\log(n)\frac{2(4R)^2(K+2)}{\varepsilon}/(c_{\gamma_l}n)\\
&\leq c_B/c_{\gamma_l}\bigg[2^{k^*}\log (p)+\frac{2^{2k^*}\log (p)^2\log(1/\delta)\log(n)^6}{n\varepsilon^2}\bigg]\frac{\log(n)}{n}\\
&+2c\log(n)\frac{2(4R)^2(K+2)}{\varepsilon}\frac{1}{c_{\gamma_l}n}
\end{split}
\end{equation*}
Then, by combining the upper bounds of $a_1^2$ and $a_2$, we obtain:
\begin{equation*}
\begin{split}
\|\hat{\boldsymbol{w}}_j-\hat{\boldsymbol{w}}_j(k^*)\|_2^2&\leq (a_1+\sqrt{a_2})^2\leq 2a_1^2+a_2\\
&\leq c_2\frac{s\log(p)\log(n)}{n}+c_3\frac{s^2\log(p)^2\log(1/\delta)\log(n)^7}{n^2\varepsilon^2}+c_4\frac{\log(n)^3}{n\varepsilon},
\end{split}
\end{equation*}
for some constant $c_2,c_3,c_4$.
	
We now consider the case where $\hat{k}\geq k^*$. By applying the triangle inequality to \eqref{eq:bic_w}, we obtain:
\begin{align*}
&c_{\gamma_l}/2\|\hat{\boldsymbol{w}}_j-\hat{\boldsymbol{w}}_j(k^*)\|_2^2\leq c_{\gamma_u}\|\hat{\boldsymbol{w}}_j-\hat{\boldsymbol{w}}_j(k^*)\|_2\|\hat{\boldsymbol{w}}_j(k^*)-\boldsymbol{w}_j\|_2\\
&+\|\hat{\boldsymbol{w}}_j-\boldsymbol{w}_j\|_1\sqrt{2\log(p)\|\boldsymbol{w}_j^{\top}\boldsymbol{x}_i\boldsymbol{x}_i^{\top}\|_{\phi_1}/(cn)}\\
&+\|\boldsymbol{w}_j-\hat{\boldsymbol{w}}_j(k^*)\|_1\sqrt{2\log(p)\|\boldsymbol{w}_j^{\top}\boldsymbol{x}_i\boldsymbol{x}_i^{\top}\|_{\phi_1}/(cn)}+c_B\{f(p,k^*)-f(p,\hat{k})\}/n
+\epsilon_{privacy}/n.
\end{align*}
Let $\|\hat{\boldsymbol{w}}_j-\hat{\boldsymbol{w}}_j(k^*)\|_2:=t$ be treated as an unknown variable. The inequality above is a quadratic in $t$. To simplify the notation, define:
$$a_1^{\prime}=2c_{\gamma_u}/c_{\gamma_l}\|\hat{\boldsymbol{w}}_j(k^*)-\boldsymbol{w}_j\|_2,$$
and 
\begin{equation*}
\begin{split}
a_2^{\prime}=&2c_B/c_{\gamma_l}\{f(n,k^*)-f(n,\hat{k})\}/n+2/c_{\gamma_l}\epsilon_{privacy}/n\\
&+2/c_{\gamma_l}\|\hat{\boldsymbol{w}}_j-\boldsymbol{w}_j\|_1\sqrt{2\log(p)\|\boldsymbol{w}_j^{\top}\boldsymbol{x}_i\boldsymbol{x}_i^{\top}\|_{\phi_1}/(cn)}\\
&+2/c_{\gamma_l}\|\boldsymbol{w}_j-\hat{\boldsymbol{w}}_j(k^*)\|_1\sqrt{2\log(p)\|\boldsymbol{w}_j^{\top}\boldsymbol{x}_i\boldsymbol{x}_i^{\top}\|_{\phi_1}/(cn)}        
\end{split}
\end{equation*}
By the event $E_4$ and the fact that the $\|\cdot\|_1$ norm is bounded by $\|\cdot\|_2\times\sqrt{\|\cdot\|_0}$, we have :
\begin{equation*}
\begin{split}
a_2^{\prime}\leq& 4/c_{\gamma_l}\sqrt{2\frac{\|\boldsymbol{w}_j^{\top}\boldsymbol{x}_i\boldsymbol{x}_i^{\top}\|_{\phi_1}}{c}}\sqrt{\frac{2^{\hat{k}}\log(p)}{n}}\\
&\times\sqrt{c_2^{\prime}\frac{2^{\hat{k}}\log (p)\log (n)}{n}+c_3^{\prime}\frac{2^{2\hat{k}}\log (p)^2\log(1/\delta)\log(n)^7}{n^2\epsilon^2}}\\
+&c_B/c_{\gamma_l}\{f(n,k^*)-f(n,\hat{k})\}/|S_T|+1/c_{\gamma_l}\epsilon_{privacy}/|S_T|
\end{split}
\end{equation*}
For $c_B>4\sqrt{\max\{c_2^{\prime},c_3^{\prime}\}}\sqrt{2\|\boldsymbol{w}_j^{\top}\boldsymbol{x}_i\boldsymbol{x}_i^{\top}\|_{\phi_1}/c}$, then:
\begin{equation*}
a_2^{\prime}\leq c_B/c_{\gamma_l}f(n,k^*)/n+1/c_{\gamma_l}\epsilon_{privacy}/n.
\end{equation*}
Using the fact that
$$f(n,k^*)/n\leq\frac{1}{\max\{c_2^{\prime},c_3^{\prime}\}}\|\hat{\boldsymbol{w}}_j(k^*)-\boldsymbol{w}_j\|_2^2$$ 
and applying the large deviation bound for $\epsilon_{privacy}$ as used in the bound for $a_2$, we conclude:
\begin{equation*}
\begin{split}
\|\hat{\boldsymbol{w}}_j-\hat{\boldsymbol{w}}_j(k^*)\|_2^2\leq c_2\frac{s_j\log(p)\log(n)}{n}+c_3\frac{s_j^2\log(p)^2\log(1/\delta)\log(n)^7}{n^2\varepsilon^2}+c_4\frac{\log(n)^3}{n\varepsilon},
\end{split}
\end{equation*}
for some constant $c_2,c_3,c_4$.
\end{proof}

\begin{proof}[Proof of Lemma \ref{lem:three}]
We first consider the event related to truncation. By the events $E_1,E_2,E_3,E_4$ defined in the proofs of Theorem \ref{thm-adapt} and Lemma \ref{lem:est_w_bic}, we know that truncation does not occur with probability approaching one. Therefore, we omit the truncation notation in the remainder of the proof. By simple algebra, we have:
\begin{equation*}
\hat{\beta}^{(db)}_j-\beta_j= \underbrace{\hat{\boldsymbol{w}}_j^\top\frac{\sum_{i=1}^{n}\boldsymbol{x}_ie_i}{n}}_{R_{1,j}}-\underbrace{(\boldsymbol{e}_j^\top-\hat{\boldsymbol{w}}_j^\top\widehat{\boldsymbol{\Sigma}})({\boldsymbol{\beta}}-\hat{\boldsymbol{\beta}})}_{R_{2,j}}+\underbrace{z_j^{(db)}}_{R_{3,j}}.
\end{equation*}
For $R_{2,j}$,
\begin{equation*}
\begin{split}
|R_{2,j}|&\leq |(\boldsymbol{e}_j^{\top}-\boldsymbol{w}_j^\top\widehat{\bSigma})(\boldsymbol{\beta}-\hat{\boldsymbol{\beta}})|+|(\hat{\boldsymbol{w}}_j-\boldsymbol{w}_j)^{\top}\widehat{\bSigma}(\boldsymbol{\beta}-\hat{\boldsymbol{\beta}})|\\
&\leq \|\boldsymbol{e}_j-\boldsymbol{w}_j^\top\widehat{\bSigma}\|_{\infty}\|\boldsymbol{\beta}-\hat{\boldsymbol{\beta}}\|_1+c_{\gamma_u}\|\hat{\boldsymbol{w}}_j-\boldsymbol{w}_j\|_{2}\|\boldsymbol{\beta}-\hat{\boldsymbol{\beta}}\|_2\\
&\leq \|\boldsymbol{e}_j-\boldsymbol{w}_j^\top\widehat{\bSigma}\|_{\infty}\sqrt{2^{K}}\|\boldsymbol{\beta}-\hat{\boldsymbol{\beta}}\|_2+c_{\gamma_u}\|\hat{\boldsymbol{w}}_j-\boldsymbol{w}_j\|_{2}\|\boldsymbol{\beta}-\hat{\boldsymbol{\beta}}\|_2\\
&\leq \sqrt{2\|\boldsymbol{w}_j^{\top}\boldsymbol{x}_i\boldsymbol{x}_i^{\top}\|_{\phi_1}/c}\sqrt{\frac{\log(p)}{n}}\sqrt{2^K}\\
&\times\sqrt{c_2\frac{s\log (p)\log (n)}{n}+c_3\frac{s^2\log(p)^2 \log(1/\delta)\log(n)^7 }{n^2\varepsilon^2}+c_4\frac{\log(n)^3}{n\varepsilon}}\\
&+c_{\gamma_u}\sqrt{c_2\frac{s_j\log(p)\log(n)}{n}+c_3\frac{s_j^2\log(p)^2\log(1/\delta)\log(n)^7 }{n^2\varepsilon^2}+c_4\frac{\log(n)^3}{n\varepsilon}}\\
&\times\sqrt{c_2\frac{s\log(p)\log(n)}{n}+c_3\frac{s^2\log(p)^2\log(1/\delta)\log(n)^7 }{n^2\varepsilon^2}+c_4\frac{\log(n)^3}{n\varepsilon}},
\end{split}
\end{equation*}
we use the triangle inequality in the first inequality; Hölder’s inequality and event $E_0$ in the second inequality; the inequality $\|\cdot\|_1\leq\sqrt{\|\cdot\|_0}\times\|\cdot\|_2$ and the bound $\|\boldsymbol{\beta}-\hat{\boldsymbol{\beta}}\|_0\leq 2^K$ in the third inequality; and results from Theorem~\ref{thm-adapt} and Lemma~\ref{lem:est_w_bic} in the final inequality. Let
$$r_n=s_0\log(p)\log(n)/n^{1/2}+s_0^{2}\log(p)^2\log(1/\delta)\log(n)^7/(n^{1.5}\varepsilon^2)+\log(n)^3/(n^{1/2}\varepsilon).$$ 
Then, we have
\begin{equation*}
    |R_{2,j}|=O\bigg(\sqrt{\frac{\log(p)}{n}\times\frac{\sqrt{n}}{\log(p)^2}\times n^{-1/2}r_n}+n^{-1/2}r_n\bigg)=O(n^{-1/2}\sqrt{r_n}+n^{-1/2}r_n).
\end{equation*}

The term $R_{1,j}$ is asymptotically normal when $\boldsymbol{\hat{w}}_j$ is replaced by $\boldsymbol{w}_j$. Let 
$$R_{1,j}^*=\boldsymbol{w}_j^{\top}\frac{\sum_{i=1}^{n}\boldsymbol{x}_ie_i}{n}.$$
Then we have
\[
\textup{Var}(R_{1,j}^{*})=\frac{\Omega_{j,j}\sigma^2}{n},
\]
where $\Omega_{j,j}$ denotes the $(j,j)$-th entry of $\bSigma^{-1}$. It remains to bound the difference between $R_{1,j}^*$ and $R_{1,j}$. We have:
\begin{equation*}
\begin{split}
&|R_{1,j}^*-R_{1,j}|\leq\|\hat{\boldsymbol{w}}_j-\boldsymbol{w}\|_1\|\frac{\sum_{i=1}^{n}\boldsymbol{x}_i e_i}{n}\|_{\infty}\leq\sqrt{\|\hat{\boldsymbol{w}}_j-\boldsymbol{w}\|_0}\cdot\|\hat{\boldsymbol{w}}_j-\boldsymbol{w}\|_2\cdot\big\|\frac{\sum_{i=1}^{n}\boldsymbol{x}_i e_i}{n}\big\|_{\infty}\\
&\leq\sqrt{2^K}\sqrt{c_2\frac{s_j\log p\log n}{n}+c_3\frac{s_j^2\log(p)^2 \log(1/\delta)\log(n)^7}{n^2\varepsilon^2}+c_4\frac{\log(n)^3}{n\varepsilon}}\\
&\times\sqrt{2\|\boldsymbol{x}_ie_i\|^2_{\psi_1}/c\frac{\log(p)}{n}}\\
&=O\bigg(\sqrt{\frac{\log(p)}{n}\times\frac{\sqrt{n}}{\log(p)^2}\times n^{-1/2}r_n}\bigg),
\end{split}
\end{equation*}
where we use Hölder’s inequality in the first step; the inequality $\|\cdot\|_1\leq\sqrt{\|\cdot\|_0}\times\|\cdot\|_2$ and the bound $\|\boldsymbol{\beta}-\hat{\boldsymbol{\beta}}\|_0\leq 2^K$ in the second step; and Lemma~\ref{lem:est_w_bic} together with the high-probability bound for $\|\sum_{i=1}^{n}\boldsymbol{x}_i e_i/n\|_{\infty}$, which appears in the proof of Theorem \ref{thm-adapt}, in the final step. 
	
It remains to combine all the terms. The combined quantity $\sqrt{n}|R_{2,j}|+\sqrt{n}|R_{1,j}-R_{1,j}^*|$ is of the order $O_p(\max(r_n^{1/2},r_n))$. Furthermore, we consider the term $z_j^{(db)}$. By Markov's inequality, we have
\begin{equation*}
z_j^{(db)}=O_p\bigg(\frac{R^2}{\varepsilon n}\sqrt{\log(1/\delta)}\bigg),
\end{equation*}
which is typically $o_p(n^{-1/2})$ under the regularity conditions in Theorem \ref{thm:main}.
\end{proof}

\begin{proof}[Proof of Theorem \ref{thm:main}]
We first establish the privacy guarantee. By Lemma~\ref{lem:privacy-adapt} and Lemma~\ref{lem:privacy_adap_w}, the first two steps of Algorithm~\ref{algo:inference} are each $(\varepsilon/4,\delta/4)$-DP. By the composition theorem, it remains to show that Steps 3 and 4 are also $(\varepsilon/4,\delta/4)$-DP, respectively.
	
The sensitivity of $\sum_{i=1}^{n}\Pi_R(\hat{\boldsymbol{w}}_j^{\top}\boldsymbol{x}_i)\Pi_R(y_i)/n-\sum_{i=1}^{n} \Pi_R(\hat{\boldsymbol{w}}_j^{\top}\boldsymbol{x}_i)\Pi_R(\boldsymbol{x}_i^{\top}\hat{\boldsymbol{\beta}})/n$ is bounded by:
\begin{equation*}
\begin{split}
&\sup_{(\boldsymbol{x}_i,y_i),(\boldsymbol{x}_i^{\prime},y_i^{\prime})}\frac{1}{n}|\Pi_R(\hat{\boldsymbol{w}}_j^{\top}\boldsymbol{x}_i)\Pi_R(y_i)-\Pi_R(\hat{\boldsymbol{w}}_j\boldsymbol{x}_i)\Pi_R(\boldsymbol{x}_i^{\top}\hat{\boldsymbol{\beta}})\\
&-\Pi_R(\hat{\boldsymbol{w}}_j^{\top}\boldsymbol{x}_i^{\prime})\Pi_R(y_i^{\prime})+\Pi_R(\hat{\boldsymbol{w}}_j\boldsymbol{x}_i^{\prime})\Pi_R(\boldsymbol{x}_i^{\prime\top}\hat{\boldsymbol{\beta}})|\leq\frac{2}{n}(R^2+R^2).
\end{split}
\end{equation*}
The sensitivity of $\frac{1}{n}\sum_{i=1}^n(\Pi_R(y_i)-\Pi_R(\boldsymbol{x}_i^\top\hat{\boldsymbol{\beta}}))^2$ is bounded by:
\begin{equation*}
\begin{split}
&\sup_{(\boldsymbol{x}_i,y_i),(\boldsymbol{x}_i^{\prime},y_i^{\prime})}\frac{1}{n}|(\Pi_R(y_i)-\Pi_R(\boldsymbol{x}_i^\top\hat{\boldsymbol{\beta}}))^2-(\Pi_R(y^{\prime}_i)-\Pi_R(\boldsymbol{x}_i^{\prime\top}\hat{\boldsymbol{\beta}}))^2|\leq \frac{2}{n}(R+R)^2.
\end{split}
\end{equation*}
Therefore, Steps 3 and 4 are each $(\varepsilon/4,\delta/4)$-DP by the Gaussian mechanism. Finally, by the composition theorem, Algorithm \ref{algo:inference} is $(\varepsilon,\delta)$-DP.
	
We now establish the validity of the proposed confidence interval. Note that under the additional order conditions in Theorem \ref{thm:main}, and by Lemma \ref{lem:three}, we have:
\begin{equation*}
\sqrt{n}(\hat{\beta}^{(db)}_j-\beta_j)\stackrel{d}{\to}N(0,\Omega_{jj}\sigma^2).
\end{equation*}
By Lemma \ref{lem:est_w_bic}, the $\ell_2$-convergence of $\hat{\boldsymbol{w}}_j$ implies its $\ell_{\infty}$-convergence, and consequently, $\hat{w}_{j,j}\stackrel{p}{\to}w_{j,j}$. It remains to consider the estimation of $\sigma^2$. We first address the event of truncation. By event $E_1$, defined in the proof of Theorem \ref{thm:main}, truncation does not occur with probability approaching one. Therefore, we omit the truncation notation in the remainder of the proof. Then we have:
\begin{equation*}
\begin{split}
&\hat{\sigma}^2-\sigma^2=\frac{1}{n}\sum_{i=1}^n(y_i-\boldsymbol{x}_i^\top\hat{\boldsymbol{\beta}})^2+z-\sigma^2\\
=&\frac{1}{n}\sum_{i=1}^n(y_i-\boldsymbol{x}_i^\top\boldsymbol{\beta})^2-\sigma^2+\frac{1}{n}\sum_{i=1}^n(\boldsymbol{x}_i^\top\hat{\boldsymbol{\beta}}-\boldsymbol{x}_i^\top\boldsymbol{\beta})^2-\frac{2}{n}\sum_{i=1}^{n}(y_i-\boldsymbol{x}_i^\top\boldsymbol{\beta})\times(\boldsymbol{x}_i^\top\hat{\boldsymbol{\beta}}-\boldsymbol{x}_i^\top\boldsymbol{\beta})+z,
\end{split}
\end{equation*}
where $z$ is the noise added in Step 4 of the algorithm. We have the convergence $\sum_{i=1}^n(y_i-\boldsymbol{x}_i^\top\boldsymbol{\beta})^2/n-\sigma^2=o_p(1)$ by the weak law of large numbers. For the second term, we observe:
\begin{equation*}
\begin{split}
&\frac{1}{n}\sum_{i=1}^n(\boldsymbol{x}_i^\top\hat{\boldsymbol{\beta}}-\boldsymbol{x}_i^\top\boldsymbol{\beta})^2=(\hat{\boldsymbol{\beta}}-\boldsymbol{\beta})^{\top}\widehat{\bSigma}(\hat{\boldsymbol{\beta}}-\boldsymbol{\beta})\leq c_{\gamma_u}\|\hat{\boldsymbol{\beta}}-\boldsymbol{\beta}\|_2^2=o_p(1),
\end{split}
\end{equation*}
where the first equality follows from simple algebra, the inequality uses event $E_0$, and the convergence follows from Theorem \ref{thm-adapt}. The remaining term satisfies:
\begin{equation*}
\begin{split}
&|\frac{2}{n}\sum_{i=1}^{n}(y_i-\boldsymbol{x}_i^\top\boldsymbol{\beta})(\boldsymbol{x}_i^\top\hat{\boldsymbol{\beta}}-\boldsymbol{x}_i^\top\boldsymbol{\beta})|=|\frac{2}{n}\sum_{i=1}^{n}e_i(\boldsymbol{x}_i^\top\hat{\boldsymbol{\beta}}-\boldsymbol{x}_i^\top\boldsymbol{\beta})|\leq 2\|\frac{1}{n}\sum_{i=1}^{n}e_i\boldsymbol{x}_i\|_{\infty}\|\boldsymbol{\beta}-\hat{\boldsymbol{\beta}}\|_1\\
&\leq 2\|\frac{1}{n}\sum_{i=1}^{n}e_i\boldsymbol{x}_i\|_{\infty}\cdot\|\boldsymbol{\beta}-\hat{\boldsymbol{\beta}}\|_2\cdot\sqrt{\|\boldsymbol{\beta}-\hat{\boldsymbol{\beta}}\|_0}\\
&\leq \sqrt{2\|\boldsymbol{x}_ie_i\|_{\phi_1}/c}\sqrt{\frac{\log(p)}{n}}\sqrt{2^K}\\
&\times\sqrt{c_2\frac{s\log (p)\log (n)}{n}+c_3\frac{s^2\log(p)^2 \log(1/\delta)\log(n)^7 }{n^2\varepsilon^2}+c_4\frac{\log(n)^3}{n\varepsilon}}\\
&=O\bigg(\sqrt{\frac{\log(p)}{n}\times\frac{\sqrt{n}}{\log(p)^2}\times n^{-1/2}r_n}\bigg)=o_p(1),
\end{split}
\end{equation*}
where the first equality follows from the definition of the linear model, the second inequality applies Hölder’s inequality, and the third inequality uses the fact that $\|\cdot\|_1\leq\sqrt{\|\cdot\|_0}\cdot\|\cdot\|_2$. The final convergence follows from Theorem \ref{thm-adapt} and the conditions in Theorem \ref{thm:main}. Therefore, we conclude that $\hat{\sigma}^2\stackrel{p}{\to}\sigma^2$, and the final result follows by Slutsky’s theorem.
\end{proof}

\subsection{Proof of FDR}
\begin{proof}[Proof of Lemma \ref{lem:privacy_fdr}]
By Lemma \ref{lem:privacy-adapt}, releasing $\tilde{\boldsymbol{\beta}}_{(1)}$ satisfies $(\varepsilon,\delta)$-DP. By the composition theorem, it remains to show that the DP-OLS procedure also satisfies $(\varepsilon,\delta)$-DP. Note that the DP-OLS procedure estimates the numerator and denominator separately. We first compute the $\ell_2$ sensitivity of the denominator, $\sum_{i\in \mathcal{D}_2}\boldsymbol{x}_{i,\mathcal{A}}\boldsymbol{x}_{i,\mathcal{A}}^{\top}/|\mathcal{D}_2|$:
$$\sup_{(\boldsymbol{x}_i,\boldsymbol{x}_i^{\prime})}\|\boldsymbol{x}_{i,\mathcal{A}}\boldsymbol{x}_{i,\mathcal{A}}^{\top}-\boldsymbol{x}_{i,\mathcal{A}}^{\prime}\boldsymbol{x}_{i,\mathcal{A}}^{\prime\top}\|_2/n_2\leq 2|\mathcal{A}|c_x^2/n_2,$$
where we use the notation $n_2 = |\mathcal{D}_2|$. Note that the sparsity level $|\mathcal{A}|$ is bounded by $2^K$ according to Algorithm \ref{alg:linear}. The $\ell_2$ sensitivity of the numerator, $\sum_{i\in \mathcal{D}_2}\boldsymbol{x}_{i,\mathcal{A}}\Pi_R(y_i)/|A_2|$, satisfies:
\begin{equation*}
\begin{split}
&\sup_{(\boldsymbol{x}_i,y_i),(\boldsymbol{x}_i^{\prime},y_i^{\prime})}\|\boldsymbol{x}_{i,\mathcal{A}}\Pi_R(y_i)-\boldsymbol{x}_{i,\mathcal{A}}^{\prime}\Pi_R(y_i^{\prime})\|_2/n_2\leq 2\sqrt{|\mathcal{A}|}c_xR/n_2.
\end{split}
\end{equation*}
The DP-OLS procedure satisfies $(\varepsilon,\delta)$-DP by the Gaussian mechanism and the composition theorem. The output of Algorithm \ref{alg:fdr} is a deterministic function of $\tilde{\boldsymbol{\beta}}_{(1)}$ and $\tilde{\boldsymbol{\beta}}_{(2)\mathcal{A}}$, and is therefore $(2\varepsilon,2\delta)$-DP by the post-processing property.
\end{proof}

\begin{proof}[Proof of Theorem \ref{thm:fdr}]
We first consider the truncation operators. Note that under the event $E_1$, truncation does not occur. Thus, we omit truncation in the following analysis. Given the signal strength condition, for each $i\in\mathcal{S}$, we have
\begin{equation*}
    |\hat{\beta}_i|\geq|\beta_i|-|\hat{\beta}_i-\beta_i|\geq\min_{j\in\mathcal{S}}|\beta_j|-\|\hat{\boldsymbol{\beta}}-\boldsymbol{\beta}\|_{2}\geq\min_{j\in\mathcal{S}}|\beta_j|/2>0.
\end{equation*}
Thus, the sure screening property holds with probability approaching one; that is, $\mathbb{P}(\mathcal{S}\subseteq\mathcal{A})\to 1$ as $n\to\infty$. Furthermore, since we only need to consider the subset $\mathcal{A}\subseteq[p]$, we omit the subscript $\mathcal{A}$ to simplify the notation. Without loss of generality, we assume $|\mathcal{D}_2|=n_2=n/2$. 

Note that the DP-OLS estimator $\tilde{\boldsymbol{\beta}}_{(2)}$ satisfies:
\begin{equation*}
\begin{split}
\tilde{\boldsymbol{\beta}}_{(2)}-\boldsymbol{\beta}&=\boldsymbol{\tilde{\Sigma}}^{-1}_{(2)XX}\big(\frac{1}{n/2}\sum_{i\in\mathcal{D}_2}\boldsymbol{x}_{i}e_i+\boldsymbol{N}_{XY}\big)+(\boldsymbol{\tilde{\Sigma}}^{-1}_{(2)XX}-\widehat{\bSigma}^{-1}_{(2)XX})\widehat{\bSigma}_{(2)XX}\boldsymbol{\beta}\\			&:=\tilde{\boldsymbol{\beta}}_{(2)}^{(0)}+\tilde{\boldsymbol{\beta}}_{(2)}^{(1)},
\end{split}
\end{equation*}
where we use the definition of a linear model and the following notations
\begin{equation*}
\widehat{\bSigma}_{(2)XX}:=\sum_{i\in\mathcal{D}_2}\boldsymbol{x}_i\boldsymbol{x}_i^{\top}/|\mathcal{D}_2|,\boldsymbol{\tilde{\Sigma}}_{(2)XX}:=\sum_{i\in\mathcal{D}_2}\boldsymbol{x}_i\boldsymbol{x}_i^{\top}/|\mathcal{D}_2|+\boldsymbol{N}_{XX}.
\end{equation*}
The decomposition of the DP-OLS estimator differs from that of the OLS estimator. Due to the additional noise added to the denominator, the DP-OLS estimator is closely related to the ridge regression estimator. Consequently, the term $\tilde{\boldsymbol{\beta}}_{(2)}^{(1)}$ represents the bias component in the DP-OLS estimator.
	
The following proof relies on two critical observations: (1) conditional on $\mathcal{D}_1\cup\{\boldsymbol{x}_i\}_{i\in\mathcal{D}_2}$, the distribution of $\tilde{\boldsymbol{\beta}}_{(2)}^{(0)}$ is symmetric around $0$; and (2) the term $\tilde{\boldsymbol{\beta}}_{(2)}^{(1)}$, representing the bias, is small. The first observation can be justified as follows. Conditional on the first part of the data $\mathcal{D}_1$, the active set $\mathcal{A}$ selection is fixed. Furthermore, conditional on the covariates $\{\boldsymbol{x}_i\}_{i\in\mathcal{D}_2}$ and the added noise $\boldsymbol{N}_{XX}$, the distribution of $\tilde{\boldsymbol{\beta}}_{(2)}^{(0)}$ is symmetric around zero because it is a linear combination of independent Gaussian random variables. Therefore, conditional on $\{\boldsymbol{x}_i\}_{i\in\mathcal{D}_2}$, the distribution of $\tilde{\boldsymbol{\beta}}_{(2)}^{(0)}$ is a weighted mixture of distributions symmetric around zero and is itself symmetric around zero. Without loss of generality, we assume that $|\mathcal{A}|=\hat{s}\to\infty$; otherwise, the FDR control problem becomes trivial.
	
Define the variable $\boldsymbol{R}$ and its normalized version $\boldsymbol{R}^0$ as follows:
\begin{equation*}
\begin{split}
\boldsymbol{R}&=\big(\widehat{\bSigma}_{(2)XX}+\boldsymbol{N}_{XX}\big)^{-1}(\sigma^2\widehat{\bSigma}_{(2)XX}+\sigma^2_r\boldsymbol{I}_p)\\
&\big(\widehat{\bSigma}_{(2)XX}+\boldsymbol{N}_{XX}\big)^{-1}\\
&=:\boldsymbol{ABA},
\end{split}
\end{equation*}
and
\begin{equation*}
\boldsymbol{R}^0:=\{R_{ij}^{0}\}\text{ for }R_{ij}^{0}=\frac{R_{ij}}{\sqrt{R_{ii}R_{jj}}},
\end{equation*}
where $\sigma^2$ is the variance of the residual $e_i$ and $\sigma^2_r=B_2^2\cdot8\log(2.5/\delta)/\varepsilon^2$ is defined in Algorithm \ref{alg:fdr}. Without loss of generality, we assume $\sigma^2=1$. Here, $R_{ij}^0$ represents the conditional correlation between the $i$-th and $j$-th components of the DP-OLS regression coefficients $\tilde{\boldsymbol{\beta}}_{(2)}$. By Wigner's semicircle law, the maximum eigenvalue of $\boldsymbol{N}_{XX}$ converges to $0$ with high probability:
\begin{equation*}
c\sqrt{\hat{s}}\hat{s}\frac{\sqrt{\log(1/\delta)}}{n\varepsilon}=o(1),
\end{equation*}
where we use the fact that the sparsity is denoted by $|\mathcal{A}|=\hat{s}$. Furthermore, by the order condition in Theorem \ref{thm:fdr}, we have $\sigma_r^2=o(1)$. By the proof of Lemma \ref{lem:est_w} and Weyl’s theorem, we have
\begin{equation*}
\lambda_j(\boldsymbol{A}-\boldsymbol{\Sigma}^{-1}_{\mathcal{A}\mathcal{A}})=o_p(1)\text{ and }\lambda_j(\boldsymbol{B}-\boldsymbol{\Sigma}_{\mathcal{A}\mathcal{A}})=o_p(1),
\end{equation*}
for $j=1,\dots,|\mathcal{A}|$, where $\lambda_j(\cdot)$ denotes the $j$-th eigenvalue. Here, $\boldsymbol{\Sigma}_{\mathcal{A}\mathcal{A}}^{-1}$ is a sub-matrix of $\boldsymbol{\Sigma}^{-1}$, and $\boldsymbol{\Sigma}_{\mathcal{A}\mathcal{A}}$ is a sub-matrix of $\boldsymbol{\Sigma}$. Moreover, we have
\begin{equation*}
\lambda_{min}(\boldsymbol{R})=1/L+o_p(1)\text{ and }\lambda_{max}(\boldsymbol{R})=L+o_p(1),
\end{equation*}
where we use the Condition \ref{cond:1}. Therefore,
\begin{equation*}
\begin{split}
&\|\boldsymbol{R}^{0}\|_{l,1}\leq\lambda_{\min}^{-1}(\boldsymbol{R})\|\boldsymbol{R}\|_{l,1}\leq\lambda_{\min}(\boldsymbol{R})^{-1}\hat{s}\|\boldsymbol{R}\|_{l,2}\\
&\leq\lambda_{\max}(\boldsymbol{R})\lambda_{\min}(\boldsymbol{R})^{-1}\hat{s}^{3/2}=O_p(\hat{s}^{3/2}),
\end{split}
\end{equation*}
where $\|\boldsymbol{R}\|_{l,p}:=(\sum_{i,j=1}^{m}|R_{ij}^p|)^{1/p}$ denotes the element-wise matrix norm. We use the relation $\sqrt{R_{ii}R_{jj}}\geq \min_{i}R_{ii}\geq\lambda_{min}(\boldsymbol{R})$ in the first inequality, the inequality $\|\cdot\|_{l,1}\leq\sqrt{\|\cdot\|_{l,0}}\cdot\|\cdot\|_{l,2}$ in the second inequality, and the relation $\|\boldsymbol{R}\|_{l,2}\leq\sqrt{\hat{s}\max_i\sum_{j=1}^{\hat{s}}R_{ij}^2}\leq\sqrt{\hat{s}}\lambda_{max}(\boldsymbol{R})$ in the last inequality. In addition, we have $\|\boldsymbol{R}^{0}_{\bar{\mathcal{S}}\cap\mathcal{A}}\|_{l,1}\leq O_p(p_0^{3/2})$, where $p_0=|\bar{\mathcal{S}}\cap\mathcal{A}|$. These results will be used to bound the correlations.
	
For any threshold $t\in\mathbb{R}$, we define
\begin{equation*}
\hat{G}_p^0(t)=\frac{1}{p_0}\sum_{j\in\bar{\mathcal{S}}\cap\mathcal{A}}1(M_j>t)\text{, }G_p^0(t)=\frac{1}{p_0}\sum_{j\in\bar{\mathcal{S}}\cap\mathcal{A}}\mathbb{P}(M_j>t);
\end{equation*}
\begin{equation*}
\hat{G}_p^1(t)=\frac{1}{p_1}\sum_{j\in\mathcal{S}\cap\mathcal{A}}1(M_j>t)\text{, }\hat{V}_p^0(t)=\frac{1}{p_0}\sum_{j\in\bar{\mathcal{S}}\cap\mathcal{A}}1(M_j<-t),
\end{equation*}
where $p_0=|\bar{\mathcal{S}}\cap\mathcal{A}|$ and $p_1=\hat{s}-p_0$. Let $r_p=p_1/p_0$ and
\begin{equation*}
\begin{split}
&\text{FDP}(t)=\frac{\hat{G}_p^0(t)}{\hat{G}_p^0(t)+r_p\hat{G}_p^1(t)}\text{, }\text{FDP}^{s}(t)=\frac{\hat{V}_p^0(t)}{\hat{G}_p^0(t)+r_p\hat{G}_p^1(t)},\\			&\text{and }\text{FDP}^{e}(t)=\frac{G_p^0(t)}{G_p^0(t)+r_pG_p^1(t)}.
\end{split}
\end{equation*}
It is easy to see $G_p^0(t)=\E\{\hat{G}_p^0(t)\}$. Furthermore, we have
\begin{equation*}
\begin{split}
&\text{Var}\{\hat{G}_p^0(t)\}=\frac{1}{p_0^2}\sum_{j\in\bar{\mathcal{S}}\cap\mathcal{A}}\text{Var}\{1(M_j>t)\}+\frac{1}{p_0^2}\sum_{i,j\in\bar{\mathcal{S}}\cap\mathcal{A};i\neq j}\text{Cov}\{1(M_i>t),1(M_j>t)\}.
\end{split}
\end{equation*}
The first term is bounded by $(1/4)/p_0\stackrel{n\to\infty}{\longrightarrow}0$. Without loss of generality, we assume $\tilde{\beta}_{(1),i}>0$ and $\tilde{\beta}_{(1),j}>0$, where $\tilde{\beta}_{(1),i}$ denotes the estimate from the first part of the data, $\mathcal{D}_1$. By definition, the function $f(u,v)$ is non-negative, symmetric in $u$ and $v$, and monotonically increasing in both arguments. Therefore, there exists a function $\mathcal{I}_t(u)$, defined by $\mathcal{I}_t(u)=\inf\{v\geq 0:f(u,v)>t\}$, such that for $\mathcal{I}_t(\tilde{\beta}_{(1),i})$ and $\mathcal{I}_t(\tilde{\beta}_{(1),j})$, we have,
\begin{equation*}
\mathbb{P}(M_i>t,M_j>t)=\mathbb{P}\{\tilde{\beta}_{(2),i}>\mathcal{I}_t(\tilde{\beta}_{(1),i}),\tilde{\beta}_{(2),j}>\mathcal{I}_t(\tilde{\beta}_{(1),j})\}.
\end{equation*}
Note that the bias satisfies
\begin{equation*}
\begin{split}
&\sqrt{n}\|\tilde{\boldsymbol{\beta}}_{(2)}^{(1)}\|_{\infty}\leq \sqrt{n}\|\tilde{\boldsymbol{\beta}}_{(2)}^{(1)}\|_{2}\leq\sqrt{n}\|\boldsymbol{\tilde{\Sigma}}^{-1}_{(2)XX}-\widehat{\bSigma}^{-1}_{(2)XX}\|_2\|\widehat{\bSigma}_{(2)XX}\|_2\|\boldsymbol{\beta}\|_2\\
&=O_p(\sqrt{n}\frac{\hat{s}^{3/2}\sqrt{\log(1/\delta)}}{n\varepsilon})=o_p(1),
\end{split}
\end{equation*}
where we use Wigner’s semicircle law and the conditions in Theorem \ref{thm:fdr}. By the Lipschitz continuity of $f(u,v)$, we have
\begin{equation*}
\begin{split}
&\mathbb{P}\{\tilde{\beta}_{(2)i}>\mathcal{I}_t(\tilde{\beta}_{(1)i}),\tilde{\beta}_{(2)j}>\mathcal{I}_t(\tilde{\beta}_{(1)j})\}=\mathbb{P}\{\tilde{\beta}_{(2)i}-\tilde{\beta}_{(2)i}^{(1)}\\
&>\mathcal{I}_t(\tilde{\beta}_{(1)i}),\tilde{\beta}_{(2)j}-\tilde{\beta}_{(2)j}^{(1)}>\mathcal{I}_t(\tilde{\beta}_{(1)j})\}+o_p(1).
\end{split}
\end{equation*}
Note that the joint distribution of $(\tilde{\beta}_{(2)i}^{(0)},\tilde{\beta}_{(2)j}^{(0)})$ is bivariate normal conditional on $\boldsymbol{N}_{XX}$. By Theorem 1 in \cite{azriel2015empirical}, for any $t_1,t_2\in\mathbb{R}$,
\begin{equation*}
\begin{split}
&\mathbb{P}(\tilde{\beta}_{(2)i}-\tilde{\beta}_{(2)i}^{(0)}>t_1,\tilde{\beta}_{(2)j}-\tilde{\beta}_{(2)j}^{(0)}>t_2)\\
&-\mathbb{P}(\tilde{\beta}_{(2)i}-\tilde{\beta}_{(2)i}^{(0)}>t_1)\mathbb{P}(\tilde{\beta}_{(2)j}-\tilde{\beta}_{(2)j}^{(0)}>t_2)\leq O(|R_{ij}^0|).
\end{split}
\end{equation*}
Therefore,
\begin{equation*}
\begin{split}
&\frac{1}{p_0^2}\sum_{i,j\in\bar{\mathcal{S}}\cap\mathcal{A};i\neq j}\text{Cov}\{1(M_i>t),1(M_j>t)\}\leq O_p(p_0^{-2}\|R^0_{\bar{\mathcal{S}}\cap\mathcal{A}}\|_1)+o_p(1)\\
\leq&  O_p(p_0^{-2}p_0^{3/2})+o_p(1)\stackrel{n\to\infty}{\longrightarrow}0.
\end{split}
\end{equation*}
By Markov's inequality,
\begin{equation*}
|\hat{G}_p^0(t)-G_p^0(t)|\to 0.
\end{equation*}
Similarly, we have
\begin{equation*}
|\hat{V}_p^0(t)-\frac{1}{p_0}\sum_{j\in\bar{\mathcal{S}}\cap\mathcal{A}}\mathbb{P}(M_j<-t)|\to 0.
\end{equation*}
Using the fact that the bias satisfies $\|\tilde{\boldsymbol{\beta}}_{(2)}^{(1)}\|_{\infty}=o_p(n^{-1/2})$, we obtain
\begin{equation*}
\begin{split}
\mathbb{P}(M_j>t)&=\mathbb{P}\{\tilde{\beta}_{(2),j}>\mathcal{I}_t(\tilde{\beta}_{(1),j})\}=\mathbb{P}\{\tilde{\beta}_{(2),j}^{(1)}>\mathcal{I}_t(\tilde{\beta}_{(0),j})\}+o_p(1)\\
&=\mathbb{P}\{\tilde{\beta}_{(2),j}^{(0)}< -\mathcal{I}_t(\tilde{\beta}_{(1),j})\}+o_p(1),
\end{split}
\end{equation*}
for $j\in\bar{\mathcal{S}}\cap\mathcal{A}$, where the last equality follows from the symmetry of $\tilde{\beta}_{(2),k}^{(0)}$ under the null. Therefore, we conclude that
\begin{equation*}
|\hat{V}_p^0(t)-G_p^0(t)|\to 0
\end{equation*}
	
Thus, by algebra, we have
\begin{equation*}
\sup_{0\leq t\leq 1}|\text{FDP}(t)-\text{FDP}^s(t)|=o_p(1).
\end{equation*}
For any $c\in (0,q)$ and $t_{q-c}$ satisfying $\mathbb{P}\{\text{FDP}^s(t_{q-c})\leq q-c\}\to 1$, we obtain 
\begin{equation}
\begin{split}
\label{eq:fdr_2}
&\mathbb{P}(\tau_q\leq t_{q-c})\geq\mathbb{P}\{\text{FDP}^s(t_{q-c})\leq q\}\\
&\geq\mathbb{P}\{\text{FDP}(t_{q-c})\leq q-c,|\text{FDP}^s(t_{q-c})-\text{FDP}(t_{q-c})|\leq c\}\\
&\geq 1-c+o(1),
\end{split}
\end{equation}
where the first inequality follows from the definition of $t_q$. It then follows that
\begin{equation*}
\begin{split}
\limsup_{n\to\infty}\E\{\text{FDP}(\tau_q)\}
&\leq\limsup_{n\to\infty}\E\{\text{FDP}(\tau_q)\mid\tau_q\leq t_{q-c}\}\mathbb{P}(\tau_q\leq t_{q-c})+\mathbb{P}(\tau_q>t_{q-c})\\
&\leq\limsup_{n\to\infty}\E\{\text{FDP}^s(\tau_q)\mid\tau_q\leq t_{q-c}\}\mathbb{P}(\tau_q\leq t_{q-c})\\
&+\limsup_{n\to\infty}\E\{|\text{FDP}(\tau_q)-\text{FDP}^e(\tau_q)|\mid\tau_q\leq t_{q-c}\}\mathbb{P}(\tau_q\leq t_{q-c})\\
&+\limsup_{n\to\infty}\E\{|\text{FDP}^s(\tau_q)-\text{FDP}^e(\tau_q)|\mid\tau_q\leq t_{q-c}\}\mathbb{P}(\tau_q\leq t_{q-c})+c\\
&\leq\limsup_{n\to\infty}\E\{\text{FDP}^s(\tau_q)\}+\limsup_{n\to\infty}\E\{|\text{FDP}(\tau_q)-\text{FDP}^e(\tau_q)|\}\\
&+\limsup_{n\to\infty}\E\{|\text{FDP}^s(\tau_q)-\text{FDP}^e(\tau_q)|\}+c+o(1)\leq q+c+o(1),
\end{split}
\end{equation*}
where the first inequality uses the law of total expectation, the second applies the triangle inequality, and the third follows from inequality \eqref{eq:fdr_2}.

So far, we have shown that $\limsup_{n\to\infty}\text{FDR}(\tau_q)\leq q$. Next, we consider the power of the procedure. Recall that the bias of the DP-OLS estimator satisfies:
\begin{equation*}
\begin{split}
&\|\tilde{\boldsymbol{\beta}}_{(2)}-\boldsymbol{\beta}\|_{\infty}\leq\|\tilde{\boldsymbol{\beta}}_{(2)}^{(0)}\|_{\infty}+\|\tilde{\boldsymbol{\beta}}_{(2)}^{(1)}\|_{\infty}\\			&\leq\|\tilde{\boldsymbol{\beta}}_{(2)}^{(0)}\|_{\infty}+\|\tilde{\boldsymbol{\beta}}_{(2)}^{(1)}\|_{2}\\
&\leq O\{\sqrt{\log(\hat{s})/n}+\sqrt{\hat{s}\log(1/\delta)\log(n)}/(n\varepsilon)+\hat{s}^{3/2}\sqrt{\log(1/\delta)}/(n\varepsilon)\},
\end{split}
\end{equation*}
with probability at least $1-\exp(-\log(n))$, where we apply the large deviation theorem to $\sum_{i\in\mathcal{D}_2}\boldsymbol{x}_{i}e_i/(n/2)$ and $\boldsymbol{N}_{XY}$, and use the bias bound of $\boldsymbol{\tilde{\beta}}^{(1)}_{(2)}$. Under the additional signal strength condition, we have
\begin{equation*}
\min_{i\in\mathcal{S}\cap\mathcal{A}}|\hat{\beta}_{2,i}|\geq\max_{i\in\bar{\mathcal{S}}\cap\mathcal{A}}|\hat{\beta}_{2,i}|.
\end{equation*}
By the sure screening property, we also have
\begin{equation*}
\min_{i\in\mathcal{S}}|\hat{\beta}_{1,i}|\geq\max_{i\in\bar{\mathcal{S}}}|\hat{\beta}_{1,i}|.
\end{equation*}
By the definition of $f(u,v)$, it follows that $\min_{i\in\mathcal{S}\cap\mathcal{A}}|M_i|\geq\max_{i\in\bar{\mathcal{S}}\cap\mathcal{A}}|M_i|$. Consequently, we have $\hat{G}^1_p(\tau_q)\to 1$ with probability approaching one, and therefore, the power asymptotically converges to one.
\end{proof}

\section{Additional Numerical Results}
\subsection{Simulation Results for DP-BIC}

We evaluate the finite-sample performance of the proposed DP-BIC procedure in Algorithm 2, focusing on its selection properties and estimation accuracy for debiased inference. The simulation settings are identical to those described in Section 6.1.1, using the identity covariance matrix. The privacy parameter for each coordinate is $\varepsilon=2$ and $\delta=1/n^{1.1}$. To assess parameter selection, we compute the proportion of correctly identified true features, the average number of falsely selected zero coefficients, and the mean squared error $\|\hat{\boldsymbol{\beta}} - \boldsymbol{\beta}\|_2^2$ for the selected model, following the evaluation criteria used by \cite{fan2013tuning}. We examine the performance of DP-BIC under varying sample sizes and signal strengths, which are two critical factors in model selection.

\begin{figure}[ht]
	\centering
	\includegraphics[width=0.7\linewidth]{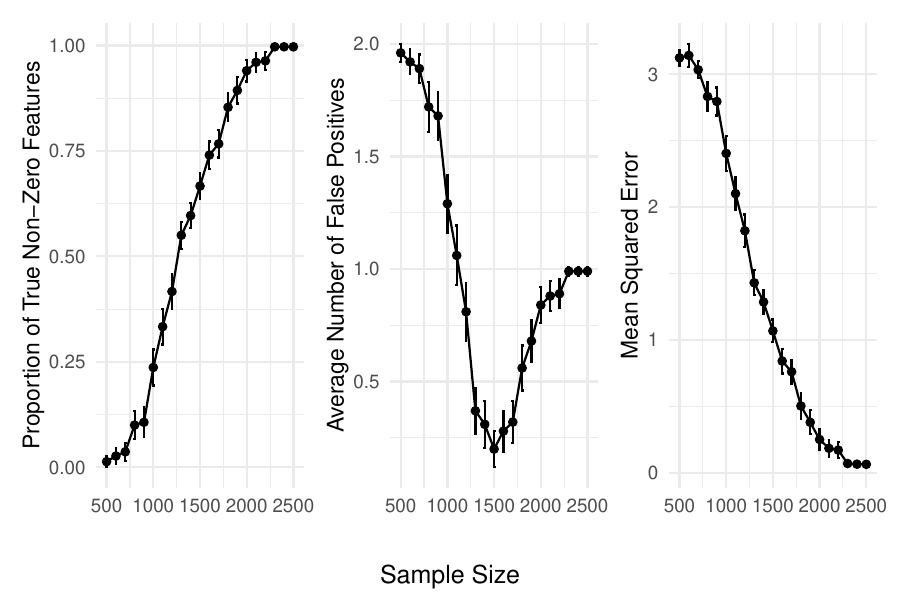}
	\caption{Simulation results evaluating the performance of the DP-BIC procedure for varying sample sizes. The three panels display: (left) the proportion of true non-zero features correctly identified, (middle) the average number of false positives, and (right) the mean squared error. Error bars represent $\pm 2$ standard errors across $100$ simulation replications.}
	\label{fig:sel-n}
\end{figure}

Figure \ref{fig:sel-n} presents simulation results evaluating the performance of the proposed DP-BIC procedure in terms of feature selection and estimation accuracy across varying sample sizes. The left panel shows the proportion of true non-zero features correctly identified, demonstrating that selection accuracy improves as the sample size increases. When the sample size is small ($n=500$), the proposed procedure fails to recover the true coefficients due to the large sample requirement imposed by DP. In contrast, when the sample size is large ($n=2000$), the procedure successfully identifies all the true coefficients. The middle panel displays the average number of false positives. Note that the true non-zero features are $S_0={1, 2, 3}$, and the proposed DP-BIC selects $2^K$ features. Thus, the proposed method tends to select $K=1$ when the sample size is smaller than 1500, and tends to select $K=2$ when the sample size is larger than $1500$. As a result, we see the average number of false positives decrease first, and increase to 1 as the sample size continues to increase, because $2^2 - 3 = 1$. The right panel illustrates the mean squared error between the estimated and true coefficients, which consistently decreases with increasing sample size, reflecting improved estimation accuracy. These results collectively confirm that the DP-BIC procedure becomes more reliable and accurate with larger sample sizes, aligning with the findings of \cite{fan2013tuning} for the non-private BIC.

\begin{figure}[ht]
	\centering
	\includegraphics[width=0.7\linewidth]{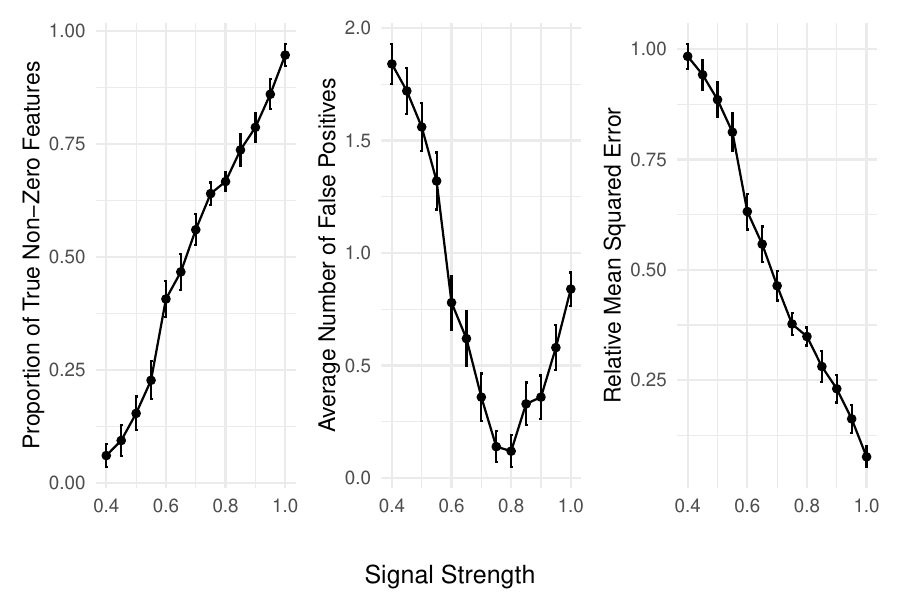}
	\caption{Simulation results evaluating the performance of the DP-BIC procedure under varying signal strengths. The three panels display: (left) the proportion of true non-zero features correctly identified; (middle) the average number of false positives; and (right) the mean squared error. Error bars represent $\pm 2$ standard errors computed over 100 simulation replications.}
	\label{fig:sel-s}
\end{figure}

Figure \ref{fig:sel-s} presents simulation results evaluating the performance of the proposed DP-BIC procedure across varying signal strengths. Instead of evaluating the mean squared error, we assess the relative mean squared error defined as $|\hat{\boldsymbol{\beta}} - \boldsymbol{\beta}|_2^2/|\boldsymbol{\beta}|_2^2$. The trends for the proportion of true non-zero features correctly identified and the average number of false positives are similar to those observed in Figure \ref{fig:sel-n}. As signal strength increases, selection becomes more accurate, consistent with the findings of \cite{fan2013tuning} for the non-private BIC. The right panel illustrates the relative mean squared error between the estimated and true coefficients, which consistently decreases with increasing signal strength. Given a fixed privacy budget, higher signal strength improves both selection accuracy and relative estimation efficiency.

\subsection{Additional Real Data Example: Parkinson's Telemonitoring}

For real data applications, we demonstrate the performance of the proposed differentially private algorithms in analyzing the Parkinson's Disease Progression data \citep{tsanas2009accurate}. In clinical diagnosis, assessing the progression of Parkinson's disease (PD) symptoms typically relies on the Unified Parkinson's Disease Rating Scale (UPDRS), which necessitates the patient's physical presence at the clinic and time-consuming physical evaluations conducted by trained medical professionals. Thus, monitoring symptoms is associated with high costs and logistical challenges for patients and clinical staff. This dataset aims to track UPDRS by noninvasive speech tests. However, it is crucial to protect the privacy of each patient's data, as any unauthorized disclosure could lead to potential harm or trouble for the participants. By ensuring privacy protection, individuals are more likely to contribute their personal data, facilitating advancements in Parkinson's disease research.

The data collection process involved the utilization of the Intel AHTD, a telemonitoring system designed for remote, internet-enabled measurement of various motor impairment symptoms associated with Parkinson's disease (PD). The research was overseen by six U.S. medical centers, namely the Georgia Institute of Technology (seven subjects), the National Institutes of Health (ten subjects), Oregon Health and Science University (fourteen subjects), Rush University Medical Center (eleven subjects), Southern Illinois University (six subjects), and the University of California, Los Angeles (four subjects). A total of 52 individuals diagnosed with idiopathic PD were recruited. Following an initial screening process to eliminate flawed recordings (such as instances of patient coughing), a total of 5923 sustained phonations were subjected to analysis. In total, 16 dysphonia measures were applied to the 5923 sustained phonations. \cite{tsanas2009accurate} proposed using a linear model and did not consider privacy issues.

\subsubsection{Debiased Inference}

To evaluate the performance of our proposed differentially private inference algorithm \ref{algo:inference} in high-dimensional settings, we add $5,000$ random features generated independently and identically from the standard normal distribution. Therefore, the dataset comprises a sample size of $n=5,923$ with covariates having a dimension of $p=5,016$, where the first $16$ of these covariates represent real features.

We consider the following three methods: the oracle method, the proposed differentially private algorithm, and the non-private debiased Lasso \citep{van2014asymptotically}. The oracle method uses only the 16 real features, while the proposed algorithm and the debiased Lasso utilize all 5,016 features. The privacy parameters are $\varepsilon=0.5$ and $\delta=1/n^{1.1}$. Figure \ref{fig:4} displays the confidence intervals for the 16 real features obtained from the three methods. Overall, the private confidence intervals consistently cover the estimates obtained from the oracle method and the debiased Lasso. The private confidence intervals exhibit substantial overlap with the confidence intervals from both the oracle method and the debiased Lasso. However, due to the privacy costs, the width of the proposed confidence intervals is slightly larger than the confidence intervals from the debiased Lasso.

\begin{figure}[ht]
	\centering
	\includegraphics[width=\linewidth]{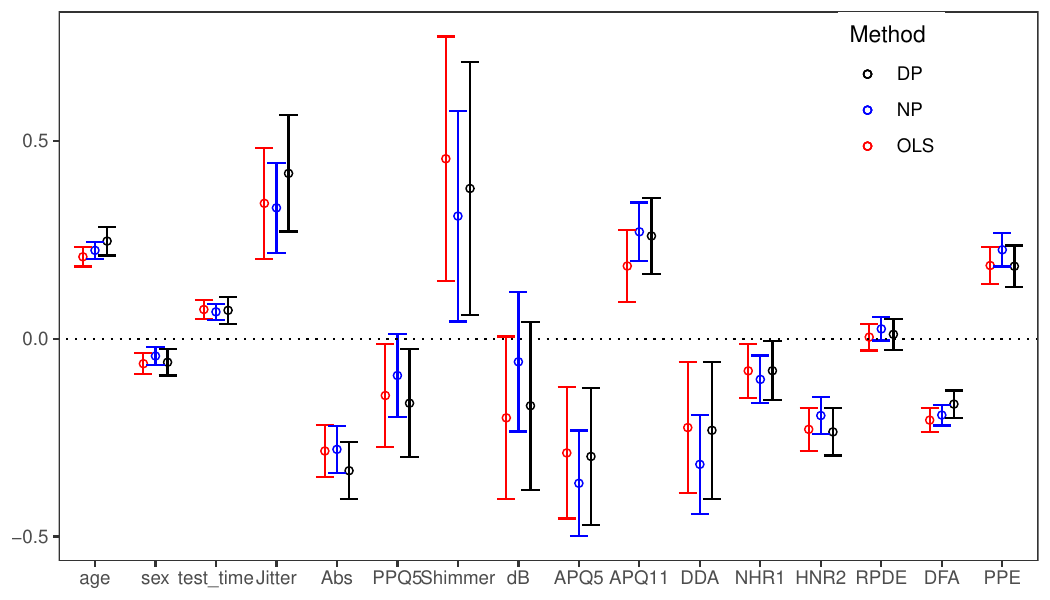}
	\caption{The $95\%$ confidence interval of OLS, nonprivate debiased Lasso (NP), and proposed Algorithm \ref{algo:inference} with finite sample correction (DP).}
	\label{fig:4}
\end{figure}

\subsubsection{FDR control}


Next, we evaluate the performance of the private FDR control algorithm proposed in Section \ref{sec: fdr} on the Parkinson's Disease Progression data. We add $100$ random features generated independently and identically from $N(0, 1)$. The proposed algorithm is compared with the non-private data splitting algorithm by \cite{dai2022false} and the knockoff by \cite{barber2015controlling}. We use equal-sized data splitting. The results are reported in Figure \ref{fig:5}, where the target FDR is set to $0.1$ and $0.3$, respectively. The proposed method exhibits a notable number of discoveries within the real features while registering only a minimal number of false discoveries among the random features, compared to knockoff and non-private data-splitting methods. 

\begin{figure}[ht]
	\centering
	\includegraphics[width=\linewidth]{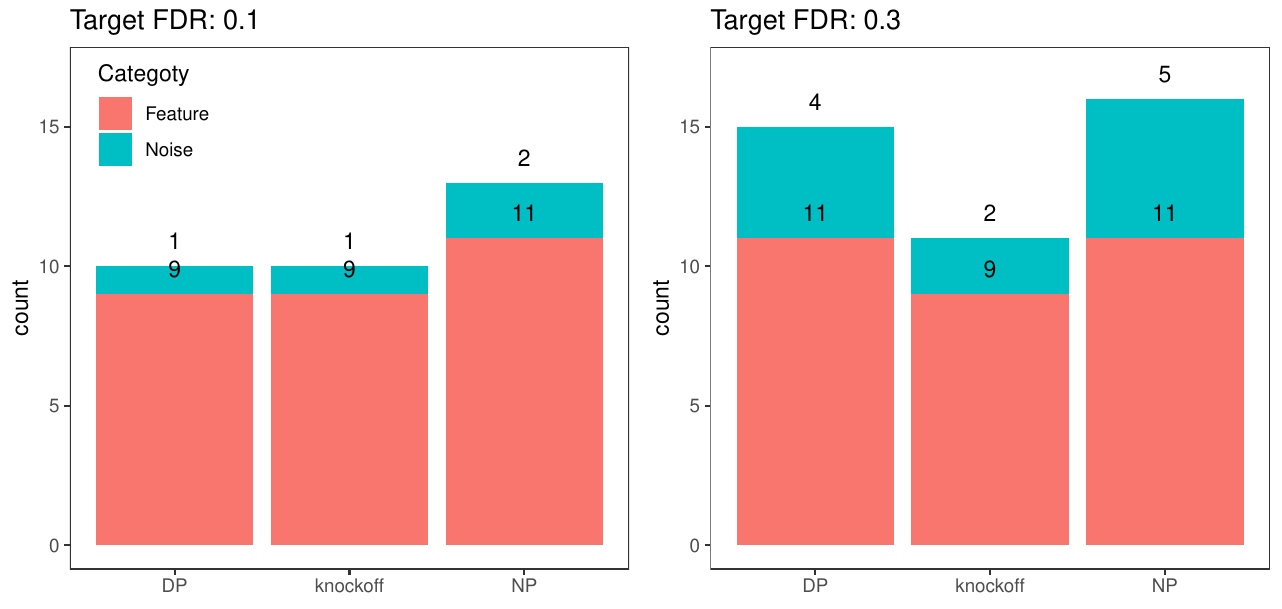}
	\caption{Numbers of the discovers for Algorithm \ref{alg:fdr} (DP), the non-private version (NP), and the knockoff at target FDR=$0.1$ and $0.3$.}
	\label{fig:5}
\end{figure}

\begin{table}[ht]
	\centering
	\begin{tabular}{l|rrrrr}
		\hline
		feature & age & sex & test\_time & Jitter & Abs  \\ 
		\hline
		knockoff & \checkmark &  & \checkmark & \checkmark & \checkmark \\ 
		Non-Private & \checkmark & \checkmark & \checkmark & \checkmark& \checkmark \\ 
		DP-FDR& \checkmark &  &  & \checkmark & \checkmark \\ 
		\hline
		\hline
		feature  & APQ5 & APQ11 & DDA & NHR1 & HNR2 \\ 
		\hline
		knockoff & \checkmark & \checkmark &  &  & \checkmark \\ 
		Non-Private &  & \checkmark & \checkmark & \checkmark  & \checkmark \\ 
		DP-FDR&  &  & \checkmark &  & \checkmark  \\ 
		\hline
		\hline
		feature  &PPQ5 & Shimmer & dB& RPDE & DFA  \\ 
		\hline
		knockoff & &  &   && \checkmark    \\
		Non-Private &  &  &    & &\checkmark\\ 
		DP-FDR  &  & \checkmark&& &\checkmark \\
		\hline
		\hline
		feature & PPE  \\ 
		\hline
		knockoff&\checkmark\\
		Non-Private &\checkmark \\ 
		DP-FDR  &\checkmark \\
		\hline
		\hline
	\end{tabular}
	\caption{Selected Real Feature for Algorithm \ref{alg:fdr} (DP), the non-private version (NP), and the knockoff at target FDR=$0.1$.}
	\label{tab:3a}
\end{table}

We further report the selected features at the target FDR level of $q=0.1$ in Table \ref{tab:3a}. There is substantial overlap among the three methods, with several features being consistently selected. For instance, age, Jitter, Abs, HNR2, DFA, and PPE are all chosen by all three methods. For Parkinson's disease (PD), which is the second most prevalent neurodegenerative disorder among the elderly, age is the most crucial risk factor. Numerous medical studies underscore the pivotal role of age as the single most significant factor associated with PD, as documented by \cite{elbaz2002risk}.
Furthermore, Jitter and Abs are commonly employed to characterize cycle-to-cycle variability in fundamental frequency, while HNR (Harmonics-to-Noise Ratio) is an essential feature in speech processing techniques. Detrended Fluctuation Analysis (DFA) and Pitch Period Entropy (PPE) represent two recently proposed speech signal processing methods, both of which exhibit a strong correlation with PD-dysphonia, as highlighted in \cite{little2008suitability}. In addition to these commonly selected features, the proposed method also identifies shimmer and DDA, which are frequently used to describe cycle-to-cycle variability in amplitude. Shimmer and DDA are also endorsed as relevant features in clinical studies by \cite{tsanas2009accurate}. The features identified by our proposed method receive substantial clinical support, with a privacy guarantee for the individual patients.

\section{Useful Tools in Differential Privacy }


\begin{lemma}[\cite{dwork2014algorithmic}]
	\label{lemma:laplace_gaussian_mechanism}
	$\ $
	\begin{enumerate}
		\item (Laplace mechanism): For a deterministic algorithm $\mathcal{T}(\cdot)$ with $l_1$ sensitivity $\Delta_1(\mathcal{T})$, the randomized algorithm $\mathcal{M}(\cdot):=\mathcal{T}(\cdot)+\boldsymbol{\xi}$ achieves $(\varepsilon,0)$-differential privacy, where $\boldsymbol{\xi}=(\xi_1,\dots,\xi_m)^{\top}$ follows i.i.d. Laplace distribution with scale parameter $\Delta_1(\mathcal{T})/\varepsilon$.
		\item (Gaussian mechanism): For a deterministic algorithm $\mathcal{T}(\cdot)$ with $l_2$ sensitivity $\Delta_2(\mathcal{T})$, the randomized algorithm $\mathcal{M}(\cdot):=\mathcal{T}(\cdot)+\boldsymbol{\xi}$ achieves $(\varepsilon,\delta)$-differential privacy, where $\boldsymbol{\xi}=(\xi_1,\dots,\xi_m)^{\top}$ follows i.i.d. Gaussian distribution with mean $0$ and standard deviation $\sqrt{2\log(1.25/\delta)}\Delta_2(\mathcal{T})/\varepsilon$.
	\end{enumerate}
\end{lemma}


\begin{lemma}
	\label{lem:post_combin_dp}
	Differentially private algorithms have the following properties \citep{dwork2006calibrating}:
	\begin{enumerate}
		\item Post-processing: Let $\mathcal{M}(\cdot)$ be an $(\varepsilon,\delta)$-DP algorithm and $f(\cdot)$ be a deterministic function that maps $\mathcal{M}(D)$ to real Euclidean space, then $f(\mathcal{M}(D))$ is also an $(\varepsilon,\delta)$-DP algorithm.
		\item Composition: Let $\mathcal{M}_1(\cdot)$ be $(\varepsilon_1,\delta_1)$-differentially private and $\mathcal{M}_2(\cdot)$ be $(\varepsilon_2,\delta_2)$-differentially private, then $\mathcal{M}_1\circ\mathcal{M}_2(\cdot)$ is $(\varepsilon_1+\varepsilon_2,\delta_1+\delta_2)$-differentially private.
		\item Advanced Composition: Let $\mathcal{M}(\cdot)$ be $(\varepsilon,0)$-differentially private and $0<\delta^{\prime}<1$, then $k$-fold adaptive composition of $\mathcal{M}(\cdot)$ is $(\varepsilon^{\prime},\delta^{\prime})$-differentially private for $\varepsilon^{\prime}=k\varepsilon(e^{\varepsilon}-1)+\varepsilon\sqrt{2k\log(1/\delta^{\prime})}$.
	\end{enumerate}
\end{lemma}
\vskip 0.2in
\bibliography{sample}

@article{xia2025statistical,
  title={Statistical Inference for Differentially Private Stochastic Gradient Descent},
  author={Xia, Xintao and Zhang, Linjun and Cai, Zhanrui},
  journal={arXiv preprint arXiv:2507.20560},
  year={2025}
}

@article{wang2022finite,
	title={Finite-and large-sample inference for model and coefficients in high-dimensional linear regression with repro samples},
	author={Wang, Peng and Xie, Min-Ge and Zhang, Linjun},
	journal={arXiv preprint arXiv:2209.09299},
	year={2022}
}

@article{cai2025knockoffs,
	title={Knockoffs Inference under Privacy Constraints},
	author={Cai, Zhanrui and Fan, Yingying and Gao, Lan},
	journal={arXiv preprint arXiv:2506.09690},
	year={2025}
}

@article{xia2025differentially,
	title={Differentially Private Sliced Inverse Regression: Minimax Optimality and Algorithm},
	author={Xia, Xintao and Zhang, Linjun and Cai, Zhanrui},
	journal={Journal of the American Statistical Association},
	pages={1--22},
	year={2025},
	publisher={Taylor \& Francis}
}

@inproceedings{dwork2006calibrating,
	title={Calibrating noise to sensitivity in private data analysis},
	author={Dwork, Cynthia and McSherry, Frank and Nissim, Kobbi and Smith, Adam},
	booktitle={TCC 2006},
	pages={265--284},
	year={2006},
	organization={Springer}
}

@article{dwork2014algorithmic,
	title={The algorithmic foundations of differential privacy},
	author={Dwork, Cynthia and Roth, Aaron},
	journal={Foundations and Trends{\textregistered} in Theoretical Computer Science},
	volume={9},
	number={3--4},
	pages={211--407},
	year={2014},
	publisher={Now Publishers, Inc.}
}

@Article{barber2015controlling,
	title={Controlling the false discovery rate via knockoffs},
	author={Barber, Rina Foygel and Cand{\`e}s, Emmanuel J},
	journal={The Annals of Statistics},
	volume={43},
	number={5},
	pages={2055--2085},
	year={2015},
	publisher={Institute of Mathematical Statistics}
}

@Article{barber2019knockoff,
	title={A knockoff filter for high-dimensional selective inference},
	author={Barber, Rina Foygel and Cand{\`e}s, Emmanuel J},
	journal={The Annals of Statistics},
	volume={47},
	number={5},
	pages={2504--2537},
	year={2019},
	publisher={Institute of Mathematical Statistics}
}

@article{candes2018panning,
	title={Panning for gold:‘model-X’knockoffs for high dimensional controlled variable selection},
	author={Candes, Emmanuel and Fan, Yingying and Janson, Lucas and Lv, Jinchi},
	journal={Journal of the Royal Statistical Society Series B: Statistical Methodology},
	volume={80},
	number={3},
	pages={551--577},
	year={2018},
	publisher={Oxford University Press}
}

@Article{dai2022false,
	title={False discovery rate control via data splitting},
	author={Dai, Chenguang and Lin, Buyu and Xing, Xin and Liu, Jun S},
	journal={Journal of the American Statistical Association},
	pages={1--18},
	year={2022},
	publisher={Taylor \& Francis}
}

@article{dai2023scale,
	title={A scale-free approach for false discovery rate control in generalized linear models},
	author={Dai, Chenguang and Lin, Buyu and Xing, Xin and Liu, Jun S},
	journal={Journal of the American Statistical Association},
	pages={1--15},
	year={2023},
	publisher={Taylor \& Francis}
}

@article{avella2021differentially,
	title={Differentially private inference via noisy optimization},
	author={Avella-Medina, Marco and Bradshaw, Casey and Loh, Po-Ling},
	journal={Annals of Statistics},
	year={2023}
}

@article{cai2017confidence,
	title={Confidence intervals for high-dimensional linear regression: Minimax rates and adaptivity},
	author={Cai, T Tony and Guo, Zijian},
	journal={Ann. Stat.},
	volume={45},
	number={2},
	pages={615--646},
	year={2017},
	publisher={Institute of Mathematical Statistics}
}

@inproceedings{dwork2014analyze,
	title={Analyze gauss: optimal bounds for privacy-preserving principal component analysis},
	author={Dwork, Cynthia and Talwar, Kunal and Thakurta, Abhradeep and Zhang, Li},
	booktitle={STOC 2014},
	pages={11--20},
	year={2014},
	organization={ACM}
}

@inproceedings{talwar2015nearly,
	title={Nearly optimal private {L}asso},
	author={Talwar, Kunal and Thakurta, Abhradeep Guha and Zhang, Li},
	booktitle={NeurIPS 2015},
	pages={3025--3033},
	year={2015}
}

@article{vershynin2010introduction,
	title={Introduction to the non-asymptotic analysis of random matrices},
	author={Vershynin, Roman},
	journal={arXiv preprint arXiv:1011.3027},
	year={2010}
}

@article{van2014asymptotically,
	title={On Asymptotically Optimal Confidence Regions And Tests For High-Dimensional Models},
	author={van de Geer, Sara and B{\"u}hlmann, Peter and Ritov, Ya'acov and Dezeure, Ruben},
	journal={The Annals of Statistics},
	volume={42},
	number={3},
	pages={1166--1202},
	year={2014},
	publisher={Institute of Mathematical Statistics}
}

@article{zhang2014confidence,
	title={Confidence intervals for low dimensional parameters in high dimensional linear models},
	author={Zhang, Cun-Hui and Zhang, Stephanie S},
	journal={Journal of the Royal Statistical Society: Series B (Statistical Methodology)},
	volume={76},
	number={1},
	pages={217--242},
	year={2014},
	publisher={Wiley Online Library}
}

@article{javanmard2014confidence,
	title={Confidence intervals and hypothesis testing for high-dimensional regression.},
	author={Javanmard, Adel and Montanari, Andrea},
	journal={Journal of Machine Learning Research},
	volume={15},
	number={1},
	pages={2869--2909},
	year={2014}
}

@article{cai2019cost,
	title={The cost of privacy: Optimal rates of convergence for parameter estimation with differential privacy},
	author={Cai, T Tony and Wang, Yichen and Zhang, Linjun},
	journal={The Annals of Statistics},
	volume={49},
	number={5},
	pages={2825--2850},
	year={2021},
	publisher={Institute of Mathematical Statistics}
}

@article{fan2008sure,
	title={Sure independence screening for ultrahigh dimensional feature space},
	author={Fan, Jianqing and Lv, Jinchi},
	journal={Journal of the Royal Statistical Society Series B: Statistical Methodology},
	volume={70},
	number={5},
	pages={849--911},
	year={2008},
	publisher={Oxford University Press}
}

@article{dwork2018differentially,
	title={Differentially private false discovery rate control}, 
	volume={11},
	number={2}, 
	journal={Journal of Privacy and Confidentiality}, author={Dwork, Cynthia and Su, Weijie and Zhang, Li}, 
	year={2021}, 
	month={Sep.}
}

@article{tsanas2009accurate,
	title={Accurate telemonitoring of Parkinson’s disease progression by non-invasive speech tests},
	author={Tsanas, Athanasios and Little, Max and McSharry, Patrick and Ramig, Lorraine},
	journal={Nature Precedings},
	pages={1--1},
	year={2009},
	publisher={Nature Publishing Group UK London}
}

@article{xia2023fdr,
	author  = {Xintao Xia and Zhanrui Cai},
	title   = {Adaptive False Discovery Rate Control with Privacy Guarantee},
	journal = {Journal of Machine Learning Research},
	year    = {2023},
	volume  = {24},
	number  = {252},
	pages   = {1--35},
}

@article{elbaz2002risk,
	title={Risk tables for parkinsonism and Parkinson's disease},
	author={Elbaz, Alexis and Bower, James H and Maraganore, Demetrius M and McDonnell, Shannon K and Peterson, Brett J and Ahlskog, J Eric and Schaid, Daniel J and Rocca, Walter A},
	journal={Journal of Clinical Epidemiology},
	volume={55},
	number={1},
	pages={25--31},
	year={2002},
	publisher={Elsevier}
}

@article{little2008suitability,
	title={Suitability of dysphonia measurements for telemonitoring of Parkinson’s disease},
	author={Little, Max and McSharry, Patrick and Hunter, Eric and Spielman, Jennifer and Ramig, Lorraine},
	journal={Nature Precedings},
	pages={1--1},
	year={2008},
	publisher={Nature Publishing Group}
}

@article{tan2020sparse,
	title={Sparse {SIR}: Optimal rates and adaptive estimation},
	author={Tan, Kai and Shi, Lei and Yu, Zhou},
	journal={The Annals of Statistics},
	volume={48},
	number={1},
	pages={64--85},
	year={2020},
	publisher={Institute of Mathematical Statistics}
}

@article{azriel2015empirical,
	title={The empirical distribution of a large number of correlated normal variables},
	author={Azriel, David and Schwartzman, Armin},
	journal={Journal of the American Statistical Association},
	volume={110},
	number={511},
	pages={1217--1228},
	year={2015},
	publisher={Taylor \& Francis}
}

@inproceedings{thakurta2013differentially,
	title={Differentially private feature selection via stability arguments, and the robustness of the {L}asso},
	author={Thakurta, Abhradeep Guha and Smith, Adam},
	booktitle={Conference on Learning Theory},
	pages={819--850},
	year={2013},
	organization={PMLR}
}

@misc{chernozhukov2018double,
	title={Double/debiased machine learning for treatment and structural parameters},
	author={Chernozhukov, Victor and Chetverikov, Denis and Demirer, Mert and Duflo, Esther and Hansen, Christian and Newey, Whitney and Robins, James},
	year={2018},
	publisher={Oxford University Press Oxford, UK}
}

@article{nusser1997national,
	title={The National Resources Inventory: a long-term multi-resource monitoring programme},
	author={Nusser, Sarah M and Goebel, J Jeffery},
	journal={Environmental and Ecological Statistics},
	volume={4},
	number={3},
	pages={181--204},
	year={1997},
	publisher={Springer}
}

@article{kim2005rapid,
	title={Rapid assessment of soil erosion in the Rio Lempa Basin, Central America, using the universal soil loss equation and geographic information systems},
	author={Kim, John B and Saunders, Peter and Finn, John T},
	journal={Environmental Management},
	volume={36},
	pages={872--885},
	year={2005},
	publisher={Springer}
}

@article{panagos2015tackling,
	title={Tackling soil loss across Europe},
	author={Panagos, Panos and Borrelli, Pasquale and Robinson, David A},
	journal={Nature},
	volume={526},
	number={7572},
	pages={195--195},
	year={2015},
	publisher={Nature Publishing Group UK London}
}

@article{alewell2019using,
	title={Using the USLE: Chances, challenges and limitations of soil erosion modelling},
	author={Alewell, Christine and Borrelli, Pasquale and Meusburger, Katrin and Panagos, Panos},
	journal={International soil and water conservation research},
	volume={7},
	number={3},
	pages={203--225},
	year={2019},
	publisher={Elsevier}
}

@article{nearing2004expected,
	title={Expected climate change impacts on soil erosion rates: a review},
	author={Nearing, MA and Pruski, FF and O'neal, MR},
	journal={Journal of soil and water conservation},
	volume={59},
	number={1},
	pages={43--50},
	year={2004},
	publisher={Soil and Water Conservation Society}
}

@article{chen2007large,
	title={Large sample sieve estimation of semi-nonparametric models},
	author={Chen, Xiaohong},
	journal={Handbook of econometrics},
	volume={6},
	pages={5549--5632},
	year={2007},
	publisher={Elsevier}
}

@article{belloni2019valid,
	title={Valid post-selection inference in high-dimensional approximately sparse quantile regression models},
	author={Belloni, Alexandre and Chernozhukov, Victor and Kato, Kengo},
	journal={Journal of the American Statistical Association},
	volume={114},
	number={526},
	pages={749--758},
	year={2019},
	publisher={Taylor \& Francis}
}

@InProceedings{rudelson2012reconstruction,
  title = 	 {Reconstruction from Anisotropic Random Measurements},
  author = 	 {Rudelson, Mark and Zhou, Shuheng},
  booktitle = 	 {Proceedings of the 25th Annual Conference on Learning Theory},
  pages = 	 {10.1--10.24},
  year = 	 {2012},
  volume = 	 {23},
  series = 	 {Proceedings of Machine Learning Research},
}

@article{fan2013tuning,
	title={Tuning parameter selection in high dimensional penalized likelihood},
	author={Fan, Yingying and Tang, Cheng Yong},
	journal={Journal of the Royal Statistical Society Series B: Statistical Methodology},
	volume={75},
	number={3},
	pages={531--552},
	year={2013},
	publisher={Oxford University Press}
}

@article{wang2009shrinkage,
  title={Shrinkage tuning parameter selection with a diverging number of parameters},
  author={Wang, Hansheng and Li, Bo and Leng, Chenlei},
  journal={Journal of the Royal Statistical Society Series B: Statistical Methodology},
  volume={71},
  number={3},
  pages={671--683},
  year={2009},
  publisher={Oxford University Press}
}

@article{chetverikov2021cross,
  title={On cross-validated lasso in high dimensions},
  author={Chetverikov, Denis and Liao, Zhipeng and Chernozhukov, Victor},
  journal={The Annals of Statistics},
  volume={49},
  number={3},
  pages={1300--1317},
  year={2021},
  publisher={JSTOR}
}

\end{document}